\documentclass[lettersize,journal]{IEEEtran}
\usepackage{amsmath,amssymb,amsfonts,amsthm,mathrsfs}
\usepackage{multicol}
\usepackage{multirow}
\usepackage{algorithmic}
\usepackage{algorithm}
\usepackage{array}
\usepackage[caption=false,font=normalsize,labelfont=sf,textfont=sf]{subfig}
\usepackage{textcomp}
\usepackage{stfloats}
\usepackage{url}
\usepackage{verbatim}
\usepackage{graphicx}
\usepackage{epsfig}
\usepackage{balance}
\usepackage{cite}
\usepackage{bm}% bold math
\usepackage{empheq}
\usepackage[colorlinks=true,linkcolor=magenta,citecolor=blue,urlcolor=cyan,filecolor=red,draft]{hyperref}
\newtheorem{theorem}{Theorem}

\newtheorem{lemma}{Lemma}
\newtheorem{remark}{Remark}

\newtheorem{proposition}{Proposition}

\newtheorem{mydef}{Definition}

\newcommand{\bR}{\mathbb{R}}
\newcommand{\bZ}{\mathbb{Z}}
\newcommand{\bN}{\mathbb{N}}
\newcommand{\bC}{\mathbb{C}}

\newcommand{\bx}{{\bm{x}}}

\newcommand{\bU}{U^{T_s}}
\newcommand{\hF}{\mathcal{F}}     
\newcommand{\hD}{\mathcal{D}}

\newcommand{\hL}{\mathcal{L}}

\newcommand{\bIm}{\mathrm{Im}}

\newcommand{\Log}{\mathrm{Log}}
\newcommand{\sspan}{\mathrm{span}}

\newcommand{\hG}{\mathscr{G}}

\newcommand{\bl}[1]{{\color{black}{#1}}}

\newcommand{\Real}{\mathbb{R}}
\newcommand{\abs}[1]{\lvert#1\rvert}
\newcommand{\norm}[1]{\lVert#1\rVert}

\providecommand{\DIFaddbeginFL}{} %DIF PREAMBLE
\providecommand{\DIFaddendFL}{} %DIF PREAMBLE
\providecommand{\DIFdelbeginFL}{} %DIF PREAMBLE
\providecommand{\DIFdelendFL}{} %DIF PREAMBLE
\LetLtxMacro{\DIFOincludegraphics}{\includegraphics} %DIF PREAMBLE
\newcommand{\DIFaddincludegraphics}[2][]{{\color{blue}\fbox{\DIFOincludegraphics[#1]{#2}}}} %DIF PREAMBLE
\LetLtxMacro{\DIFOaddbeginFL}{\DIFaddbeginFL} %DIF PREAMBLE
\LetLtxMacro{\DIFOaddendFL}{\DIFaddendFL} %DIF PREAMBLE
\LetLtxMacro{\DIFOdelbeginFL}{\DIFdelbeginFL} %DIF PREAMBLE
\LetLtxMacro{\DIFOdelendFL}{\DIFdelendFL} %DIF PREAMBLE
\DeclareRobustCommand{\DIFaddbeginFL}{\DIFOaddbeginFL \let\includegraphics\DIFaddincludegraphics} %DIF PREAMBLE
\begin{document}
	
	\title{A Generalized Nyquist-Shannon Sampling Theorem Using the Koopman Operator}
	
	\author{Zhexuan Zeng, Jun Liu, and Ye Yuan
		\thanks{This work was supported by the National Natural Science Foundation of China under Grant 92167201 and China Scholarship Council program (Project ID: 202306160107).}
		\thanks{
			Zhexuan Zeng and Ye Yuan are with School of Artificial Intelligence and Automation, Huazhong University of Science and Technology, Wuhan, China.
			
			Jun Liu is with the Department of Applied Mathematics, Faculty of Mathematics, University of Waterloo, Waterloo, Ontario N2L 3G1, Canada. }
		\thanks{{$^*$}For correspondence, \href{mailto:
				yye@hust.edu.cn}{\tt yye@hust.edu.cn}.}}

	\markboth{Journal of \LaTeX\ Class Files,~Vol.~X, No.~X, Jan~2023}%
	{Shell \MakeLowercase{\textit{et al.}}: A Sample Article Using IEEEtran.cls for IEEE Journals}
	
	
	\maketitle
	
	\begin{abstract}
		%\bl{In the field of signal processing, the sampling theorem plays a fundamental role for the signal reconstruction as it brides the gap between analog and digital signal.} In this work, a generalized sampling theorem-- which builds on the Koopman operator-- is proposed for signals in generator-bounded space (Theorem \ref{thKs}). It naturally extends the Nyquist-Shannon sampling theorem that, 1) for band-limited signals, the lower bounds of sampling frequency \bl{and reconstruction formula} given by these two theorems are exactly the same; 2) the Koopman operator-based sampling theorem can also provide finite bound of sampling frequency \bl{and reconstruction formula} for certain types of non-band-limited signals, which cannot be addressed by Nyquist-Shannon sampling theorem. These types of non-band-limited signals include but not limited to, for example, inverse Laplace transform with limited imaginary interval of integration, and linear combinations of complex exponential functions. Moreover, the Koopman operator-based reconstruction method is provided with theoretical result of convergence. \bl{This method shows robustness against low sampling frequency, which illustrates the sampling theorem numerically on several examples.}
		
		In the field of signal processing, the sampling theorem plays a fundamental role for signal reconstruction as it bridges the gap between analog and digital signals.
		Following the celebrated Nyquist-Shannon sampling theorem, \bl{generalizing the sampling theorem to non-band-limited signals} remains a major challenge. In this work, a generalized sampling theorem, which builds upon the Koopman operator, is proposed for signals in a generator-bounded space. It naturally extends the Nyquist-Shannon sampling theorem in that: 1) for band-limited signals, the lower bounds of the sampling frequency and the reconstruction formulas given by these two theorems are exactly the same; 2) the Koopman operator-based sampling theorem can also provide a finite bound of the sampling frequency and a reconstruction formula for certain types of non-band-limited signals, which cannot be addressed by Nyquist-Shannon sampling theorem. These non-band-limited signals include\bl{, but are} not limited to, the inverse Laplace transform with limit imaginary interval of integration, and linear combinations of complex exponential functions. Furthermore, the Koopman operator-based reconstruction method is supported by theoretical results on its convergence. This method is illustrated numerically through several examples, demonstrating its robustness against low sampling frequencies.
	\end{abstract}

	
	\section{INTRODUCTION}
	\IEEEPARstart{O}{ne} of the central problems in signal processing is to reconstruct the continuous-time (CT) signal from discrete-time (DT) samples. Sampling theorem, which allows faithful representation of CT signal by its DT samples without aliasing, bridges the gap between analog and digital signals. Its wide range of applications includes radio engineering, crystallography, optics, and other scientific areas.
	
	The Nyquist-Shannon sampling theorem \cite{shannon1949communication} is a landmark in both mathematical and engineering literature, which gives a mechanism to convert band-limited signals into a sequence of numbers. Over the years, there are various generalized sampling theorems focusing on the extensions of signal classes, which can be primarily divided into two perspectives. %, i.e., the generalizations of integral transforms and the interpolation functions. 
	The first extension perspective starts from the integral form of a signal, which extends to general limit integral transforms besides the Fourier transform \cite{kramer1959sampling,walter1988sampling} or signals possessing specific integral properties \cite{cambanis1976zakai}. Most extensions within this perspective are formulated for band-limited functions \cite{jerri1977shannon}. The second one is inspired by wavelet theory, generalizing the basis function in the reconstruction formula from the classical sinc-function to more general functions \cite{unser2000sampling,zayed1993advances,eldar2015sampling}, which leads to not necessarily band-limited signals. It includes generating functions, such as splines or wavelets, of shift-invariant spaces  \cite{de1994approximation,deboor1994structure}. \bl{There are also works focusing on developing precise sampling schemes for parametric signals of finite rate of innovation, which allow the signals to be reconstructed from samples with a finite sampling frequency} \cite{vetterli2002sampling}. 
	%Most generalizations of the sampling theorem focusing on the extensions to other classes of signals \cite{jerri1977shannon}, such as signals of generalized kernels \cite{kramer1959sampling}, multi-variable signals \cite{shi2010multichannel,xu2017multichannel}, and band-pass signals \cite{kohlenberg1953exact}. With the development of wideband communication and radio-frequency technology, the extension to non-band-limited signals, which allows sampling at finite frequency of signals with unlimited frequency support, has attracted increasingly attention \cite{mishali2011sub}. \bl{For example, Zakai extends the sampling theorem to functions with polynomial growth that do not have a Fourier representation \cite{cambanis1976zakai, palmieri1986sampling}. Vetterli et al. extends to signals with finite rate of innovation \cite{vetterli2002sampling}.%, and with other kernel functions \cite{shi2014sampling}.
		However, most generalizations primarily leverage known signal structures for signal reconstruction, which does not breach the Nyquist-Shannon theorem \cite{mishali2011sub}.

		Based on the conclusion of the Nyquist-Shannon sampling theorem, the lower bound of sampling frequency is determined by the oscillation frequency in the (Fourier) frequency domain for band-limited signals $g(t)$. Despite the fact that most signals in the real world also oscillate at finite frequencies (e.g., $e^{-t}\cos t$) \cite{yuan2019data}, they cannot be analyzed by the Nyquist-Shannon sampling theorem because they have non-zero amplitude growth and are not band-limited. It inspires us that the spectrum of the Koopman operator \cite{Koopman1931hamiltonian}, which we will refer to as the Koopman spectrum for \bl{brevity, has} a significant connection with the features of oscillation and amplitude growth of functions \cite{mezic2005spectral}. Furthermore, a similar sampling issue is investigated for exact identification of dynamical systems \cite{Yue2020a, zeng2022sampling, zeng2024sampling}, whose results also depend on the Koopman spectrum. Therefore, it leads us to reexamine the foundations of sampling and develop a generalized sampling theorem \bl{based on} the Koopman operator. 
		
		%These related signals generated by systems mostly are non-band-limited, having finite oscillation frequency and finite growth rate of amplitude. These two properties of the signal are characterized separately by imaginary and real part of the spectrum of the generator of the Koopman operator \cite{Koopman1931hamiltonian}, which we simply call Koopman spectrum. Inspired by the generalization of the spectrum, we consider the signal belonging to a generator-bounded space, which requires bounded Koopman spectrum. In other words, this work extends the sampling theorem to signals that oscillate with finite frequency and finite growth of amplitude. 
		
		%Also, the eigenfunctions of the Koopman operator are strongly related to complex exponentials, which is the case of expansion of band-limited signals. Therefore, we propose the sampling theorem of non-band-limited signals by the Koopman operator.

		This work investigates the sampling theorem for perfect reconstruction of \bl{nonparametric} signals that are ``band-limited'' in the sense of Koopman spectrum, where signals' oscillation frequency and growth rate of amplitude are both finite. By describing the signal in the Koopman operator framework, the sampling issue, i.e., one-to-one relationship between the signal and its samples, is translated into a one-to-one map between the DT Koopman operator and its generator. The result shows that the signal aliasing depends on the imaginary part of the Koopman spectrum, where the Nyquist rate is showed to be its special case. Moreover, the reconstruction formula is also represented by the Koopman operator, which can be reduced to classical forms for band-limited \cite{shannon1949communication} and Zakai's \cite{cambanis1976zakai} signal classes. %\xuan{as previous version}\bl{The Koopman spectrum characterizes the oscillation and growth rate through its imaginary and real parts, respectively, which offers a more comprehensive understanding of signal behavior compared to Fourier spectral analysis. This spectral characterization provides a novel perspective for studying the sampling of traditionally non-band-limited signals. Specifically, it extends the sampling bound that prevents signal aliasing from the Nyquist rate to a bound determined by the imaginary part of the Koopman spectrum, which remains finite for many types of non-band-limited signals. Moreover, the reconstruction formula represented by the Koopman operator is general, encompassing many classical forms of certain signal types, including band-limited signals \cite{shannon1949communication} and Zakai's signal classes \cite{cambanis1976zakai}.}
		To numerically illustrate this generalized sampling theorem, the Koopman operator-based reconstruction method is proposed with theoretical convergence. Moreover, several examples of signals are presented for numerical illustration, including band-limited and linear combinations of polynomial and complex exponential signals. 
		
		The rest of this paper is organized as follows. The Nyquist-Shannon sampling theorem and Koopman operator theory are introduced in Section \ref{sec:K_f_s}. The generalized sampling theorem formulated by the Koopman operator is proposed in Section \ref{sec:mainresult}. Then this sampling theorem is analyzed for signals belonging to infinite-dimensional and finite-dimensional spaces in Section \ref{sec:infinitespace} and Section \ref{finitespace}, respectively. In Section \ref{sec:method}, the reconstruction algorithm is provided with convergence result. Finally, the sampling theorem and the reconstruction in the presence of noise are illustrated numerically with four types of examples in Section \ref{sec:num}.%we analyze band-limited signals by  are classified by periodicity, and analyzed respectively by these two sampling theorems. 
		
		\subsection{Notation}
		The paper uses the notation shown here:
		%$\lambda(\cdot)$ denotes eigenvalues (point spectrum) of operators, $\sigma(\cdot)$ denotes the spectrum of operators, 
		%$\hB$ denotes Banach space, %$\hH$ denotes Hilbert space,
		%$\hB$ denotes the Banach space, 
		$\hL(X)$ represents all bounded linear operators from linear space $X$ to itself, $\hD(\cdot)$ represents the domain of the operator, $\sspan\{\cdot\}$ denotes the linear space spanned by basis functions, $\bIm(\cdot)$ represents the imaginary part of complex numbers, $\sigma(\cdot)$, $\sigma_p(\cdot)$ respectively represent the spectrum and eigenvalues (point spectrum) of the linear operator, $U^\tau|_{\hF_e}$ and $L|_{\hF_e}$ represent the restriction of the Koopman operator $U^\tau$ and the generator $L$ to the functional space $\hF_e$, respectively. 
		
		\section{Preliminary}\label{sec:K_f_s}
		%We first formulate the sampling problem and well-known Nyquist-Shannon sampling theorem in Section \ref{sec:problem}. Then we introduce the basic properties of the Koopman operator in Section \ref{sec:intr_koopman}. 
		
		\subsection{Sampling problem and Nyquist-Shannon sampling theorem}\label{sec:problem}
		Consider the original signal $g(t)$ and uniform sampling with sampling period $T_s$. The idealized sampling of $g(t)$ can be represented by a periodic impulse train multiplied by $g(t)$ \cite[Page 516]{oppenheim1997signals}, as illustrated in Fig. \ref{fig00}, where the sampling function is $p(t) = \sum_{k=-\infty}^\infty \delta(t-kT_s)$.  It leads to $$g_p(t) = \sum_{k=-\infty}^\infty g(t)\delta(t-kT_s)$$ and the samples $\{g(kT_s)\}_{k\in\bZ}$. The key question of signal processing is  %\begin{itemize}\item[1)] 
		whether the samples are faithful representations of $g(t)$. %\item[2)] 
		If so, how do we reconstruct $g(t)$ from $\{g(kT_s)\}_{k\in\bZ}$? 
		%\end{itemize}
		\begin{figure}[!t]%[thpb]
			\centering
			\includegraphics[height=.1\textheight]{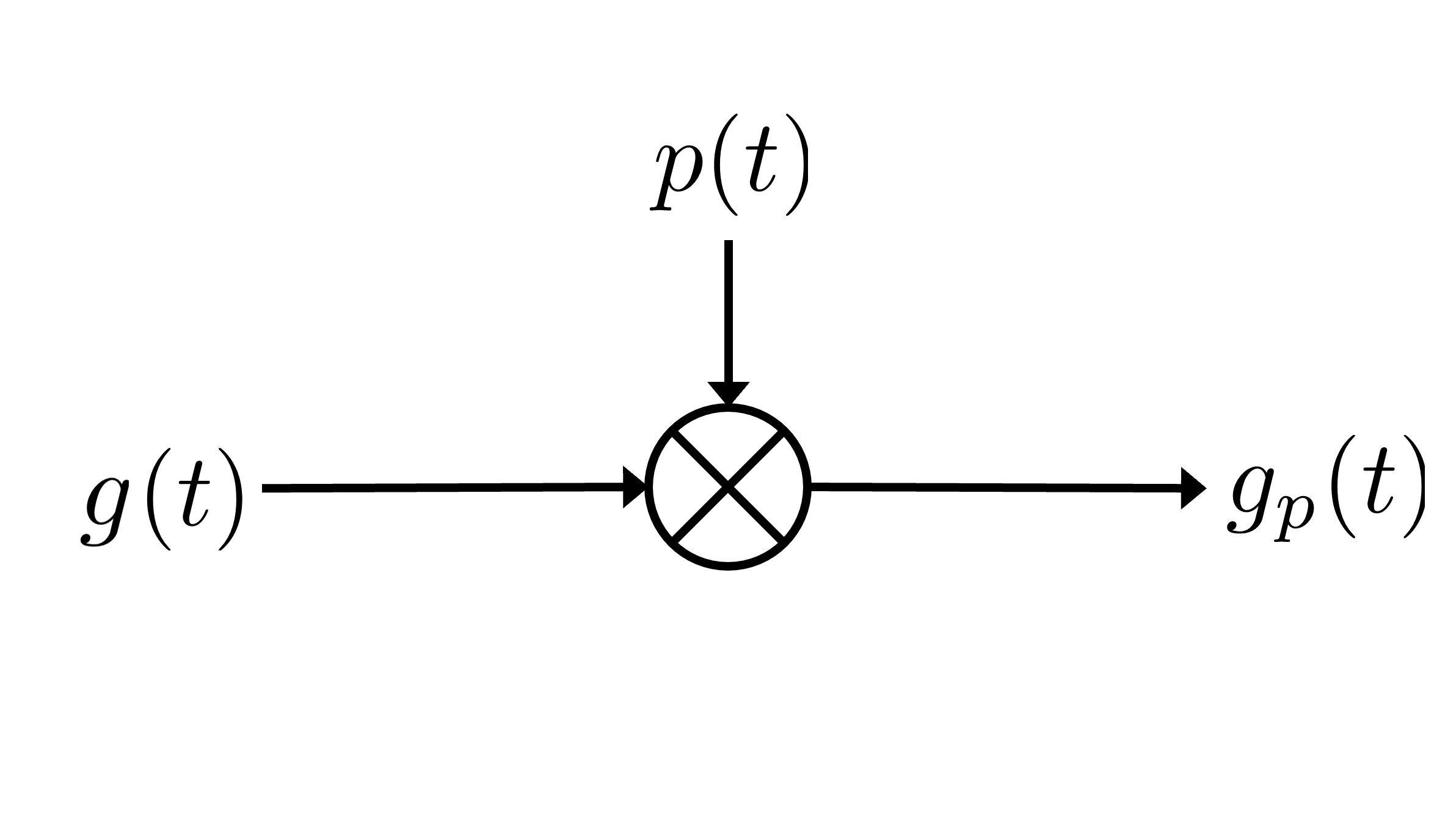}
			\caption{Impulse-train sampling \cite[Page 516]{oppenheim1997signals}. The sampling of the signal $g(t)$ leads to $g_p(t)$, which is represented by the multiplication of $g(t)$ and sampling function $p(t) = \sum_{k=-\infty}^\infty \delta(t-kT_s)$, where $T_s$ is the sampling period.}
			\label{fig00}
		\end{figure}
		
		%\begin{figure}[!t]%[thpb]
		%	\centering
		%	\subfloat[]{
			%		\label{Fig0.sub1}
			%		\includegraphics[height=.21\textheight]{sampling_time}}
		%	\subfloat[]{
			%		\label{Fig0.sub2}
			%		\includegraphics[height=.21\textheight]{sampling_frequency}}
		%	\caption{Effect of impulse-train sampling for band-limited signal in time (a) and frequency domain (b) \cite[Page 516-518]{oppenheim1997signals}.}
		%	\label{fig01}
		%\end{figure}	
		
		The Nyquist-Shannon sampling theorem answers this question for band-limited signals, which plays a fundamental role in signal processing \cite{shannon1949communication}. It states that a band-limited signal $g(t)$ with a maximum frequency $\omega_M$ (rad/s) can be exactly reconstructed from samples with a sampling period $T_s\le \pi/\omega_M$. The reconstruction formula is given as \begin{equation}\label{bandlimited_recons}
			g(t) = \sum_{n=-\infty}^{\infty}g(nT_s)\frac{\sin(\pi/T_s(t-nT_s))}{\pi/T_s(t-nT_s)}.
		\end{equation} %Specifically, the multiplication of $\bl{g(t)}$ and $p(t)$ in the time domain corresponds to convolution of associated Fourier transforms in frequency domain, i.e., \xuan{j i}.$X(j\omega)$ and $P(j\omega)$. Then we denote the resulted Fourier transform of $g_p(t)$ as $X_p(j\omega)$. This effect of sampling is characterized in time and frequency domain in Fig. \ref{Fig0.sub1} and Fig. \ref{Fig0.sub2}, respectively. As Fig. \ref{Fig0.sub2} shows, when $\omega_s>2c$, $X(j\omega)$ can be faithfully reproduced without overlap. And the CT signal $\bl{g(t)}$ can be exactly reconstructed from samples by a suitable lowpass filter \cite[Page 516-518]{oppenheim1997signals}.
		\bl{The study of the sampling theorem has a long history as detailed in \cite{butzer2011interpolation}. In fact, this reconstruction \eqref{bandlimited_recons} was first clearly articulated by Ogura in his paper \cite{ogura1920certain}, which predates Shannon's well-known work \cite{shannon1949communication}.}
		
		Since this work generalizes the sampling theorem using the Koopman operator, we briefly introduce the Koopman operator theory for a space of signals as follows.
		
		\subsection{The Koopman operator}\label{sec:intr_koopman}
		Consider the semigroup of the Koopman operator $U^\tau,\tau\ge 0$ \cite{Koopman1931hamiltonian}, which is defined as follows
		\begin{equation}\label{eq4}
			U^\tau g=g\circ S^\tau  , \,\, g\in\hF,
		\end{equation} where $g:\bR\to\bC$ is a function that belongs to a Banach space $\hF$, and $S^\tau(t) = t+\tau$ is the flow of time $t\in\bR$. The Koopman operator $U^\tau$ is linear, i.e., for signal functions $g_1,g_2\in \hF$,
		\begin{equation*}
			U^\tau (\alpha_1 g_1+\alpha_2g_2)=\alpha_1U^\tau  g_1+\alpha_2U^\tau  g_2, \ \alpha_1,\alpha_2\in \bR.
		\end{equation*}
		Therefore, $U^\tau$ %interprets the flow $S^\tau $ in operator framework, and 
		describes the evolution of signal functions in $\hF$ in terms of linear transformation. %dynamical interpretation of $U^\tau $ is as follows: the linear transformation of the observable, i.e., $U^\tau  g(\bx)$, has the value of observable function $g$ at $S^\tau (\bx)$. Specifically, if $g(\bx)=x_j$, where $x_j\ (1\le j\le n)$ is the $j$th component of $\bx$, $U^\tau  g(\bx(0))$ has the value that the initial state $x_j(0)$ flows after the lapse of the time $t$. %Therefore, Koopman operator $U^\tau $, which is an infinite-dimensional and linear operator, interprets the flow $S^\tau $ in operator framework.
		
		The infinitesimal generator $L$ of the Koopman operator is also a linear operator, which is defined as:
		\begin{equation}\label{eq5}
			L g=\lim_{\tau \to0^+}\frac{1}{\tau}(U^\tau -I)g, \ g\in \hD(L).
		\end{equation} So $Lg$ can be seen as the ``derivative'' of the signal function $g$ with respect to time, i.e., $\dot{g}=L g$, where $\dot{g}$ denotes $\partial(g\circ S^\tau )/\partial{\tau}|_{\tau=0}$\cite{mauroy2020koopman}. 
		%When $L$ acts on the full observable functions, i.e., $g(\bx)=\bx$, one obtains the vector field $L g = f$. This connection between the generator $L$ and $f(\bx)$ stems from the formula:
		%\begin{equation*}\lim_{t \to0^+}\frac{1}{t}(S^\tau (\bx)-\bx)=f(\bx).\end{equation*}
		%If the observable function $g$ is continuously differentiable with compact support, there is a general relationship between the vector field $f$ and the generator $L$ \cite[Section 7.6]{lasota2013chaos}:
		
		%which provides tools to study nonlinear dynamical system as if they were linear:
		%It should be noted that, when the observable function are $g(\bx(t)) = t$, the generator 
		%The infinitesimal generator $L$ is also a linear operator. Exploiting \eqref{eq4}-\eqref{eq6}, the analysis of nonlinear systems by Koopman operators can be interpreted as an extension problem of linear systems, from finite-dimensional to infinite-dimensional ones. In particular, the state space $\bM$ is extended to an observable space $\hF$ so that the nonlinear vector field $f$ becomes embedded in the linear operator $L$. The extended linear system of observable functions $g\in \hF$ can be described by \begin{equation}\label{eq7}\begin{aligned}\dot{g}&=L g,\\U^{T_s}g&=g\circ S^{T_s},\end{aligned}\end{equation}where $\dot{g}$ denotes the ``derivative'' of the observable function $\partial(g\circ S^\tau )/\partial{t}$\cite{mauroy2020koopman}. 
		
		When $U^\tau$ and $L$ are both bounded operators, we have \begin{equation}\label{exp}
			U^\tau  = e^{L\tau}.
		\end{equation} Besides, the spectrum also admits this exponential relationship \cite[Lemma 3.13]{engel2000one}, i.e., $$\sigma(U^\tau) = \{e^{\lambda \tau}:\lambda\in\sigma(L)\}.$$ In the following, we also call $\sigma(L)$ as the Koopman spectrum. 
		
		Here we introduce some definitions of the spectrum of a linear operator $L\in \hL(X)$. Denote $I$ as the identity operator on $\hF$. Then a complex value $\lambda$ belongs to the spectrum of $L$, i.e., $\lambda\in\sigma(L)$, if the operator $L-\lambda I$ does not have an inverse. The value $\lambda$ belongs to the point spectrum of $L$, i.e., $\lambda\in\sigma_p(L)$, if $L - \lambda I$ is not injective. 
		When the Koopman operator $U^\tau$ admits a point spectrum, the eigenvalue and associated eigenfunction of the Koopman operator are defined as follows.
		\begin{mydef}[Koopman eigenvalues and eigenfunctions]
			Given the Koopman operator $U^\tau$, the function $\phi_k\in\hF, \phi_k\neq0$ is an eigenfunction of $U^\tau$ if
			$$
			U^\tau\phi_k = e^{\lambda_k \tau }\phi_k,
			$$
			where $\lambda_k\in\bC$ is the associated Koopman eigenvalue.%\ye{define $\sigma_p(U^\tau )$}
		\end{mydef}

		%It is composed of three parts: point spectrum $\sigma_p(\cdot)$, residual spectrum $\sigma_r(\cdot)$, and continuous spectrum $\sigma_c(\cdot)$.
		The function $\phi_k$ is also an eigenfunction of the generator $L$ when it is an eigenfunction of the Koopman operator, i.e.,
		$$
		L\phi_k = \lambda_k \phi_k,
		$$ where $\lambda_k$ is the associated eigenvalue of the generator $L$, and it belongs to its point spectrum $\sigma_p(L)$.

		\section{The generalized sampling theorem}\label{sec:mainresult}
		
		\subsection{Signal Class}\label{sec:signal}
		We address the sampling theorem for signals $g(t)$ belonging to the generator-bounded space $\hF_e$. In this section, we first present the definition and properties of $\hF_e$. 
		%\subsection{Generator-bounded space}\label{sec:space}
		The generator-bounded space $\hF_e$ is defined as follows. 
		
		%known: S^{T_s} U^{T_s} g(nT_s). unknown:U^\tau , g(t)
		\begin{mydef}[Generator-bounded space $\hF_e$]\label{space}
			The generator-bounded space $\hF_e$ is a separable Banach space that satisfies: \begin{itemize}
				\item[1)] The space $\hF_e$ is invariant under the Koopman operator $U^\tau $, i.e., $\forall \tau\ge 0, \forall g\in\hF_e, g(t+\tau)\in\hF_e$.%The space $\hF_e$ is invariant under the Koopman operator $U^\tau $, i.e., $\forall g\in\hF_e,\forall  \tau\ge0$, we have $U^\tau  g\in \hF_e$.
				\item[2)] The generator $L|_{\hF_e}$ is bounded, i.e., $\|L|_{\hF_e}\|=\sup_{\|g\|=1}\|g^{(1)}\|\bl{<}\infty$, where $g^{(1)}$ is the derivative of $g$.
				%\item[3)] The map from $\{c_k\}_{k=1}^{n}$ to $\{\bl{g}(k T_s)\}_{k=0}^{n}$ is invertible, where $c_k$ is the $k$th coefficient of expansion of $\bl{g}$, i.e., $\bl{g} = \sum_{k=1}^{n}c_k\phi_k$.
			\end{itemize}
		\end{mydef}
		
		Here we show the property that a signal $g(t)$ belonging to $\hF_e$ oscillates with a finite frequency and a finite growth rate of amplitude at $t<\infty$. Since  $g\in\hF_e$ and $\hF_e$ is an invariant space of the Koopman operator, the value of the signal $g(t)$ at $t = \tau$ is represented by the Koopman operator $U^\tau|_{\hF_e}$, i.e., $$g(\tau) = U^\tau |_{\hF_e}g(0) = e^{L|_{\hF_e}\tau}g(0) .$$
		When $\hF_e$ is infinite-dimensional, it follows from the functional calculus that %${\color{red} ?}$\begin{equation}\label{eqnpr}\bl{g(t)} =  U^\tau |_{\hF_e}\bl{g} (0)= \int_{\sigma(L|_{\hF_e})}e^{\lambda t}{\rm E}({\rm d}\lambda)\bl{g(0)},\end{equation}where ${\rm E}$ is the resolution of the identity for the generator $L|_{\hF_e}$\cite[Chapter X]{dunford1954spectral}. 
		\begin{equation}\label{functionalc}
			g(\tau) =  U^\tau |_{\hF_e}g (0)=\int_{\sigma(L|_{\hF_e})}e^{\lambda \tau}E({\rm d}\lambda)g(0),
		\end{equation} where $E$ is the resolution of the identity for $L|_{\hF_e}$ \cite[P. 898]{dunford1954spectral}. Besides, when $\hF_e$ is finite-dimensional, i.e., $\hF_e=\sspan\{\phi_1,\ldots,\phi_n \}$ with the spectrum $\sigma(L|_{\hF_e}) = \{\lambda_k\}_{k=1}^n$, then $g(\tau)$ is determined by 
		\begin{equation}\label{eqpr}
			g(\tau) = U^\tau |_{\hF_e}g (0)=\sum_{k=1}^{n} e^{\lambda_k \tau} E(\lambda_k) g(0).\end{equation}
		Hence the oscillation and growth of amplitude are characterized by the real and imaginary parts of $\lambda\in \sigma(L|_{\hF_e})$, respectively. Based on the theorem of spectral radius \cite[P. 241]{engel2000one}, i.e., $$r_\sigma(L|_{\hF_e})\le\|L|_{\hF_e}\|,$$ where the spectral radius $r_\sigma(L|_{\hF_e}) := \sup\{|\lambda|:\lambda\in\sigma(L|_{\hF_e})\}$, we have a bounded spectrum $\sigma(L|_{\hF_e})$ because $\|L|_{\hF_e}\|<\infty$. 
		Hence, the signal $g(t)\in\hF_e$ is ``band-limited'' in the sense of Koopman spectrum, which oscillates with a finite frequency and a finite growth rate of amplitude for $t<\infty$.
		
		\begin{remark}
			The generator-bounded space $\hF_e$ is a natural extension of the band-limited signal space, as it transforms the characterization of signal properties from the one-dimensional Fourier frequency $\omega\in\bR$ to the richer two-dimensional Koopman spectrum $\lambda\in\bC.$ In particular, a band-limited signal is determined by a limited Fourier frequency domain, which corresponds to the bounded Koopman spectrum that \bl{lies} only on the imaginary axis of the complex space $\bC$. Consequently, based on the physical interpretation of the Koopman spectrum, a signal that is ``band-limited'' in this spectrum suggests a natural expansion of band-limited signals to much broader signal classes allowing for non-zero growth rates of amplitude. %In essence, the generator-bounded space broadens the scope of signals beyond traditional band-limitations, opening the door to a richer and more nuanced understanding of signal characteristics.
			%Hence, based on the physical interpretation of the Koopman spectrum, the signal that is ``band-limited'' in the Koopman spectrum suggests a natural extension of band-limited signals to a much broader class of signals that allow non-zero growth rate of amplitude.
			%the signal that is ``band-limited'' in the sense of Koopman spectrum can include much more signals.
			
		\end{remark}
		
		\subsection{The lower bound of sampling frequency to avoid aliasing}\label{sec:koop_analyze}
		In this section, we propose the Koopman operator-based sampling theorem for signals $g(t)$ in the generator-bounded space $\hF_e$. It specifies the lowest sampling frequency of samples to reconstruct the original signal without aliasing, which is determined by the imaginary part of the Koopman spectrum. %This Koopman operator-based sampling theorem is proved in Section \ref{sec:maintheorem}. For further investigation of the result $\omega_\gamma$, the intrinsic connection between signal aliasing and Koopman spectrum is explored in Section \ref{sec:theorem_signal}, showing the extension from Fourier spectrum to Koopman spectrum.
		
		%\subsection{Main theorem}\label{sec:maintheorem}
		%In this section, we will prove the Koopman operator-based sampling theorem.
		
		We analyze the sampling theorem by the Koopman operator by transforming the one-to-one relationship between the signal and its samples to that between the generator $L|_{\hF_e}$ and the DT Koopman operator $U^{T_s}|_{\hF_e}$. Specifically, consider the signal $g\in\hF_e$ and $U^{T_s}|_{\hF_e}$ defined by the DT flow $S^{T_s}(t) = t+T_s$. Then $U^{T_s}|_{\hF_e}$ describes samples of sampling period $T_s$, i.e., $U^{T_s}|_{\hF_e} g(t_0) = g(t_0+T_s)$. By obtaining the generator $L|_{\hF_e}$ from $U^{T_s}|_{\hF_e}$ uniquely, the signal $g(\tau)$ can be reconstructed by $g(\tau) = U^\tau |_{\hF_e}g(0) = e^{L|_{\hF_e}\tau}g(0)$. In summary, the framework of the analysis is illustrated in Fig. \ref{fig1}.
		
		\begin{figure}
			\centering
			%\framebox{
				\includegraphics[width=0.48\textwidth]{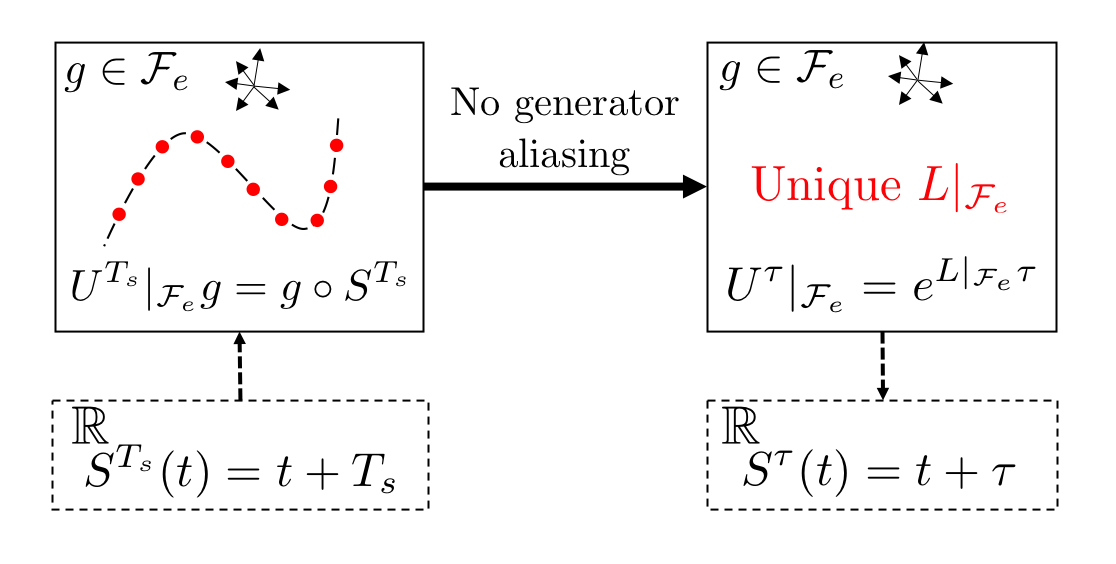}
				%}
			%\includegraphics[scale=1.0]{figurefile}
			\caption{The framework to investigate the sampling theorem of signals by the Koopman operator. The analysis consists of two steps: (1) Considering the sampling problem in the generator-bounded space $\hF_e$ of signal $g(t)$, (2) obtaining the generator $L|_{\hF_e}$ uniquely from DT Koopman operator $\bU|_{\hF_e}$, (3) reconstructing the signal $g(t)$ by the generator $L|_{\hF_e}$.}
			\label{fig1}
		\end{figure}
		For the convenience of proving the main result, we first introduce the essential lemma to guarantee the uniqueness of operator logarithm.
		
		\begin{lemma}[Principal logarithm{\cite[Thm 2]{krabbe1956logarithm}}]\label{principal}
			Consider an operator $Y\in\hL(X)$ whose logarithm is well-defined and denote the principal logarithm of $Y$ as $B=\Log(Y)$. The principal logarithm $B$ is uniquely obtained, where the spectrum satisfies $\sigma(B)\subset \hG(\pi)$ and $\hG(\pi)=\{z \in \bC:-\pi <\bIm(z)< \pi \}$. 
		\end{lemma}
		
		Based on Lemma \ref{principal}, here we propose the sampling theorem for signals belonging to the generator-bounded space $\hF_e$. 
		
		\begin{theorem}[Koopman operator-based sampling theorem]\label{thKs}
			%\ye{re-write this statement}
			Consider a signal $g(t)$ belonging to a generator-bounded space $\hF_e$. There is no aliasing of the signal $g(t)$ %can be reconstructed from its uniform samples 
			if and only if the sampling frequency (rad/s) satisfies  
			\begin{equation}\label{eqcons}
				\omega_s>2\min_{\hF_e}\{\max|\bIm(\sigma(L|_{\hF_e}))|\}.
			\end{equation}
		\end{theorem}
		
		\begin{proof}
			%Since the samples of  $\bl{g(t)}$ is described by DT Koopman operator $U^{T_s}|_{\hF_e}$, i.e., $\bl{g}((k+1)T_s) = U^{T_s}|_{\hF_e}\bl{g}(kT_s), k\in\bZ$, and the signal can be represented by $U^tg(0)$, the 
			The samples of $g(t)$ is described by the DT Koopman operator $U^{T_s}|_{\hF_e}$, i.e., $g((k+1)T_s) = U^{T_s}|_{\hF_e}g(kT_s), k\in\bZ$. Since $U^{T_s}|_{\hF_e}$ and its generator $L|_{\hF_e}$ are both bounded, we have
			\begin{equation*}
				U^{T_s}|_{\hF_e}=\exp(L|_{\hF_e}T_s)% = \sum_{k=0}^{\infty}\frac{{T_s}^k {L|_{\hF_e}}^k}{k!}
				,
			\end{equation*} where the exponential is defined by \cite[P. 172]{hille1996functional}.
			Then $L|_{\hF_e}$ can be obtained from $U^{T_s}|_{\hF_e}$ by operator logarithm \cite[P. 172]{hille1996functional}. However, there are multiple solutions of the logarithm \cite[P. 173]{hille1996functional}, i.e., 
			\begin{align*}
				\log(U^{T_s}|_{\hF_e})&=\log(\exp (L|_{\hF_e}T_s))=L|_{\hF_e}T_s+2\pi i \sum_{k}n_k I_k,
			\end{align*}where $n_k$ are integers, $I_k$ denote idemponents, i.e., $\exp(2\pi i I_k)=I$, and $I$ denotes the identity operator. Based on Lemma \ref{principal}%the essential theorem of linear operator logarithm \cite[Thm 2]{krabbe1956logarithm}
			, the generator $L|_{\hF_e}$ is obtained uniquely from $U^{T_s}|_{\hF_e}$ if and only if $L|_{\hF_e}T_s$ is principal logarithm, i.e., $L|_{\hF_e}T_s = \Log(U^{T_s}|_{\hF_e})$. It follows that the spectrum of $T_sL|_{\hF_e}$ lies in the strip $\{z \in \bC:-\pi <\bIm(z)< \pi \}$, i.e., %$$\sigma(L|_{\hF_e})\in\{z \in \bC:-\pi/T_s <\bIm(z)< \pi/T_s \}.$$ 
			\begin{equation}\label{omega_s}
				\omega_s=2\pi/T_s>2\max|\bIm(\sigma(L|_{\hF_e}))|.
			\end{equation}
			Since $g(\tau) = U^\tau|_{\hF_e}g(0)$ can be reconstructed if there exists one $\hF_e$ such that $L|_{\hF_e}$ can be obtained from $U^{T_s}|_{\hF_e}$. Therefore, we only require the ``kernel'' generator-bounded space \begin{equation}\label{kernel}%\boxed{
					\hF_e^\gamma = \arg\min_{\hF_e}\{\max|\bIm(\sigma(L|_{\hF_e}))|\},%}
			\end{equation} to satisfy \eqref{omega_s}. %So we get the frequency condition \eqref{eqcons}. 
			%By taking $\bl{g}$ as one of basis functions of $\hF_e^\gamma,$ it can be represented by the coordinate $\{c_k\}_{k\in\bN^+},$ where $c_1 = 1$ and $c_k = 0$ for $k\neq 1$. 
			Then the value of $g(t)$ at $t=\tau$ can be reconstructed by \begin{equation}\label{reconstruction} 
				g(\tau) = \exp\left(\frac{\tau}{T_s}\Log(U^{T_s}|_{\hF_e^\gamma}) \right)g(0).
			\end{equation}
		\end{proof}	
		
		The sampling bound in \eqref{eqcons} ensures that $g(t)$ is the unique function in $\hF_e^\gamma$ passing through the samples. \bl{When the sampling frequency is too low to satisfy \eqref{eqcons}, signal aliasing exists. Specifically, there exists another signal $\hat{g}(t)$ that generates all the samples and is characterized by a Koopman spectrum whose imaginary part is confined to $i[-\omega_s/2, \omega_s/2]\subsetneqq \bIm(\sigma(L|_{\hF_e^\gamma}))$. This results in a different signal $\hat{g}(t)$ with lower oscillation frequency compared to $g(t)$, which will be illustrated numerically in later sections.}
		
		Moreover, the sampling bound $\omega_\gamma = 2\max|\bIm(\sigma(L|_{\hF_e^\gamma}))|$ is related only to the imaginary part of the Koopman spectrum, which reflects the oscillation frequency of the signal. Therefore, this result coincides with the fact that the oscillation causes aliasing while amplitude growth rate does not. %Moreover, in order to numerically realize the reconstruction \eqref{reconstruction}, we propose an algorithm based on the Koopman operator in Section \ref{sec:method}. 	
		
		\begin{remark}[The reason for generalization]
			%Fig. \ref{fig7} show that the sampling causes the superposition of shifted transformed functions along imaginary axis, where the value of shift interval depends on the sampling frequency $\omega_s$. To avoid the aliasing of transformed functions of signals, the lower bound of sampling frequency is only related to the imaginary part of $\sigma(L|_{\hF_e^\gamma})$, which is consistent with Theorem \ref{thKs}. Therefore, t
			The reason why the Koopman operator-based sampling theorem extends Nyquist-Shannon sampling theorem is that the two-dimensional Koopman spectrum is able to independently characterize the phase oscillation frequency that causes signal aliasing and the exponential growth of amplitude that does not. 
			Hence, when the sampling is fast enough to capture the phase oscillation of $g\in\hF_e^\gamma$ whose growth rate is limited at $t<\infty$, the signal $g(t)$ can be definitely determined by the samples because it is the unique solution in $\hF_e^\gamma$ that matches $g(kT_s),k\in\bN$%, where Koopman spectrum $\sigma(L|_{\hF_e^\gamma})\subset\bC$ is a bounded set
			. 
			
			%Theorem \ref{thKs} shows that the lower bound of sampling frequency is only related to the imaginary part of the spectrum, which is associated with the oscillation of the signal. Therefore, the reason why it extends Nyquist-Shannon sampling theorem is because the spectrum of generator is 2-dimensional and is able to describe the oscillation frequency and the growth rate of amplitude of signals separately. Then it extends to the signal that oscillates with finite frequency and finite amplitude growth rate. %Thus this type of signal is able to be reconstructed from samples of finite sampling frequency when $\sigma(L|_{\hF_e})$ is bounded, but might have unlimited Fourier spectrum. %It also sheds light on the fact that, with the suitable functions, such signals can be reconstructed from samples.
		\end{remark}

		\subsection{Reconstruction from samples}\label{sec:formula}
		
		When the sampling period $T_s$ satisfies \eqref{eqcons}, we have $$\exp\left(\frac{\tau}{T_s}\Log(U^{T_s}|_{\hF_e^\gamma}) \right) = \exp(\tau L|_{\hF_e^\gamma})=U^\tau|_{\hF_e^\gamma}.$$ It leads \eqref{reconstruction} to $
		g(\tau) = U^{\tau}|_{\hF_e^\gamma} g(0),$ where $\tau\ge 0$. Denote $\Phi = [\varphi_0,\ldots,\varphi_{M-1}]$ as a set of basis functions of $\hF_e^\gamma$. \bl{Since $\hF_e^\gamma$ is invariant under the action of the Koopman operator (see \eqref{kernel} and Definition \ref{space}), we have} $U^\tau|_{\hF_e^\gamma}g = \Phi \bm c(\tau)$, where $\bm c(\tau) = [c_1(\tau),\ldots,c_M(\tau)]^T$ is the coordinate of $U^\tau|_{\hF_e^\gamma}g$ on the basis $\Phi$. Then we have \begin{equation}\label{rerecons}
			g(\tau) = U^\tau|_{\hF_e^\gamma}g(0) = \Phi(0)\bm c(\tau),
		\end{equation} where \bl{the explicit representation of} $\bm c(\tau)$ is determined by the specific type of signals and basis $\Phi$.% and governed by the Koopman operator $U^\tau|_{\hF_e^\gamma}$. 
		
		Hence, the formula \eqref{rerecons} implies that the reconstruction of $g$ may leverage the values of unknown basis functions $\Phi(0)$, even though Theorem \ref{thKs} guarantees $g(t)$ as the unique function in $\hF_e ^{\gamma}$ passing through all the samples $g(kT_s)$. However, the ability to reconstruct the signal from its samples is the fundamental statement of the sampling theorem. Hence, another requirement and its criterion are presented to ensure that the signal can be expressed by its samples. %Moreover, the reconstruction formula is provided for specific signal types.
		%\subsection{The requirement of reconstruction from samples}
		\begin{theorem}[Reconstruction requirement]\label{recons_requrie}
			Consider the signal $g\in \hF_e^\gamma$ (see \eqref{kernel}) and sampling period $T_s = 2\pi/\omega_s$ satisfying \eqref{eqcons}. The signal $g(t)$ can be represented by the value of its samples if $\{U^{kT_s}|
			_{\hF_e^\gamma}g(t)=g(t+kT_s)\}_{k=0,\ldots,M-1}$ is a set of basis functions of $\hF_e^\gamma$.
		\end{theorem}
		\begin{proof}
			The signal can be written as $g(t) = \sum_{k=0}^{M-1} c_k(0) g(t+kT_s)$, where $c_k(0) = 1$ for $k=0,$ $c_k(0) = 0$ for $k\neq 0$. If $\{g(t+kT_s)\}_{k=0,\ldots,M-1}$ is a set of basis functions of $\hF_e^{\gamma}$, it follows \bl{from} $U^\tau|_{\hF_e^\gamma}g\in\hF_e^\gamma$ that for $\forall \tau\ge0$,  \begin{equation}\label{recon_formula}
				g(t+\tau) = U^\tau|_{\hF_e^\gamma} \sum_{k=0}^{M-1} c_k(0) g(t+kT_s) = \sum_{k=0}^{M-1} c_k(\tau) g(t+kT_s).
			\end{equation} %where $c_k(\tau)$ is governed by $U^\tau|_{\hF_e^\gamma}$
			Then $g(\tau)$ is represented by its samples by letting $t=0$.
		\end{proof}
		\begin{remark}[The comparison with shift-invariant (SI) space]
			The reconstruction formula \eqref{recon_formula} seems to be similar with that of signals in SI spaces, i.e., $g(t)=\sum_{k\in\bZ}c_k\varphi(t-k),$ where $\varphi(t)$ is a generating function satisfying Riesz basis condition and the partition of unity \cite{unser2000sampling}. However, these two formula essentially reconstruct signals from two different perspectives. Specifically, SI spaces are spanned by known basis functions $\{\varphi(t-k)\}_{k\in\bZ}$, which to some extent leverages signal structure. Then the goal of signal reconstruction is to recover coefficients $\{c_k\}_{k\in\bZ}$ from samples $\{g(kT_s)\}_{k\in\bZ}$. In contrast, the formula \eqref{recon_formula} is built on obtaining the evolution $U^\tau g(t)= g(t+\tau)$ for a given $t=0$. Leveraging the property that  $g(t+\tau)\in\sspan_{k=0,\ldots,M-1}\{g(t+kT_s)\}$, the values of time-delay functions $\{g(kT_s)\}_{k\in\bZ}$ at $t=0$ are available samples. Then the values of samples $\{g(kT_s)\}_{k=0}^{M-1}$ become coefficients and $\{c_k(\tau)\}_{k=1}^{M-1}$, governed by the Koopman operator, become basis functions in the reconstruction formula. Therefore, this method reconstructs signals by recovering basis functions $\{c_k(\tau)\}_{k=1}^{M-1}$ by the Koopman operator, with known coefficients $\{g(kT_s)\}_{k=0}^{M-1}$. 
		\end{remark}
		Here we further provide a criterion to determine whether this requirement (Theorem \ref{recons_requrie}) can be satisfied or not when $\hF_e^\gamma$ is finite-dimensional. %$\{U^{kT_s}g(t)=g(t+kT_s)\}_{k=0,\ldots,M-1}$ can be a set of basis of $\hF_e$ when $\hF_e$ is finite-dimensional.
		
		\begin{proposition}[Criterion of time-delay basis]\label{criterion}
			Given finite-dimensional generator-bounded space $\hF_e^\gamma$ with a set of basis function $\Phi = [\varphi_0,\ldots,\varphi_{M-1}]$. The time-delay functions $\{g(t+kT_s)\}_{k=0}^{M-1}$ can be basis functions of $\hF_e^\gamma$ if and only if \begin{equation*}
				\mathrm{rank}[\lambda I-U_M ,\bm a]=M
			\end{equation*} for every eigenvalue $\lambda$ of $U_M$, where $U_M$ denotes the matrix representation of the Koopman operator $U^{T_s}|_{\hF_e^\gamma}$, and $\bm a = [a_0,\ldots,a_{M-1}]^T$ denotes the coordinate of $g$, i.e., $U^{T_s}|_{\hF_e^\gamma} \Phi = \Phi U_M, g = \Phi \bm a$.
		\end{proposition}
		\begin{proof}
			The functions $\{g(t+kT_s)\}_{k=0}^{M-1}$ are basis functions of $\sspan_{k=0}\{\varphi_k, \ldots, \varphi_{M-1}\}$ if and only if $\{g(t+kT_s)\}_{k=0}^{M-1}$ are linearly independent. Since $\hF_e^\gamma$ is invariant under the action of Koopman operator, given basis functions $\Phi = [\varphi_0,\ldots,\varphi_{M-1}]$ and signal $g = \Phi \bm a$, we have \begin{equation*}
				g(t+T_s) = U^{T_s}|_{\hF_e^\gamma} g(t) = \Phi(t) U_M \bm a,
			\end{equation*} where $U_M\in\bR^{M\times M}$ denotes the matrix representation of $U^{T_s}|_{\hF_e^\gamma}$, $\bm a \in\bR^{M\times 1}$ denotes the coordinate of $g$ on the basis $\Phi$. Then the linearly independence of $\{g(t+kT_s)\}_{k=0,\ldots,M-1}$ is equivalent to \begin{equation}\label{rank}
				\mathrm{rank} [\bm a, U_M \bm a, \ldots, U_M^{M-1}\bm a] = M.
			\end{equation}
			\bl{This requirement is equivalent to ensuring that the linear system defined by $U_M$ and $\bm a$, i.e., $\dot{\bx}(t) = U_M \bx(t)+ \bm a u(t)$, is controllable \cite[Theorem 12.1, Page 149]{hespanha2018linear}. Based on the Popov-Belevitch-Hautus test for controllability \cite[Theorem 12.3, Page 152]{hespanha2018linear}}, the necessary and sufficient condition of \eqref{rank} is $\mathrm{rank}[\lambda I-U_M, \bm a] = M$ for every eigenvalue $\lambda$ of $U_M$.
		\end{proof}
		Although the criterion for finite-dimensional $\hF_e^\gamma$ is proposed by Proposition \ref{criterion}, it is still challenging to derive a similar result for signals in infinite-dimensional space. In particular, the Popov-Belevitch-Hautus test used in the proof of Proposition \ref{criterion} \bl{remains an} unsolved problem for infinite-dimensional systems in the mathematical control field \cite{blondel2009unsolved}.

		To further illustrate the Koopman operator-based sampling theorem, we present two spaces of signals as examples of infinite-dimensional and finite-dimensional generator-bounded space $\hF_e$ in the following, of which non-periodic and periodic band-limited signals are proved to be their special cases.
		
		\section{Infinite-dimensional signal space}\label{sec:infinitespace}
		Since Laplace transform is both a powerful tool in signal processing and a natural extension from Fourier transform, here we consider \bl{inverse Laplace-type} signal $g(t) = \frac{1}{2\pi}\int_{-c}^{c}G(\omega)e^{(\alpha+i\omega) t}{\rm d}\omega$, where $G(\omega)\in L^2[-c,c], c\ge 0$, as an example in infinite-dimensional $\hF_e$, which is common in the field of engineering. 
		\subsection{Inverse Laplace-type signal}\label{sec:invers}	
		%In this section, we show that inverse-Laplace type of signal $\bl{g(t)} = \int_{-c}^{c}F(\omega)e^{(\alpha+i\omega) t}{\rm d}\omega$, where $F(\omega)\in L^2[-c,c], c\ge 0$, belongs to generator-bounded space $\hF_e$. And non-periodic band-limited signal is a particular case that spectrum $\sigma(L|_{\hF_e})$ is restricted to imaginary axis of $\bC$. 
		By the definition of the signal $g(t)$, it belongs to the space $\left\{g(t) = \frac{1}{2\pi}\int_{-c}^{c}G(\omega)e^{(\alpha+i\omega) t}{\rm d}\omega:G(\omega)\in L^2[-c, c]\right\}$. To prove that this space is a generator-bounded space, we first introduce a lemma.% that shows the property of $g_0(t) = \frac{1}{2\pi}\int_{-c}^{c}G(\omega)e^{i\omega t}{\rm d}\omega$.

		\begin{lemma}[Parseval’s Theorem]\label{lemma7}
			For the function $g_0(t) = \frac{1}{2\pi}\int_{-c}^{c}G(\omega)e^{i\omega t}{\rm d}\omega$, we have \begin{align*}
				\int_{-\infty}^{\infty}|g_0(t)|^2{\rm d}t = \frac{1}{2\pi}\bl{\int_{-c}^{c}}|G(\omega)|^2{\rm d}\omega,
			\end{align*} where $G(\omega)$ is the Fourier transform of $g_0(t)$.
		\end{lemma}
		%\bl{Since the Fourier transform $G(\omega)$ of $g(t) = \int_{-\infty}^{\infty}h(\omega)e^{i\omega t}{\rm d}\omega$ is $G(\omega) = \sqrt{2\pi} h(\omega)$, we have \begin{align}\label{eqparseval}\int_{-\infty}^{\infty}|g(t)|^2{\rm d}t = \int_{-c}^{c}|h(\omega)|^2{\rm d}\omega.\end{align}}
		
		\begin{proposition}[Inverse Laplace-type signal]\label{non-bandlimited}
			The space of inverse Laplace-type signals \begin{equation}\label{inverse}
				\hF_e = \left\{g(t) = \frac{1}{2\pi}\int_{-c}^{c}G(\omega)e^{(\alpha+i\omega) t}{\rm d}\omega:G(\omega)\in L^2[-c, c]\right\},
			\end{equation}where $c>0,\alpha\in\bR$, is a generator-bounded space.
		\end{proposition}
		\begin{proof}
			We first prove that $\hF_e$ given by \bl{\eqref{inverse}} is a Hilbert space. For $\forall g\in\hF_e$, we can write it as follows
			\begin{equation*}
				g(t) = e^{\alpha t}g_0(t),
			\end{equation*} where $g_0(t)$ belongs to the traditional band-limited space, i.e.,
			\begin{equation*}
				\left\{g_0(t) = \frac{1}{2\pi}\int_{-c}^{c}G(\omega)e^{i\omega t}{\rm d}\omega:G(\omega)\in L^2[-c,c]\right\}\subset L^2.
			\end{equation*}
			Hence, the space \bl{\eqref{inverse}} is a Hilbert space with the weighted $L^2(e^{-2\alpha t})$ norm $\|\cdot\|$.
			Then we show that the Koopman operaor $U^\tau$ is invariant on $\hF_e$. For $\forall \tau>0, \forall g\in\hF_e$, we have \begin{align*}
				U^\tau  g(t) = g\circ S^\tau (t) = \bl{\frac{1}{2\pi}}\int_{-c}^{c} \bl{G(\omega)}e^{(\alpha+i\omega)\tau}e^{(\alpha+i\omega)t}{\rm d}\omega. 
			\end{align*} It follows that $U^\tau  g\in \hF_e$ because $e^{(\alpha+i\omega)\tau}\bl{G(\omega)}\in L^2[-c,c]$. 
			Here we prove that the restriction of the generator $L|_{\hF_e}$ is bounded. It follows \bl{from} the definition of the generator that $L|_{\hF_e} g(t) = g^{(1)}(t)$, where $g^{(1)}(t)$ denotes the derivative, i.e.,
			\begin{equation*}
				g^{(1)}(t) = g^{(1)}_0 (t)e^{\alpha t} + \alpha e^{\alpha t}g_0(t)
			\end{equation*} Hence, we have $g^{(1)}(t)e^{-\alpha t} = g^{(1)}_0 (t) + \alpha g_0(t)$. The Fourier transform of $g^{(1)}(t)e^{-\alpha t}$ is $(i\omega+\alpha)G(\omega)$, where $G(\omega)$ is the Fourier transform of $g_0(t) = \frac{1}{2\pi}\int_{-c}^{c}G(\omega)e^{i\omega t}{\rm d}\omega$. It follows \bl{from} Lemma \ref{lemma7} that
			\begin{equation}\label{limit_derivative}
				\begin{aligned}
					&\|g_0^{(1)}+\alpha g_0\|_2^2=\frac{1}{2\pi}\int_{-c}^{c}(\omega^2+\alpha^2)|G(\omega)|^2 {\rm d}\omega\\
					&\le\frac{(c^2+\alpha^2)}{2\pi}\int_{-c}^{c}|G(\omega)|^2 {\rm d}\omega=(c^2+\alpha^2)\|g_0\|_2^2,
				\end{aligned} 
			\end{equation} where $\|\cdot\|_2$ denotes the $L^2$ norm. 
			Then we have \begin{align*}
				\|L|_{\hF_e}\| &= \sup_{\|g\|=1}\left\|L|_{\hF_e}g\right\|=\sup_{\|g_0\|_2=1}\|g^{(1)}e^{-\alpha t}\|_2
				\\
				& = \sup_{\|g_0\|_2=1}\|g_0^{(1)}+\alpha g_0\|_2\le \sqrt{c^2+\alpha^2}.
			\end{align*}
			%We can similarly prove that $\|U^\tau|_{\hF_e}\|<\infty$ for $\forall \tau>0$.
		\end{proof}	
		Based on this Proposition, the space of band-limited signal $g_0(t)$ is a special case of generator-bounded space ($\alpha=0$).
		
		\bl{\begin{remark}[Other examples of generator-bounded space]
				Besides the space of inverse Laplace-type signals, there are other types of infinite-dimensional generator-bounded space. For example, the Zakai's class of signals forms a generator-bounded space, i.e.,
				\begin{equation*}
					\begin{aligned}
						\hF(c,\delta) = \left\{g = g * h: g\in L^2\left((1+t^2)^{-1}\right), \right. \\ \left.\ h(t) = \frac{1}{2\pi} \int_{-\infty}^{\infty}H(\omega)e^{i\omega t}{\rm d}\omega\right\},
					\end{aligned}
				\end{equation*} where 
				\begin{equation*}
					H(\omega) = \left\{\begin{matrix}
						&1~ &|\omega|\le c,\\
						&1-\frac{|\omega|-c}{\delta}~& c< |\omega|\le c+\delta,\\
						&0~&|\omega|>c+\delta.
					\end{matrix}
					\right.
				\end{equation*}% with the norm $\|g\| = \sqrt{\int_{-\infty}^{\infty} \frac{|g(t)|^2}{1+t^2}{\rm d}t}$.
				The proof can be found in the Appendix \ref{otherexamples} (Proposition \ref{generalized_zakai}). Moreover, based on this result, its generalized version, i.e., \begin{equation*}
					\begin{aligned}
						\hF(c,\delta,\alpha) = \{g(t) = e^{\alpha t} g_0(t): g_0\in \hF(c,\delta),\alpha\in\bR\},
					\end{aligned}
				\end{equation*} can also be proved as a generator-bounded space in a similar manner by considering the norm $\|g\| = \sqrt{\int_{-\infty}^{\infty} \frac{|g(t)e^{-\alpha t}|^2}{1+t^2}{\rm d}t}$.% with the norm $\|g\| = \sqrt{\int_{-\infty}^{\infty} \frac{|g(t)e^{-\alpha t}|^2}{1+t^2}{\rm d}t}$.
		\end{remark}}
		
		\subsection{Sampling and reconstruction for inverse-Laplace of signals}\label{sec:reconstruction_inverse}
		In the following, we analyze the Koopman spectrum $\sigma(L|_{\hF_e})$ for the space $\hF_e$ given by \eqref{inverse}.
		\begin{proposition}\label{invers_spectrum}
			Consider the generator-bounded space $\hF_e$ given by \eqref{inverse}. The Koopan spectrum is $\sigma(L|_{\hF_e}) = \alpha+i[-c,c]$.
		\end{proposition} 
		\begin{proof}
			The proof will be given in Appendix \ref{invers_spectrum_proof}.
		\end{proof}
		According to on Proposition \ref{invers_spectrum} and Theorem \ref{thKs}, the sampling bound of the inverse Laplace-type signal can be directly obtained, i.e., $\omega_s>2 c$ (rad/s). %Moreover, we can find that the real part of the Koopman spectrum $\alpha$ characterizes the exponential growth of the amplitude of $g(t)$. And its imaginary part describes the signal's oscillation frequency, which causes the possibility of aliasing. 
		In fact, the imaginary part of $\sigma(L|_{\hF_e})$ is closely related to the Fourier frequency. In particular, when $\alpha = 0$, the inverse Laplace-type signal is reduced to the non-periodic band-limited signal $g(t) = \frac{1}{2\pi}\int_{-c}^c G(\omega)e^{i\omega t}{\rm d}\omega$, where $G(\omega)\in L^2[-c,c]$. In this case, the Koopman spectrum is restricted to the imaginary axis and is consistent with the Fourier spectrum, i.e., $\sigma(L|_{\hF_e})= i[-c,c]$. Then the conclusion of Theorem \ref{thKs} is consistent with Nyquist-Shannon sampling theorem.
		
		%\subsection{Reconstruction formula}	
		Now we proceed to present the reconstruction formula of this type of signals from the Koopman operator perspective.
		
		\begin{proposition}\label{invers_recons}
			Consider an inverse Laplace-type signal $g(t) = \frac{1}{2\pi}\int_{-c}^{c}G(\omega)e^{(\alpha+i\omega) t}{\rm d}\omega$, where $G(\omega)\in L^2[-c, c],c>0,\alpha\in\bR$. Given $T_s$ satisfying Theorem \ref{thKs} (i.e., $T_s<\pi/c$), we have \begin{equation}\label{inverse_recons}
				g(\tau) = \sum_{k=-\infty}^{\infty}g(kT_s) e^{\alpha(\tau-kT_s)}\frac{\sin(\pi/T_s(\tau-kT_s))}{\pi/T_s(\tau-kT_s)}.
			\end{equation}
		\end{proposition}
		\begin{proof}
			We have \begin{align}\label{Utg}
				U^\tau g(t) = g(t+\tau) =  \frac{1}{2\pi}\int_{-c}^{c}G(\omega)e^{(\alpha+i\omega) t}e^{(\alpha+i\omega) \tau}{\rm d}\omega.
			\end{align} The function $e^{i\tau \omega}$ can be expressed as a Fourier series on $[-\pi/T_s,\pi/T_s]$, i.e.,
			\begin{align}\label{eitau}
				e^{i\tau \omega} = \sum_{k=-\infty}^\infty c_k(\tau) e^{ikT_s\omega},
			\end{align} where the coefficient is $c_k(\tau) = \frac{\sin(\pi/T_s(\tau-kT_s))}{\pi/T_s(\tau-kT_s)}$. Since $[-c,c]\subset [-\pi/T_s,\pi/T_s]$ by Theorem \ref{thKs}, substituting \eqref{eitau} into \eqref{Utg} gives 
			\begin{align*} 
				g(t+\tau)&=\sum_{k=-\infty}^\infty  e^{\alpha(\tau-kT_s)}\frac{c_k(\tau)}{2\pi}\int_{-c}^{c}G(\omega)e^{(\alpha+i\omega) (t+kT_s)}  {\rm d}\omega\\
				%&= \sum_{k=-\infty}^\infty c_k(\tau)e^{\alpha(\tau-kT_s)}\frac{1}{2\pi} \int_{-c}^{c}G(\omega)e^{(\alpha+i\omega) (t+kT_s)}  {\rm d}\omega\\
				&= \sum_{k=-\infty}^\infty c_k(\tau)e^{\alpha(\tau-kT_s)}g(t+kT_s).
			\end{align*} Then the reconstruction formula can be obtained by letting $t=0$, i.e., $\forall \tau>0,$ \begin{equation*}
				g(\tau) = \sum_{k=-\infty}^{\infty} \frac{\sin(\pi/T_s(\tau-kT_s))}{\pi/T_s(\tau-kT_s)} e^{\alpha(\tau-kT_s)}g(kT_s).
			\end{equation*}  
		\end{proof}
		This proof is similar to Papoulis' proof \cite{marks2009handbook}, and the formula \eqref{inverse_recons} can be reduced to the classical form \eqref{bandlimited_recons} for band-limited signals ($\alpha = 0$). Moreover, the reconstruction formula \bl{of a generalized version of Zakai's class \cite{cambanis1976zakai} can also be derived similarly (see Appendix \ref{zakai})}. Based on Proposition \ref{invers_recons}, the truncation error of inverse Laplace-type signal is given as follows.

		\begin{proposition}[Truncation error]
			Consider $N>0,$ \bl{$T_s <\pi/c,$} and the truncation error of $g(t)$ given by \begin{equation}
				e(t) = g(t) - \sum_{n=-N}^{N}g(nT_s)\frac{e^{\alpha (t-nT_s)}\sin((\pi/T_s)(t-nT_s))}{(\pi/T_s)(t-nT_s)}.
			\end{equation} Let $K(t)$ be an integer nearest to the truncation error observation time $t$, i.e., $t/T_s-1/2 < K(t) < t/T_s + 1/2$, and $N_1$ and $N_2$ be the number of samples available to the left and to the right of $K(t)$ in the finite approximation. We have \begin{equation}\label{bound}
				|e(t)| \le \frac{2(Ec/\pi)^{1/2}e^{\alpha t}|\sin(\pi t/T_s )|}{\pi(\pi-cT_s)}\left(\frac{1}{N_1}+\frac{1}{N_2}\right ),
			\end{equation} where $E =  \int_{-\infty}^\infty |g_0(t)|^2{\rm d}t<\infty$.
		\end{proposition}
		\begin{proof}
			The inverse Laplace-type signal can be written as \begin{equation}\label{gg0}
				g(t) = e^{\alpha t}g_0(t),
			\end{equation} where $g_0(t) = \frac{1}{2\pi}\int_{-c}^c G(\omega)e^{i\omega t}{\rm d}t$ is a band-limited signal. It leads to $
			e(t) %& = e^{\alpha t}g_0(t) - \sum_{n=-N}^{N}e^{\alpha t}g_0(nT_s)\frac{\sin((\pi/T_s)(t-nT_s))}{(\pi/T_s)(t-nT_s)}
			= e^{\alpha t} e_0(t),$
			%\end{align} 
			where \begin{equation}
				e_0(t) = g_0(t) - \sum_{n=-N}^{N}g_0(nT_s)\frac{\sin((\pi/T_s)(t-nT_s))}{(\pi/T_s)(t-nT_s)}
			\end{equation} denotes the truncation error of $g_0(t)$ that has been well studied \cite{yao1966truncation, brown1969bounds}. Since $G(\omega)\in L^2[-c,c]$, the bound \eqref{bound} can be obtained by leveraging the bound of $e_0(t)$ under the constraints on signal energy $E =  \int_{-\infty}^\infty |g_0(t)|^2{\rm d}t<\infty$ \cite{yao1966truncation}.
		\end{proof}
		
		\bl{It should be noted that there are still open problems concerning the tight bounds of truncation error. In particular, when $T_s<\pi/c$, the truncation error can be further reduced by introducing sampling windows $f(t)$ that have been investigated for band-limited $g_0(t)$ \cite{pawlak1996recovering}. The reconstruction formula then becomes \begin{equation*}
				g(t) = \sum_{n=-N}^{N}g(nT_s)\frac{e^{\alpha (t-nT_s)}\sin((\pi/T_s)(t-nT_s))}{(\pi/T_s)(t-nT_s)} f(t-kT_s),
			\end{equation*} where $f(t)$ is band-limited to $\pi/T_s-c$ and $f(0) = 1$. 
			To achieve fast convergence of the reconstruction error, extensive research has been conducted on $f(t)$, such as \cite{campbell1968sampling, jagerman1966bounds,%natterer1986efficient,
				knab1979interpolation}. }

		\section{Finite-dimensional signal space}\label{finitespace}
		Here we analyze an example of signals that belong to finite-dimensional $\hF_e$. Since polynomials and exponential functions can represent a wide class of functions, we consider linear combinations of polynomial and complex exponential signal $g(t) = \sum_{k=1}^{m}\sum_{l=0}^{b_k} a_{k,l} t^l e^{\lambda_k t}$ where $a_{k,l}\in\bR, a_{k,b_k}\neq 0, \lambda_k\in\bC,b_k\in\bN$. 
		\subsection{Polynomial and exponential signals}\label{sec:poly_exp}
		By the definition of $g(t)$, it belongs to the space $%\begin{equation*}\begin{aligned}
		\sspan\{g_{1,0}(t) = e^{\lambda_1 t},g_{1,1}(t) = te^{\lambda_1 t},\ldots, g_{1,b_1}(t) = %\\& t^{b_1}e^{\lambda_1 t},
		g_{2,0}(t) = e^{\lambda_2 t},\ldots, %g_{2,b_2}(t) = t^{b_2}e^{\lambda_2 t},\ldots,%g_{m,0}(x) = e^{\lambda_m x},\ldots,
		g_{m,b_m}(t) = t^{b_m}e^{\lambda_m t}$. 
		%\}. \end{aligned}\end{equation*} 
In the following, we prove that it is a generator-bounded space, of which periodic band-limited signals are a special case ($\Re(\lambda_k)=0, b_k=0,\forall k=1.\ldots,m$).

\begin{proposition}[Polynomial and exponential signals]\label{poly_exp}
	The space of polynomial and exponential signals \begin{equation}\label{poly_exp_signal}
		\begin{aligned}
			\hF_e = \sspan\{e^{\lambda_1 t}, te^{\lambda_1 t},\ldots, t^{b_1}e^{\lambda_1 t}, e^{\lambda_2 t},\ldots, %g_{2,b_2}(t) = t^{b_2}e^{\lambda_2 t},\ldots,%g_{m,0}(x) = e^{\lambda_m x},\ldots,
			t^{b_m}e^{\lambda_m t} \}. 
		\end{aligned}
	\end{equation} is a generator-bounded space.
\end{proposition}
\begin{proof}
	%By the definition of the signal, it can be proved that $\bl{g(t)}$ belongs to $\hF_e = \sspan\{g_{1,0}(t) = e^{\lambda_1 t},\ldots, g_{1,b_1}(t) = t^{b_1}e^{\lambda_1 t},g_{2,0}(t) = e^{\lambda_2 t},\ldots, g_{m,b_m}(t) = t^{b_m}e^{\lambda_m t} \}.$ 
	Firstly, the space $\hF_e$ is a Banach space because it is spanned by finite number of basis functions. Additionally, $\forall g\in\hF_e$, i.e., $g(t) = \sum_{k=1}^{m}\sum_{l=0}^{b_k}a_{k,l} t^{l}e^{\lambda_k t}$, we have \begin{equation*}
		U^{\tau}g(t) = \sum_{k=1}^{m}\sum_{l=0}^{b_k}\sum_{n=0}^{l}a_{k,l} C_l^n \tau^n e^{\lambda_k \tau} t^{l-n}e^{\lambda_k t}\in\hF_e.
	\end{equation*} Hence, it is an invariant space of $U^\tau$. Restricted in this finite-dimensional space, the generator is bounded because of its linearity.
\end{proof}

\subsection{Sampling and reconstruction for polynomial and exponential signals}
Then we analyze the associated Koopman spectrum to present the sampling bound by Theorem \ref{thKs}.
\begin{proposition}\label{poly_exp_spectrum}
	Consider the generator-bounded space $\hF_e$ given by \eqref{poly_exp_signal}. The Koopman spectrum is $\sigma(L|_{\hF_e}) = \{\lambda_k\}_{k=1}^m$.
\end{proposition}
\begin{proof}
	The proof will be given in Appendix \ref{poly_exp_spectrum_proof}.
\end{proof}

It follows \bl{from} Proposition \ref{poly_exp_spectrum} and Theorem \ref{thKs} that the sampling bound of $g(t)$ is $\omega_s>2 \max_k|\bIm(\lambda_k)|$ (rad/s). %Moreover, the eigenvalues of the Koopman operator $\{\lambda_k\}_{k}$ characterizes the exponential growth and oscillation components of $g(t)$. 
In particular, periodic band-limited signal is a special case of this signal class, i.e., $\forall k=1,\ldots,m, b_{k} = 0$ and $\lambda_k = i\beta_k, \beta_k\in\bR$. In this case, the Koopman (point) spectrum $\sigma(L|_{\hF_e})$ is restricted to the imaginary axis, i.e., $\sigma(L|_{\hF_e})= \{i \beta_k\}_{k=1}^{m}$, which is consistent with the Fourier frequency. Then we can also find that the conclusion of Theorem \ref{thKs} is consistent with the Nyquist-Shannon sampling theorem for these periodic band-limited signals. %Combined with the result of non-periodic band-limited signals in Section \ref{sec:reconstruction_inverse}, we conclude that band-limited signals are covered by this Koopman operator-based sampling theorem as a special case.

%The signal was also considered in \cite{palmieri1986sampling}  $\bl{g(t)} = \sum_{k=1}^{m}a_k t^{b_k}e^{\lambda_k t},a_k,\lambda_k\in\bC,b_k\in\bN$. 
%Then the restriction of the generator can be expressed by a big upper triangular matrix composed of $n$ upper triangular matrices, with the eigenvalues $\{\lambda_1,\ldots,\lambda_m \}$.
%Therefore, the sampling condition to avoid signal aliasing is $\omega>2\max_{k=1,\ldots,n}|\bIm(\lambda_k)|$ based on Theorem \ref{thKs}, which is consistent with the sampling condition of the imaginary part in \cite{palmieri1986sampling}. However, it also has requirement on the real part of the eigenvalues for the convergence of polynomial interpolation.

In the following, we proceed to present the formula represented by samples for this type of signals. 

\begin{proposition}\label{poly_exp_recons}
	Consider the polynomial and exponential signal $g(t) = \sum_{k=1}^{m}\sum_{l=0}^{b_k} a_{k,l} t^l e^{\lambda_k t},$ where $a_{k,l}\in\bR,a_{k,b_k}\neq 0,\lambda_k\in\bC,b_k\in\bN$, and $T_s$ satisfying Theorem \ref{thKs} (i.e., $T_s<\pi/\max_k|\bIm(\lambda_k)|$). Denote the basis $\Phi(t) = [t^{b_1} e^{\lambda_1 t},\ldots, e^{\lambda_1 t},\ldots, e^{\lambda_m t}]$ and the transition matrix $Q$ such that $[g(t),\ldots,g(t+(M-1)T_s)]= \Phi(t) Q$. We have \begin{equation}\label{poly_exp_recons_al}
		g(\tau) = [g(0), \ldots, g((M-1)T_s)] Q^{-1} \overline{U}_M^\tau Q \bm c(0),\end{equation} where $M = m+\sum_{k=1}^{m}b_k$, $\bm c(0) = [1,0,\ldots, 0]^T$, 
	%$Q = [\bm s(0),\bm s(T_s),\ldots,\bm s((M-1)T_s)],\bm s = [\bm s_1^T,\ldots,\bm s_m^T]^T, \bm s_k = [s_{k,0},\ldots,s_{k,b_k}]^{T}, s_{k,p}(\tau) =\sum_{l=b_k-p}^{b_k}C_{l}^{b_k-p}a_{k,l}\tau^{l-b_k+p}e^{\lambda_k \tau}$ for $\forall k=1,\ldots,m, p=0,\ldots,b_k$.
	\begin{equation}\label{matrix_koopman}
		\overline{U}_M^{\tau} = \left(\begin{matrix}
			\overline{U}_{b_1}^\tau&\\
			%&U_{b_2}^\tau&\\
			&\ddots&\\
			&&\overline{U}_{b_m}^\tau
		\end{matrix}
		\right),
	\end{equation} and $\overline{U}_{b_k}^\tau\in\bR^{(b_k+1)\times (b_k+1)}$ for $\forall k = 1,\ldots,m$, i.e.,
	\begin{equation*}
		\overline{U}_{b_k}^\tau = \left(\begin{matrix}
			e^{\lambda_k \tau}&&&\\
			C_{b_k}^1\tau e^{\lambda_k \tau}&e^{\lambda_k \tau}&&\\
			\vdots&\vdots&\ddots&\\
			C_{b_k}^{b_k}\tau^{b_k}e^{\lambda_k \tau}&C_{b_k-1}^{b_k-1}\tau^{b_k-1}e^{\lambda_k \tau}&\ldots&e^{\lambda_k\tau}
		\end{matrix}
		\right).
	\end{equation*} 
\end{proposition} 
%\begin{proposition}[Polynomial and exponential signals]\label{poly_exp_recons}For the signal $\bl{g(t)} = \sum_{k=1}^{m}\sum_{l=0}^{b_k} a_{k,l} t^l e^{\lambda_k t},$ where $a_{k,l}\in\bR,c_{k,b_k}\neq 0,\lambda_k\in\bC,b_k\in\bN$, we have \begin{equation}\label{poly_exp_recons_al}				g(\tau) = [g(0), \ldots, g((M-1)T_s)]T^{-1}\bm c(\tau),\end{equation} where $M = m+\sum_{k=1}^{m}b_k$, $\bm c = [\bm c_1^T,\ldots,\bm c_m^T]^T$, $\bm c_k = [c_{k,0},\ldots,c_{k,b_k}]^{T}$, $c_{k,p}(\tau) = \sum_{l=b_k-p}^{b_k}C_{l}^{b_k-p}a_{k,l}\tau^{l-b_k+p}e^{\lambda_k \tau}$ for $\forall k=1,\ldots,m, p=0,\ldots,b_k$, and $T$ denotes the transition matrix, i.e, \begin{equation*} [g(t),\ldots,g(t+(M-1)T_s)] = [t^{b_1} e^{\lambda_1 t},\ldots, e^{\lambda_1 t},\ldots, e^{\lambda_m t}]T. \end{equation*} \end{proposition} 
\begin{proof}
	The matrix representation of the Koopman operator $\overline{U}_M^{\tau}\in\bR^{M\times M}$ is given by \eqref{matrix_koopman} on the basis $\Phi(t)$, i.e., $$U^{\tau} \Phi  = \Phi \overline{U}_M^{\tau}.$$ The lower triangular matrix $\overline{U}_M^{\tau}$ has distinct eigenvalues $\{e^{\lambda_k T_s}\}_{k=1}^m$ when $T_s$ satisfies \eqref{eqcons}. Based on the basis $\Phi$, the coordinate of $g$ is $\bm a=[\bm a_1^T,\ldots, \bm a_m^T]^T$, where $\bm a_k^T = [a_{k,b_k},\ldots,a_{k,0}]$ for $\forall k=1,\ldots,m$. Then it can be proved that $\mathrm{rank}[\lambda_k I-\overline{U}_M^{T_s},\bm a]=M$. It follows \bl{from} Proposition \ref{criterion} that $\{g(t),\ldots,g(t+(M-1)T_s)\}$ is a set of basis functions and transition matrix $Q = [\bm s(0),\bm s(T_s),\ldots,\bm s((M-1)T_s)]$ is invertible, where $\bm s(\tau) = [\bm s_1^T,\ldots,\bm s_m^T]^T$, $\bm s_k = [s_{k,0},\ldots,s_{k,b_k}]^{T}$, $s_{k,p}(\tau) = \sum_{l=b_k-p}^{b_k}C_{l}^{b_k-p}a_{k,l}\tau^{l-b_k+p}e^{\lambda_k \tau}$ for $\forall k=1,\ldots,m, p=0,\ldots,b_k$. Hence, for $g(t) = [g(t),\ldots,g(t+(M-1)T_s)]\bm c(0)$, we have \begin{equation}
		\begin{aligned}
			U^\tau g(t) &= [g(t),\ldots,g(t+(M-1)T_s)] \overline{U}^\tau \bm c(0)\\
			&=[g(t),\ldots,g(t+(M-1)T_s)]Q^{-1}\overline{U}_M^\tau Q \bm c(0),
		\end{aligned}
	\end{equation} where $\bm c(0) = [1,0,\ldots,0]^T$ and $\overline{U}^\tau$ represents the matrix representation of $U^\tau$ on basis $[g(t),\ldots,g(t+(M-1)T_s)]$. 
	Then the result \eqref{poly_exp_recons_al} can be obtained by letting $t=0$.
\end{proof}

\section{Reconstruction method}\label{sec:method}
%Since Koopman operator-based sampling theorem is valid for all signals belonging to generator-bounded space $\hF_e$, 
To numerically illustrate the sampling theorem, we propose a Koopman operator-based reconstruction method (KR) with theoretical convergence in this section. \bl{This algorithm is developed based on the reconstruction formula \eqref{recon_formula}. While the derivation of explicit reconstruction formula depends on the specific signal type (e.g., Proposition \ref{invers_recons} and Proposition \ref{poly_exp_recons}), this algorithm does not rely on prior knowledge of the signal. Specifically, the values of $c_k(\tau),\tau>0$ in \eqref{recon_formula} are computed using the Koopman operator identified from available samples.}%The reconstruction algorithm is described in Section \ref{sec:description_method}, and the theoretical convergence result is provided in Section \ref{sec:theoguarantee}. 

\subsection{Description of the method}\label{sec:description_method}
%\ye{I would move this out as a separated section and link this to sufficiency}
This method is applicable under the assumption that time-delay functions can be basis functions of $\hF_e^\gamma$ as required in Theorem \ref{recons_requrie}, which ensures that the signal can be represented by its samples.
The main idea is to compute the evolution of $g(t)$ under the action of $U^\tau$. Inspired by Theorem \ref{recons_requrie}, we first lift data to functional space $\hF_M$ spanned by $M$ time-delay functions of $g(t)$. Then we reconstruct the signal $g(t)$ by approximating the CT Koopman operator. The steps are given as follows.

\subsubsection{Lift data to functional space} With the sampling period $T_s$, we have the DT values of the signal $\{g(kT_s)\}_{k=0}^{N-1}$. Consider the functional space $\hF_M$ spanned by $M$ time-delay functions, i.e., $\hF_M = \sspan\{g,U^{T_s}g,\ldots,U^{(M-1)T_s}g \}$, where $U^{kT_s}g(t) = g(t+kT_s)$. Then we construct the  Hankel matrices $X,Y\in\bR^{(N-M)\times M}$ $(N\ge2M)$ as follows.
\begin{align*}
	X &= \left(\begin{matrix}
		g(0),&\ldots,&g((M-1)T_s)\\
		\vdots&\ddots&\vdots\\
		g((N-M-1)T_s),&\ldots,&g((N-2)T_s)
	\end{matrix}\right),\\Y &= \left(\begin{matrix}
		g(T_s),&\ldots,&g(MT_s)\\
		\vdots&\ddots&\vdots\\
		g((N-M)T_s),&\ldots,&g((N-1)T_s)
	\end{matrix}\right).
\end{align*}

\begin{remark}[The choice of $M$]\label{choice}
	Theoretically, the number of time-delay basis functions $M$ should be equal to the dimension of $\hF_e^\gamma$ to form a set of basis. In numerical practice, $M$ can be approximated by applying the singular value decomposition (SVD) on \begin{align*}
		X_1 = \left(\begin{matrix}
			g(0),&\ldots,&g((K-1)T_s)\\
			\vdots&\ddots&\vdots\\
			g((N-K-1)T_s),&\ldots,&g((N-2)T_s)
		\end{matrix}\right),
	\end{align*} where $K$ is large enough and $N\ge2K$. %Noted that $(k+1)$th column vector of $X$ represents values of time-delay $g(t+kT_s)$ at specific times, i.e., $t=0,T_s,\ldots, (N-K-1)T_s$. 
	Since singular values indicate the linear independence of column vectors and each column vector of $X_1$ represents values of a time-delay function, the dimension of the functional space $\hF_M$ can be approximated by counting the number of singular values that do not seem to decay to zero. Specifically, we can set a hard thresholding of singular values (e.g., 1e-10 as in \cite{arbabi2017ergodic}), to count the number of non-zero singular values. Then we choose suitable $M$ such that the constructed matrix $[X]_{(N-M)\times M}$ also has the same number of non-zero singular values. 
\end{remark}

\subsubsection{Approximate the Koopman operator and its generator}
Here we identify the finite-rank approximation of the DT Koopman operator and its generator, which are represented by the following matrices:
%Based on Extended Dynamic Mode Decomposition \cite{williams2015data}, 
%The matrix representation of  $\widehat{U}^{T_s}_M(N)$ and $\widehat{L}_M(N)$ can be computed as
\begin{equation}\label{identified_koopman}\overline{U}^{T_s} = X^{\dagger}Y,\end{equation} 
\begin{equation}\label{log_matrix}
	\overline{L} = \log(\overline{U}^{T_s})/T_s.
\end{equation}
Then we approximate the matrix representation $\overline{U}^{\tau}$ of the CT Koopman operator $U^{\tau}$ on $\hF_M$ based on $\overline{L}$, i.e.,
\begin{equation}\label{exp_matrix}
	\overline{U}^{\tau} = \exp(\tau\overline{L}).
\end{equation} 
The $\log$ and $\exp$ functions in \eqref{log_matrix} and \eqref{exp_matrix} are calculated based on the definitions of exponential and logarithmic matrix functions \cite[Page 59]{petersen2008matrix}, i.e., \begin{align*}
	\log(\overline{U}^{T_s}) &\equiv \sum_{\bl{k=1}}^{\infty}\frac{(-1)^{k+1}}{k}(\overline{U}^{T_s}-I)^k,\\
	\exp(\tau \overline{L}) &\equiv \sum_{k=0}^\infty \frac{1}{k!}(\tau \overline{L})^k.
\end{align*}
\bl{In practice, the evaluation of matrix logarithm and exponential in \eqref{log_matrix} and \eqref{exp_matrix} requires efficient algorithms. In the field of numerical linear algebra, numerous approaches have been developed that are tailored to specific types of matrices \cite{moler2003nineteen, trefethen2022numerical, higham2008functions}. Here, we use the matrix logarithm algorithm described in \cite{%al2012improved,
		al2013computing} and the matrix exponential algorithm described in \cite{%higham2005scaling, 
		al2010new}. These algorithms are implemented in MATLAB as the functions \texttt{logm} and \texttt{expm}, respectively.}

\subsubsection{Reconstruct the signal}
Finally, we reconstruct the signal by the approximated Koopman operator. With the time-delay basis function, we have 
$$g(\tau) = U^\tau  [g(0),\ldots,g((M-1)T_s)]\bm e_1,$$ where $\bm e_1 = [1,0,\ldots,0]^T$. 
Then we reconstruct the signal by \begin{equation*}
	\hat{g}(\tau)= \bm{g}(0)(\overline{U}^\tau\bm{e}_1),
\end{equation*}where $\bm{g}(0) = [g(0), g(T_s),\ldots,g((M-1)T_s)]$.

\begin{remark}[Connection with Koopman operator-based method for system identification]
	This approach for signal reconstruction is analogous to the Koopman operator-based method of identifying nonlinear system \cite{mauroy2019koopman}. Both methods approximate the Koopman operator in the lifted functional space $\hF_M$. However, the basis function $g(t)$ in $\hF_M$ is unknown for signal reconstruction, which is intended to be recovered through the flow of time $S^\tau(t) = t+\tau$. In contrast, for system identification, all basis functions in $\hF_M$ can be represented analytically and but the flow is unknown. By lifting to the chosen $\hF_M$, the unknown flow $S^\tau(\bx)$ of system states $\bx$ can be recovered from samples $\{\bx(kT_s)\}_{k\in\bZ}$. 
\end{remark}

\bl{\begin{remark}[Connection with Prony's method \cite{hauer1990initial}]
		Both the KR method and Prony method can reconstruct signals composed of exponential functions, i.e., $g(t) = \sum_{k=1}^{M} A_k e^{\lambda_k t}$. To reconstruct this signal, these two methods require the same minimum number of samples, i.e., $2M$. However, the applicability of the KR approach extends beyond Prony's method, accommodating not only exponential signals but also polynomial ones. Specifically, Prony's method requires finding the roots of a polynomial, which is derived based on properties of exponential signals. In contrast, the KR method reconstructs them based on the linearity of the Koopman operator, which does not depend on specific signal properties and holds for all types of signals within generator-bounded spaces.
		%Hence, the KR method is applicable to a broader class of signals. 
\end{remark}}

\subsection{Theoretical convergence}\label{sec:theoguarantee}
%\ye{perhaps give some guarantee for this method?}		
In this section, we prove the convergence of the algorithm proposed in Section \ref{sec:description_method} in the optimal conditions%, i.e., 1) when the generator-bounded space $\hF_e^\gamma$ is infinite-dimensional, there are infinite number of basis functions of $\hF_M\subset\hF_e = L^2[0,T_{\rm max}]$, and infinite number of samples whose sampling frequency satisfying \eqref{eqcons}; 2) when $\hF_e^\gamma$ is finite-dimensional,  there are infinite number of samples whose sampling frequency satisfying \eqref{eqcons} and $\hF_M = \hF_e^\gamma$
. 

For brevity, we use $U^\tau:\hF_e^\gamma\to\hF_e^\gamma$ to denote the Koopman operator defined on $\hF_e^\gamma$, $L$ to denote its generator, $\widehat{U}^{T_s}_M(N):\hF_M\to\hF_M$ to denote the finite-approximation of Koopman operator from $N$ samples, $P_M$ to denote the projection operator onto $\hF_M$, and $U^{\tau}_M = P_M U^{\tau}|_{\hF_M}:\hF_M\to\hF_M$ to denote the projection of the Koopman operator onto $\hF_M$%, and $\widehat{U}^{T_s}_M(N):\hF_M\to\hF_M$ to denote the identified DT Koopman operator by \eqref{identified_koopman} from $N$ samples
. 

Here we show the representation of the reconstruction error by the Koopman operator. The signal value $g(\tau)$ can be represented as
$$g(\tau) = U^{\tau}g(0).$$
Since $g\in\hF_M$, the reconstructed signal is computed as
$$\hat{g}(\tau)= \exp\left(\frac{\tau}{T_s}\Log ~\widehat{U}^{T_s}_M(N)
\right) g(0).$$
Then the reconstruction error can be measured as 
\begin{equation}\label{error1}
	\left\|U^\tau g-\exp\left(\frac{\tau}{T_s}\Log~\widehat{U}^{T_s}_M(N)\right) g\right\|.
\end{equation}
%When $\hF_e^\gamma$ is infinite-dimensional, we consider $\hF_e^\gamma= L^2([0,T_{max}])$ with $l_2$-norm, where $T_{max}>0$, and we aim to prove the convergence as follows. \begin{equation}\label{convergoal}\lim_{M\to \infty}\lim_{N\to\infty}\left\|\left\{U^\tau-\exp\left(\frac{\tau}{T_s}\Log~\widehat{U}^{T_s}_M(N)\right)\right\}g \right\|=0.\end{equation}In particular, when $\hF_e^\gamma$ is finite-dimensional and $\hF_M=\hF_e^\gamma$, we aim to prove the convergence as follows.\begin{equation}\label{convergoal_finite}\lim_{N\to\infty}\left\|\left\{U^\tau-\exp\left(\frac{\tau}{T_s}\Log~\widehat{U}^{T_s}_M(N)\right)\right\}g \right\|=0.\end{equation}

To prove its convergence, we first show the continuity of exponential and logarithm of linear bounded operators, whose proofs are given in Appendix \ref{sec:exp} and Appendix \ref{sec:log}. 

\begin{lemma}[Continuity of operator exponential]\label{exponential}
	Consider bounded operators $A,A_n\in\hL(X).$ If $\lim_{n\to\infty}\|(A-A_n)g\|=0,~\forall g\in X$, then $$\lim_{n\to\infty}\|(\exp A-\exp A_n )g\|=0.$$
\end{lemma}

\begin{lemma}[Continuity of operator logarithm]\label{logarithm}
	Consider bounded operators $A,A_n\in\hL(X).$ If $\lim_{n\to\infty}\|(A-A_n)g\|=0,\forall g\in X$, then $$\lim_{n\to\infty}\|(\log A-\log A_n )g\|=0,$$ where the logarithm is defined as $$\log A = \frac{1}{2\pi i}\int_{+\partial \Omega}\log \lambda~ (\lambda I - A)^{-1}{\rm d}\lambda,$$ and $+\partial \Omega$ is a smooth, positively oriented boundary of $\Omega\subset\bC$ that $\sigma(A)\cup\sigma(A_n)\subset \Omega$. Noted that the argument of $\lambda\in+\partial\Omega$, i.e., $\arg \lambda$, is single-valued.
\end{lemma}

%Then we will use Lemma \ref{exponential} and Lemma \ref{logarithm} to 
Now we are ready to prove the convergence of the reconstruction over a finite time horizon.

\begin{theorem}[Convergence]\label{convergence}
	Assume $\{U^{kT_s}g\}_{k=0}^{K}$ being a set of basis functions of $\hF_e^\gamma$ for some $K$ ($K$ can be $\infty$). Consider $\hF_M = \sspan\{g,g(t+T_s),\ldots,g((M-1)T_s)\}$ and $g\in \hF_e^\gamma=\hL^2[0,T_{\rm max}]$, $\tau\in[0,T_{\rm max}]$, where $T_{\rm max}>0$. If the sampling frequency satisfying $\omega_s>2 \max|\bIm(\sigma(L))|$ (rad/s), we have\begin{equation*}
		\lim_{M\to \infty}\lim_{N\to\infty}\left\|\left\{U^\tau-\exp\left(\frac{\tau}{T_s}\Log~\widehat{U}^{T_s}_M(N)\right)\right\}g \right\|=0.
	\end{equation*} Particularly, when $\hF_M = \hF_e^\gamma$ and the sampling frequency $\omega_s>2 \max|\bIm(\sigma(L))|$ (rad/s), we have \begin{equation*}
		\lim_{N\to\infty}\left\|\left\{U^\tau-\exp\left(\frac{\tau}{T_s}\Log~\widehat{U}^{T_s}_M(N)\right)\right\}g \right\|=0.
	\end{equation*}
\end{theorem}
\begin{proof}
	When $\omega_s>2 \max|\bIm(\sigma(L))|$, it follows \bl{from} Lemma \ref{principal} that $LT_s$ is the principal logarithm of $U^{T_s}$, i.e., $$L= \frac{1}{T_s}\Log ~U^{T_s},$$ where $T_s = 2\pi/\omega_s$. 
	Since $\|U^\tau\|<\infty$, it follows that
	\begin{equation*}
		\begin{aligned}
			&\left\|\left\{U^\tau-\exp\left(\frac{\tau}{T_s}\Log~\widehat{U}^{T_s}_M(N)\right)\right\}g \right\|\\ 
			=& \left\|\left\{\exp\left(\frac{\tau}{T_s}\Log~U^{T_s}
			\right)-\exp\left(\frac{\tau}{T_s}\Log~\widehat{U}^{T_s}_M(N)\right)\right\} g\right\|\\
			\le& \left\|\left\{\exp\left(\frac{\tau}{T_s}\Log~U^{T_s}
			\right)-\exp\left(\frac{\tau}{T_s}\Log~(P_M U^{T_s}P_M)\right)\right\} g\right\|\\+&\left\|\left\{\exp\left(\frac{\tau}{T_s}\Log~(P_M U^{T_s}P_M)
			\right)-\exp\left(\frac{\tau}{T_s}\Log~U^{T_s}_M\right)\right\} g\right\|\\+&\left\|\left\{\exp\left(\frac{\tau}{T_s}\Log~U^{T_s}_M
			\right)-\exp\left(\frac{\tau}{T_s}\Log~\widehat{U}^{T_s}_M(N)\right)\right\} g\right\|.
		\end{aligned}
	\end{equation*}
	
	Here we show that these three terms in the last inequality tend to zero as $N\to\infty$ and $M\to\infty$ for $\hF_e^\gamma=\hL^2[0,T_{max}]$.
	Since the identified Koopman operator $\widehat{U}^{T_s}_M(N)$ conveges to $U^{T_s}_M$ as $N\to\infty$ for any norm \cite{korda2018convergence}, we have 
	$$\lim_{N\to \infty}\|U^{T_s}_Mg - \widehat{U}^{T_s}_M(N)g\|=0, ~\forall g\in\hF_M.$$ It follows \bl{from} Lemma \ref{logarithm} that $$\lim_{N\to\infty}\|\Log U^{T_s}_Mg-\Log \widehat{U}^{T_s}_M(N) g\|=0.$$
	By Lemma \ref{exponential} and $\tau<\infty$, we have \begin{equation*}
		\begin{aligned}
			&\lim_{N\to\infty}\left\|\left\{\exp\left(\frac{\tau}{T_s}\Log~U^{T_s}_M
			\right)-\exp\left(\frac{\tau}{T_s}\Log~\widehat{U}^{T_s}_M(N)\right)\right\} g\right\|\\&=0.
		\end{aligned}
	\end{equation*}
	
	Since the projection operator $P_M$ converges to identity operator $I$ in the strong operator topology as $M\to\infty$. So we have, for $\forall g\in\hF_M$,
	\begin{equation*}
		\begin{aligned}
			&\lim_{M\to \infty}\left\|(U^{T_s}
			-P_MU^{T_s}P_M) g\right\|\\ 
			\le &\lim_{M\to \infty}\left\{\left\|(U^{T_s}
			-P_MU^{T_s}) g\right\| +\left\|(P_MU^{T_s}
			-P_MU^{T_s}P_M) g\right\|\right\}\\
			= &\lim_{M\to \infty}\left\{\left\|(I-P_M)U^{T_s}g\right\| +\left\|P_MU^{T_s}(I-P_M)g\right\|\right\}=0.
		\end{aligned}
	\end{equation*}
	
	Additionally, it follows from $P_M g = g, U_M^{T_s} g = P_M U^{T_s}g$ for $\forall g\in\hF_M$ that
	\begin{equation*}\begin{aligned}
			&\lim_{M\to\infty}\left\|(P_M U^{T_s}P_M-U_M^{T_s}) g\right\|=%\lim_{M\to\infty}\left\| U_M^{T_s}(P_M-I) g\right\|=
			0.
	\end{aligned}\end{equation*}
	Based on Lemma \ref{exponential}--\ref{logarithm}, it follows that 
	\begin{equation*}
		\begin{aligned}
			\left\|\left\{\exp\left(\frac{\tau}{T_s}\Log~U^{T_s}
			\right)- 
			\exp\left(\frac{\tau}{T_s}\Log~(P_MU^{T_s}P_M)\right)\right\} g\right\|
		\end{aligned}
	\end{equation*}tends to zero as $M\to\infty$, and 
	\begin{equation*}
		\begin{aligned}
			\left\| \left\{\exp\left(\frac{\tau}{T_s}\Log~(P_M U^{T_s}P_M)
			\right)-\exp\left(\frac{\tau}{T_s}\Log~U^{T_s}_M\right)\right\} g\right\|
		\end{aligned}
	\end{equation*} tends to zero as $M\to\infty$.
	Therefore, we have \begin{equation*}
		\begin{aligned}
			\lim_{M\to \infty}\lim_{N\to\infty}\left\|\left\{U^\tau-\exp\left(\frac{\tau}{T_s}\Log~\widehat{U}^{T_s}_M(N)\right)\right\}g \right\|=0.
		\end{aligned}
	\end{equation*}
	Particularly, when $\hF_M = \hF_e^\gamma$, i.e., the dimension of $\hF_e^\gamma$ is $M$, we have $P_M = I$. It can be proved similarly that \begin{equation*}\begin{aligned}
			\lim_{N\to\infty}\left\|\left\{U^\tau-\exp\left(\frac{\tau}{T_s}\Log~\widehat{U}^{T_s}_M(N)\right)\right\}g \right\|=0.
		\end{aligned}
	\end{equation*}
	
\end{proof}

%Therefore, Theorem \ref{convergence} shows that the reconstruction over finite time horizon based on the Koopman operator is asymptotically exact when $\omega_s>2|\bIm(\sigma(L))|$ (rad/s) and $M,N\to\infty$ for $\hF_e^\gamma = L^2[0,T_{\rm max}]$, or $N\to\infty$ for finite-dimensional $\hF_M = \hF_e^\gamma$.	

\section{Numerical examples}\label{sec:num}
In this section, we illustrate Theorem \ref{thKs} numerically by reconstructing four types of band-limited and non-band-limited signals using  KR method given in Section \ref{sec:method}. Specifically, we aim to show that the signal can be successfully recovered when \eqref{eqcons} is satisfied; otherwise, signal aliasing exists. Hence, the reconstruction method needs to be effective even when the sampling frequency approaches the sampling bound, which is the reason why we choose the KR method. Its advantage of reconstructing signals when sampling is not fast enough will be shown by the comparisons with the methods of cubic spline interpolation, piecewise cubic Hermite interpolation (PCHIP), and polynomial curve fitting. Moreover, the reconstruction behavior in the presence of noise is also analyzed. In the following, we first present the examples of signals.
%Since this work focus on the reconstruction from perfect data, we would like to study the improvement of the reconstruction algorithm by considering noise and missing samples in the further research.

%\subsection{Numerical results}\label{sec:expm}
\subsubsection*{a) Band-limited signal}
\begin{equation*}
	g(t) = -\cos(2t)+\cos(0.5t+\pi/2)+1.5\cos(4t+\pi/3).
\end{equation*}

\subsubsection*{b) Non-band-limited exponential signal}
\begin{equation*}
	g(t) = e^{-t}\cos(4t+\pi/6)+e^{-0.5t}\cos(2t).
\end{equation*}

\subsubsection*{c) Non-band-limited polynomial signal}
\begin{equation*}
	%\bl{g(t)} = -t^2\cos(4t+\pi/3)+t\cos(2t).
	g(t) = t\cos(4t+\pi/3).
\end{equation*}

\subsubsection*{d) Non-band-limited exponential and polynomial signal}
\begin{equation*}
	g(t) = te^{-t}\cos(4t+\pi/3).
\end{equation*}

Although these signals can also be reconstructed by many other methods, we choose the KR method because it remains effective for signal reconstruction when the sampling period approaches the bound, which \bl{helps} illustrate Theorem \ref{thKs} numerically. To \bl{demonstrate} this advantage, we first show the comparison of reconstruction error with interpolation and fitting methods. The experimental setup of KR method is given in Table \ref{setup}, where the the number of basis functions is selected according to Remark \ref{choice}. The sampling condition of these signals is $\omega_s> 8$ rad/s and the critical sampling period is $T_\gamma \approx 0.785$s, as analyzed according to Proposition \ref{invers_spectrum}, Proposition \ref{poly_exp_spectrum} and Theorem \ref{thKs}. We collect samples with growing sampling periods that satisfy $T_s<T_\gamma$ to compare the behavior of reconstruction methods. Using the same samples, the reconstruction error\bl{s} are illustrated respectively in Fig. \ref{fig2}--Fig. \ref{fig5} for these four signals a)-d), where the top, middle, and bottom subgraphs represent the cases when the sampling periods are $0.2$s, $0.4$s, and $0.6$s, respectively. In these figures, blue, orange, yellow and purple lines represent the reconstruction error of cubic spline interpolation, PCHIP, polynomial curve fitting of degree $12$, and KR method, respectively. It shows that, compared to \bl{the} KR method, the reconstruction error\bl{s} of these three methods become non-negligible as the sampling period increases. In contrast, the KR method shows robustness against the increase of the sampling period. Hence, it can distinguish the error caused by the method itself and \bl{the error due to} signal aliasing when the sampling period exceeds the critical sampling period $T_\gamma$.

\begin{table*}[!htb]
	\caption{Setup of signal reconstruction}
	%\begin{ruledtabular}
	\centering
	\begin{tabular}{ccccc}
		\hline\hline
		\multirow{2}{*}{Signals}&\multirow{2}{*}{Critical sampling period (s)}& \multirow{2}{*}{Number of basis functions}&\multirow{2}{*}{Number of samples}&\\ %\cline{4-5} \cline{6-7}
		%{Signals}&{Critical sampling period (s)}&{Number of basis functions}&{Number of samples}& %\cline{4-5} \cline{6-7}
		\\\hline%&&& & 
		%\hline 
		a) Band-limited signal           & $\pi/4$   &   $6$   & $20$ & \\ %\hline
		b) Exponential signal               & $\pi/4$   &   $4$ & $20$ &    \\ %\hline
		c) Polynomial signal              & $\pi/4$  &    $6$ &  $20$ & \\ %\hline
		d) Exponential and polynomial    & $\pi/4$   &  $8$   & $20$ & \\ %\hline
		\hline\hline
	\end{tabular}
	%\end{ruledtabular}
	\label{setup}
\end{table*}

\begin{figure}[!t]%[thpb]
	\centering
	%\DIFdelbeginFL\DIFdelendFL \DIFaddbeginFL
	\subfloat[]{
		\label{Fig2.sub1}
		\includegraphics[width=.45\textwidth]{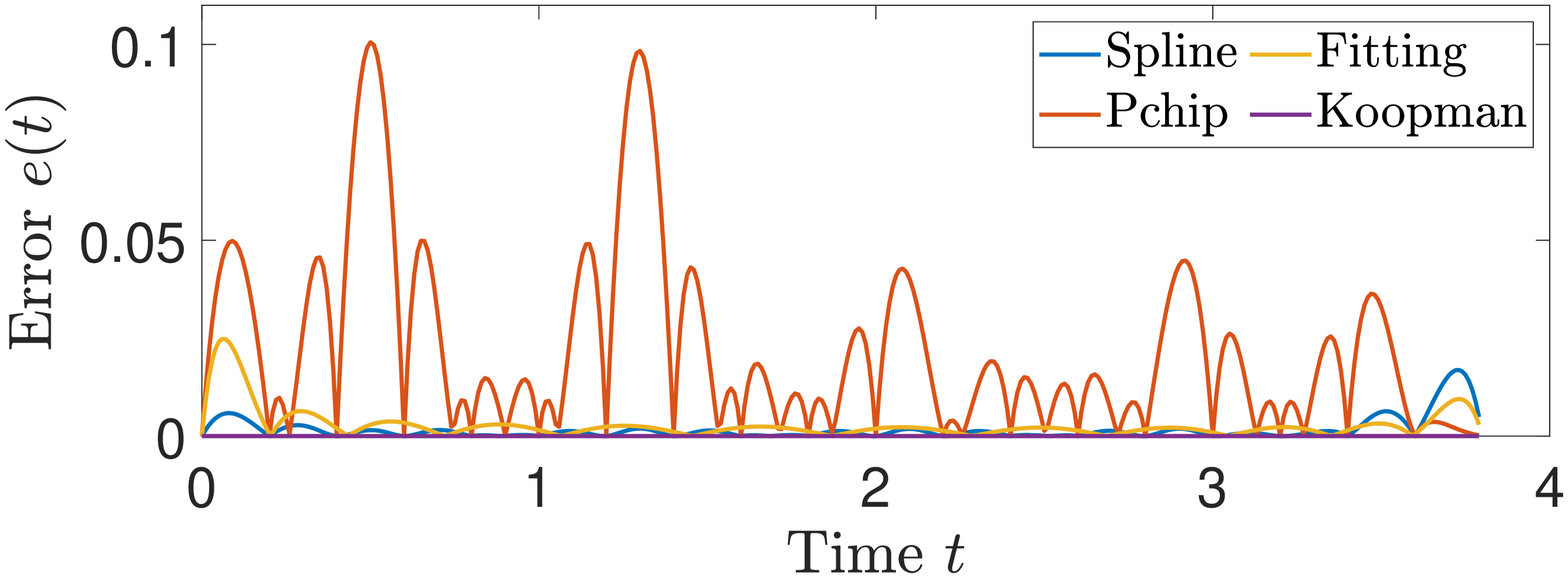}}\\
	\subfloat[]{
		\label{Fig2.sub2}
		\includegraphics[width=.45\textwidth]{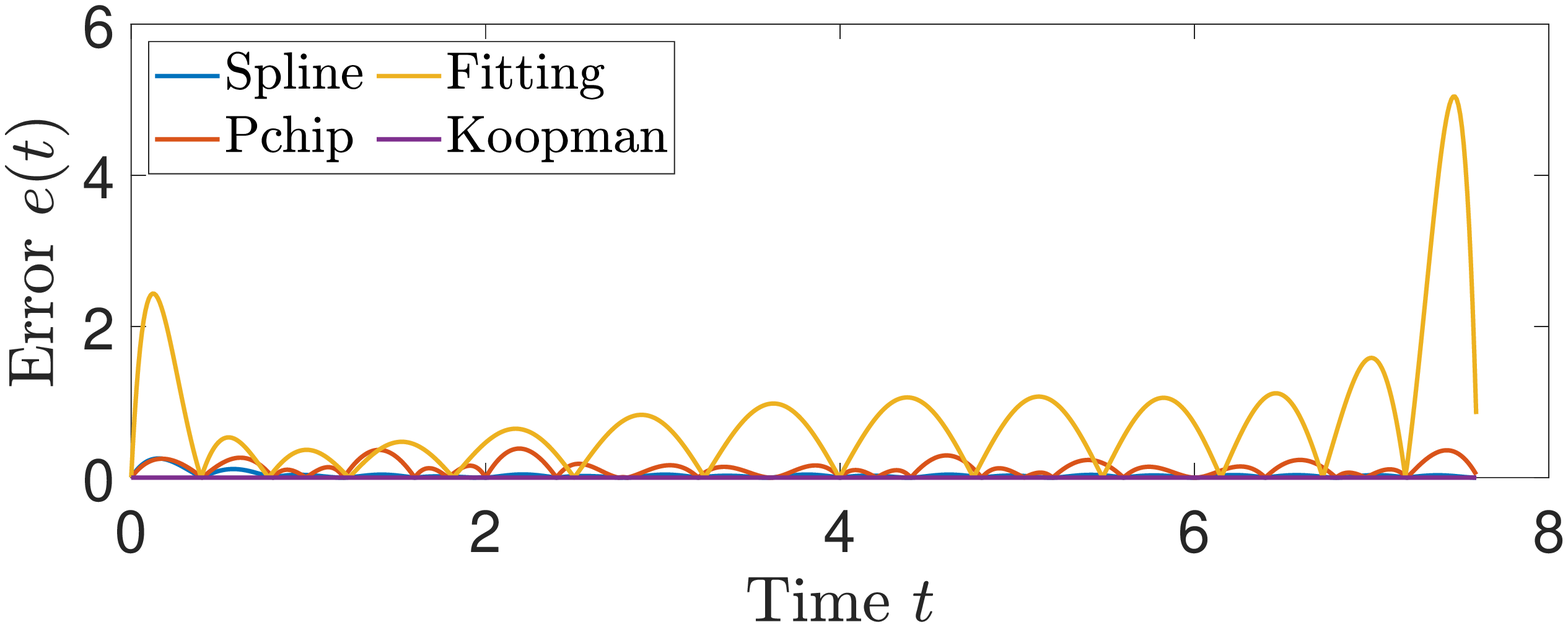}}\\
	\subfloat[]{
		\label{Fig2.sub3}
		\includegraphics[width=.45\textwidth]{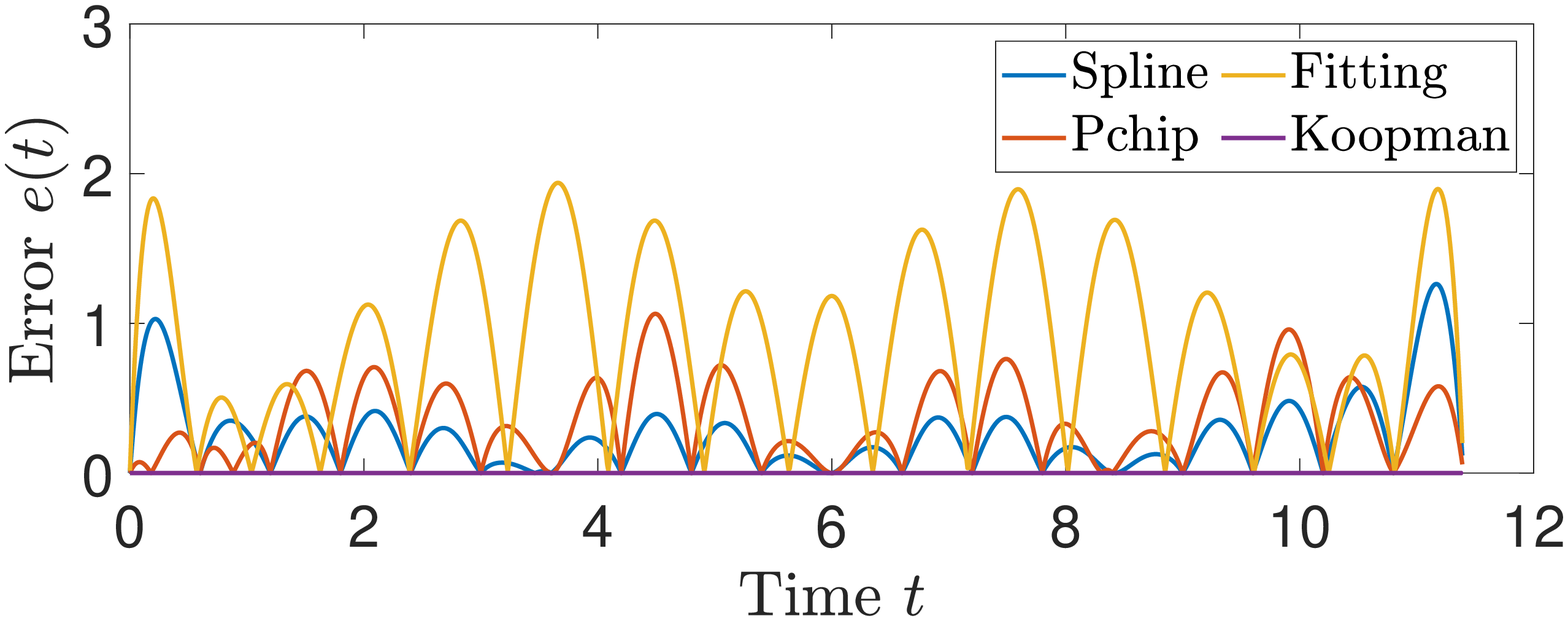}}\\
	%\subfloat[]{	\label{Fig2.sub3}\includegraphics[width=.42\textwidth]{err_fourier.eps}}
	\caption{Reconstruction of the band-limited signal from samples of sampling period $T_s = 0.2$s (a), $T_s = 0.4$s (b), and $T_s=0.6$s (c)%, and their reconstruction error (c)
		.}
	\label{fig2}
\end{figure}

\begin{figure}[thpb]%[!t]%
	\centering
	%\DIFdelbeginFL\DIFdelendFL \DIFaddbeginFL
	\subfloat[]{
		\label{Fig3.sub1}
		\includegraphics[width=.42\textwidth]{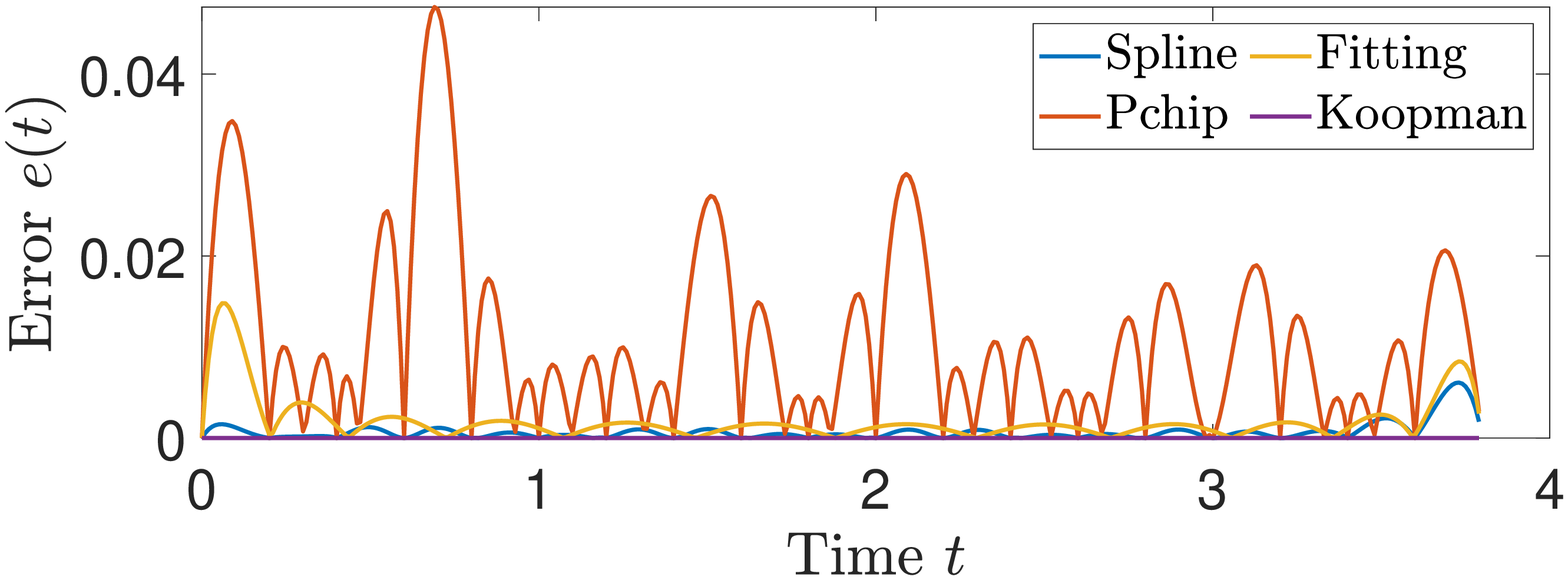}}\\
	\subfloat[]{
		\label{Fig3.sub2}
		\includegraphics[width=.42\textwidth]{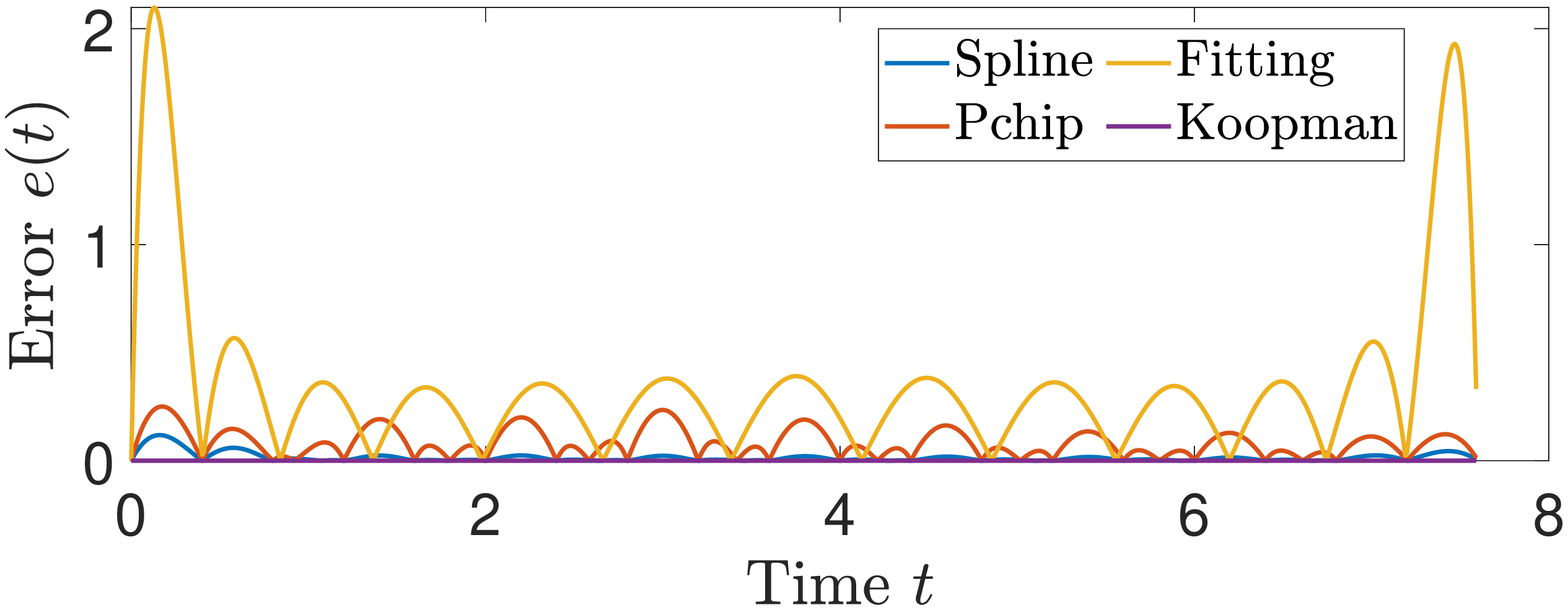}}\\
	\subfloat[]{
		\label{Fig3.sub3}
		\includegraphics[width=.42\textwidth]{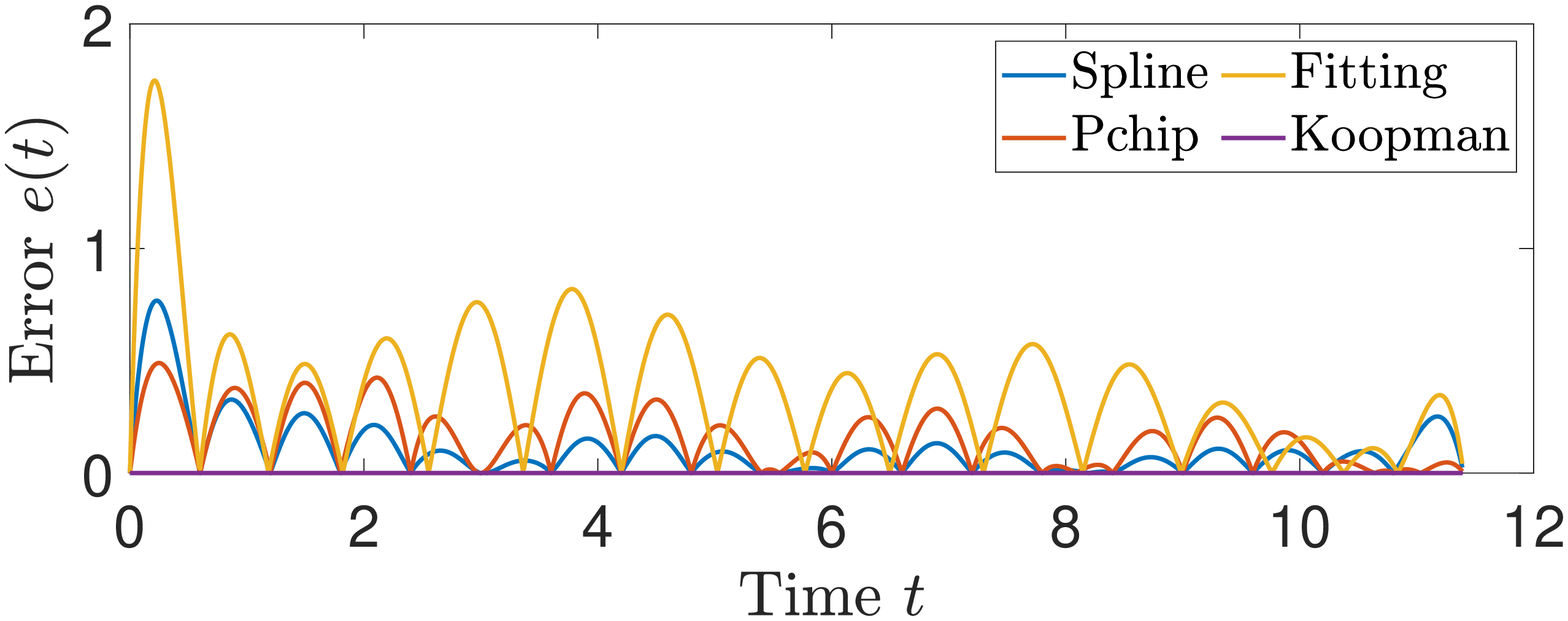}}\\
	%\subfloat[]{	\label{Fig2.sub3}\includegraphics[width=.42\textwidth]{err_fourier.eps}}
	\caption{Reconstruction of the signal with exponential growth from samples of sampling period $T_s = 0.2$s (a), $T_s = 0.4$s (b), and $T_s=0.6$s (c)}
	\label{fig3}
\end{figure}

\begin{figure}[thpb]
	\centering
	%\DIFdelbeginFL \DIFdelendFL \DIFaddbeginFL
	\subfloat[]{
		\label{Fig4.sub1}
		\includegraphics[width=.42\textwidth]{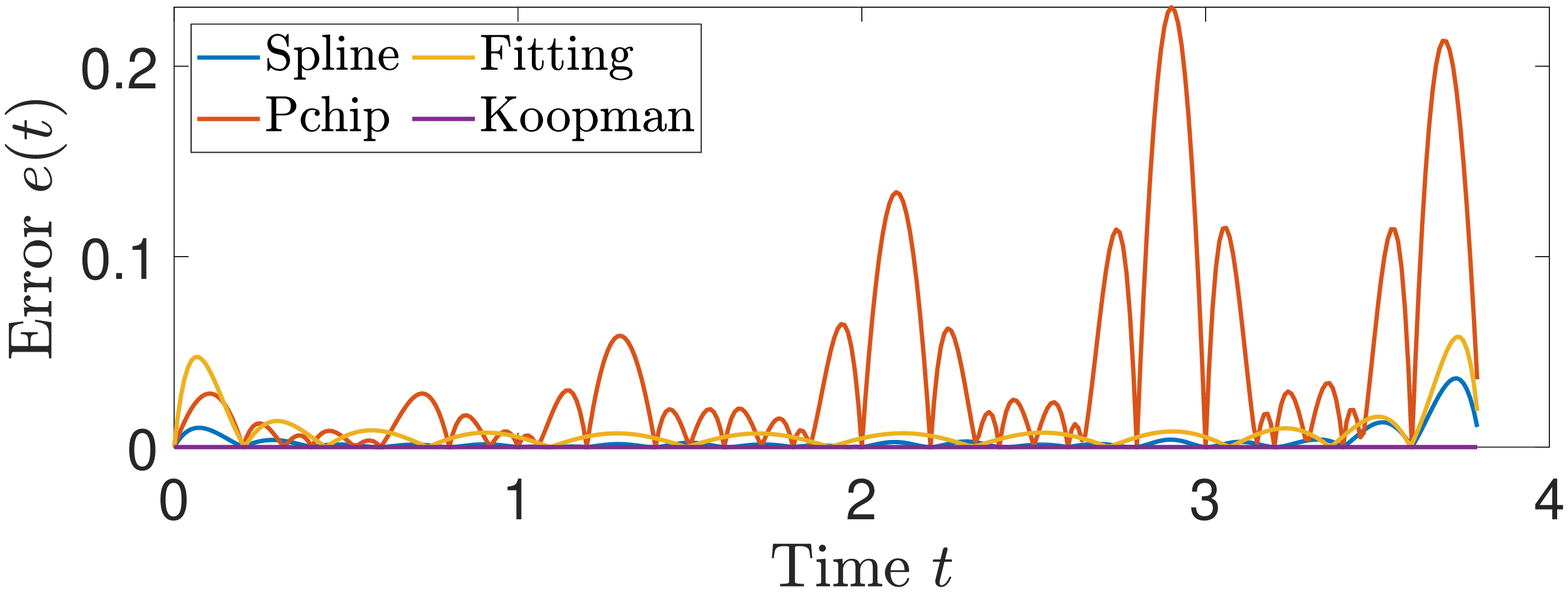}}\\
	\subfloat[]{
		\label{Fig4.sub2}
		\includegraphics[width=.42\textwidth]{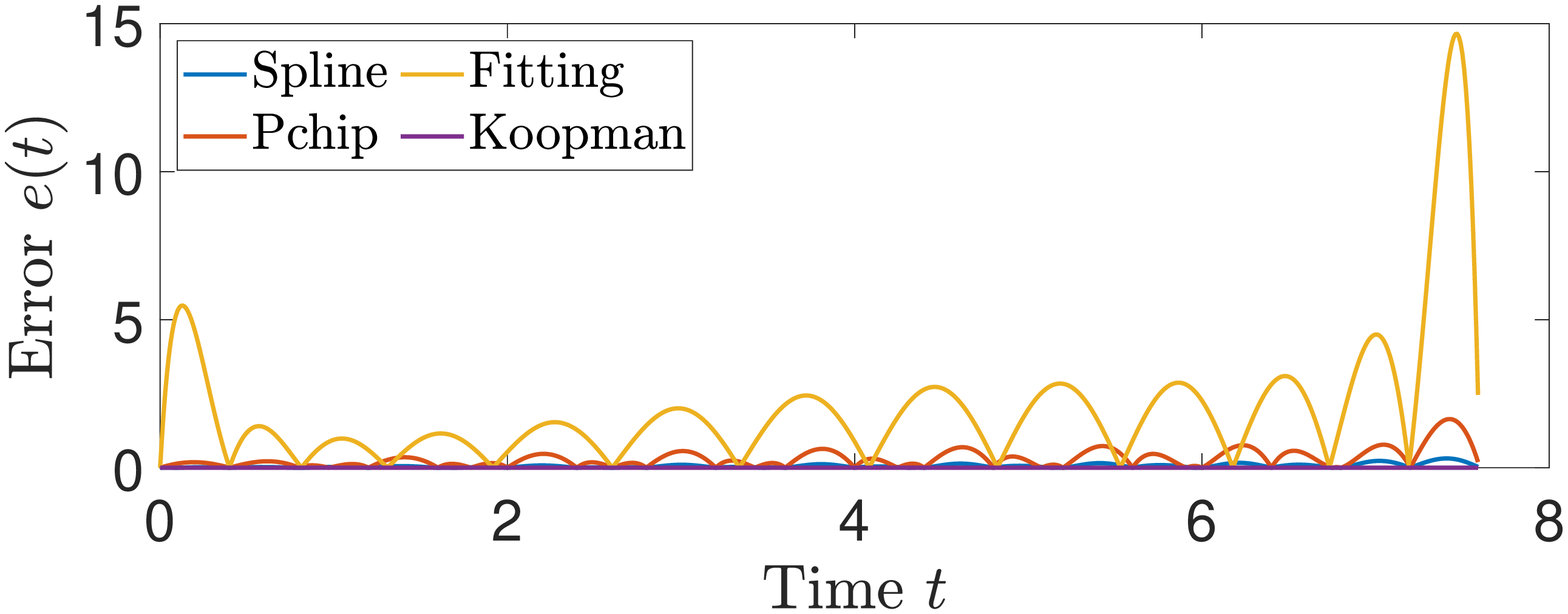}}\\
	\subfloat[]{
		\label{Fig4.sub3}
		\includegraphics[width=.42\textwidth]{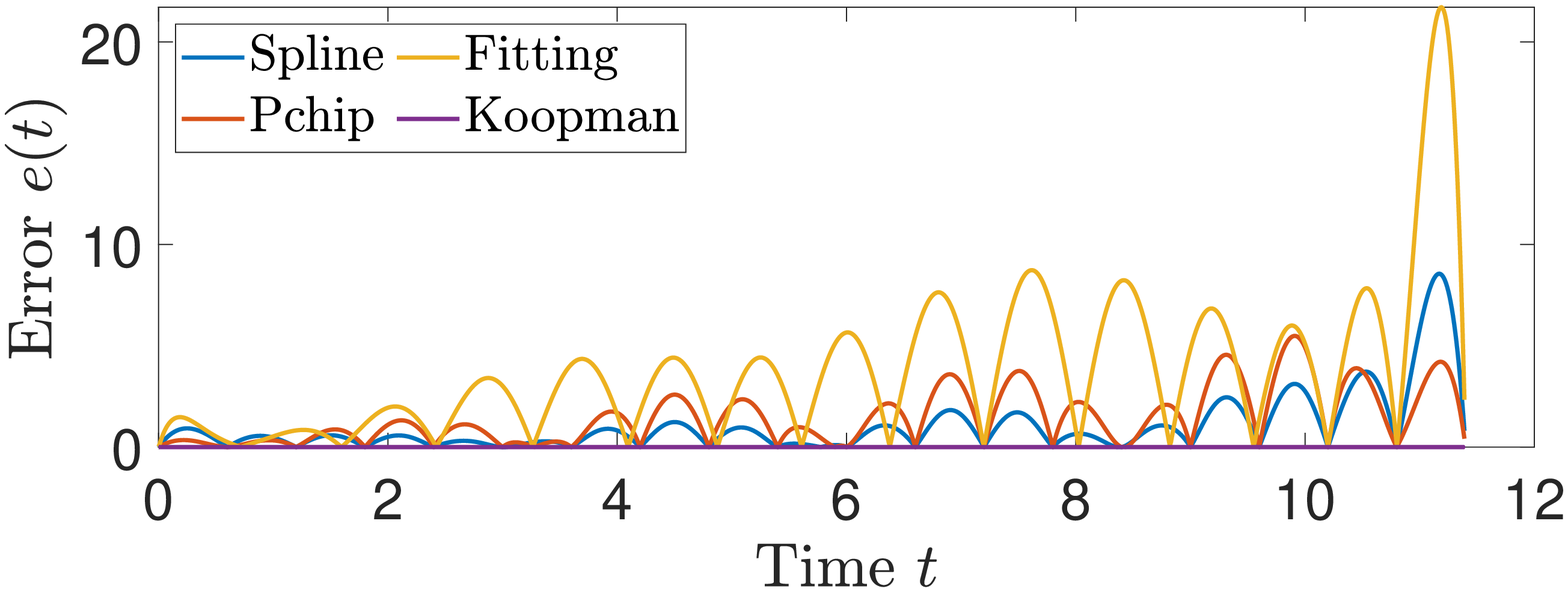}}\\
	%\subfloat[]{	\label{Fig2.sub3}\includegraphics[width=.42\textwidth]{err_fourier.eps}}
	\caption{Reconstruction of the signal with polynomial growth from samples of sampling period $T_s = 0.2$s (a), $T_s = 0.4$s (b), and $T_s=0.6$s (c)}
	\label{fig4}
\end{figure}

\begin{figure}[thpb]
	\centering
	%\DIFdelbeginFL \DIFdelendFL \DIFaddbeginFL
	\subfloat[]{
		\label{Fig5.sub1}
		\includegraphics[width=.42\textwidth]{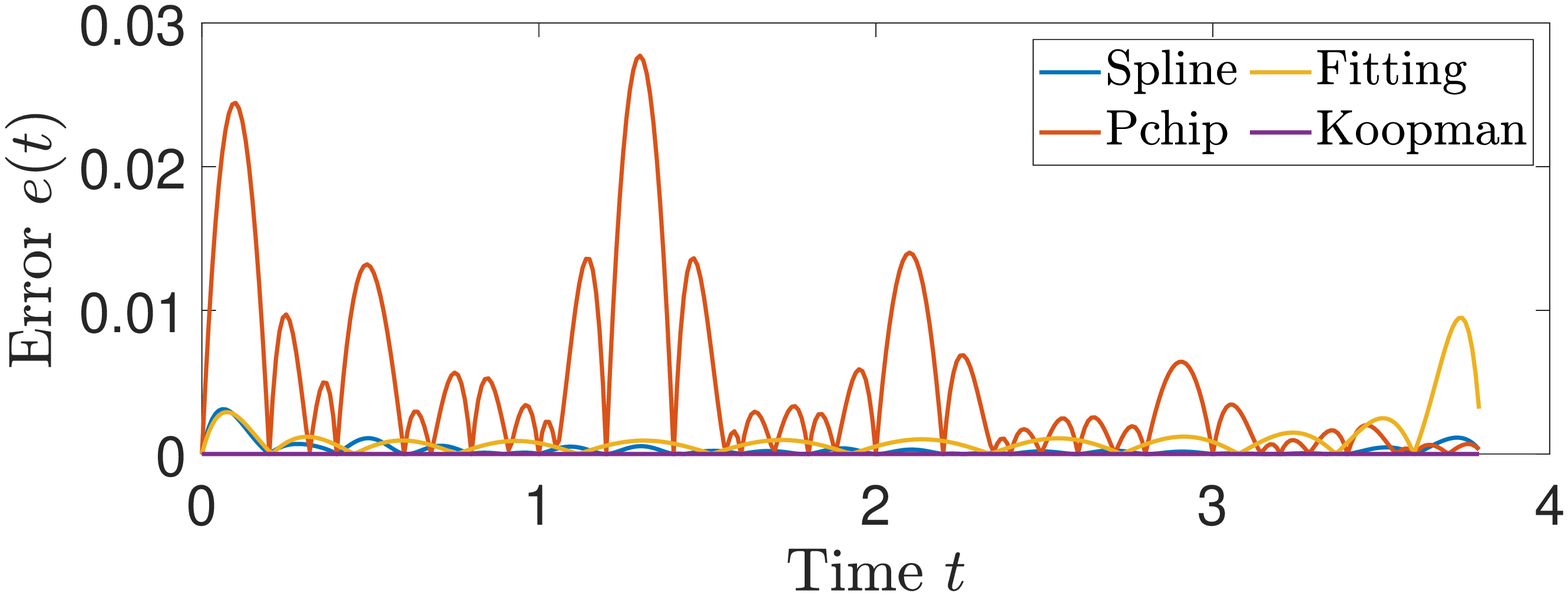}}\\
	\subfloat[]{
		\label{Fig5.sub2}
		\includegraphics[width=.42\textwidth]{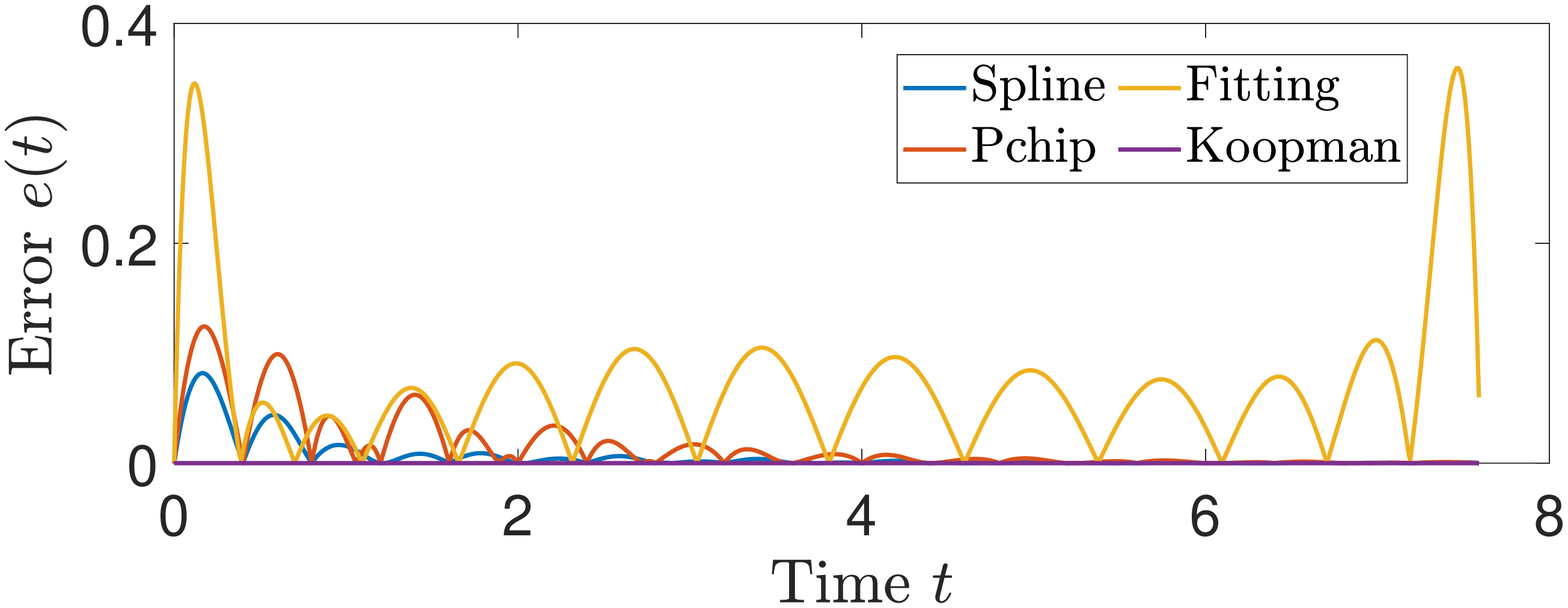}}\\
	\subfloat[]{
		\label{Fig5.sub3}
		\includegraphics[width=.42\textwidth]{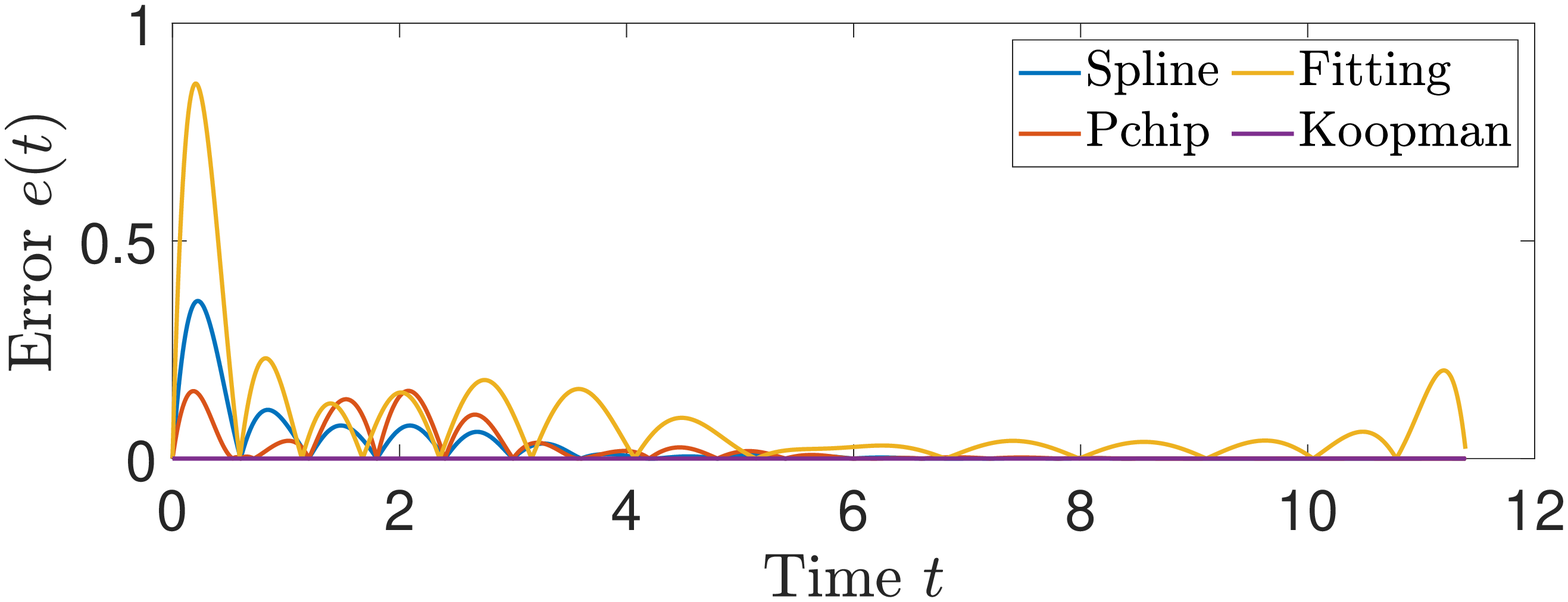}}\\
	%\subfloat[]{	\label{Fig2.sub3}\includegraphics[width=.42\textwidth]{err_fourier.eps}}
	\caption{Reconstruction of the signal with exponential and polynomial growth when sampling periods are $T_s = 0.2$s (a), $T_s = 0.4$s (b), and $T_s=0.6$s (c)}
	\label{fig5}
\end{figure}

Hence, we proceed to illustrate Theorem \ref{thKs} using the KR method. Fig. \ref{fig6}--Fig. \ref{fig9} show the results of signal reconstruction of these four signals, where the sampling periods are chosen to be $0.78$s and $0.79$s, both close to the sampling bound $T_\gamma \approx 0.785$s given by Theorem \ref{thKs}. In these figures, the samples, true signal $g(t)$, and reconstructed signal $\hat{g}(t)$ are denoted as triangles, blue lines, and orange dashed lines, respectively. It shows that, even though the sampling period $T_s$ is very closed to $T_\gamma$, the recovered signal is almost identical to the truth when $T_s <T_\gamma$ as shown in Fig. \ref{Fig6.sub1}--Fig.\ref{Fig9.sub1}. Moreover, the reconstruction fails when the sampling condition is not satisfied as shown in Fig. \ref{Fig6.sub2}--Fig. \ref{Fig9.sub2}. Specifically, when the sampling period $T_s=0.79$s \bl{slightly} exceeds $T_\gamma\approx0.785$, the reconstructed signal is distinctly different from the true signal but still passing through the samples. Therefore, signal aliasing exists when the sampling condition is not satisfied, which is consistent with Theorem \ref{thKs}.%.the reconstruction fails with signal aliasing 

%	Here we show errors $e_1(t)$, $e_2(t)$, which are computed as \begin{equation}\label{error}\begin{aligned}e_1(t) = |\bl{g(t)} - \hat{g}^{(1)}_s(t)|,\\	e_2(t) = |\bl{g(t)} - \hat{g}^{(2)}_s(t)|,	\end{aligned}	\end{equation}	where $\hat{g}^{(1)}_s(t)$ and $\hat{g}^{(2)}_s(t)$ denote reconstructed signal from samples of sampling period $T_s = 0.78$s and $T_s = 0.79$s, respectively. Then the error is shown in Fig. \ref{Fig2.sub3}. 

\begin{figure}[thpb]
	\centering
	\subfloat[]{
		\label{Fig6.sub1}
		\includegraphics[width=.44\textwidth]{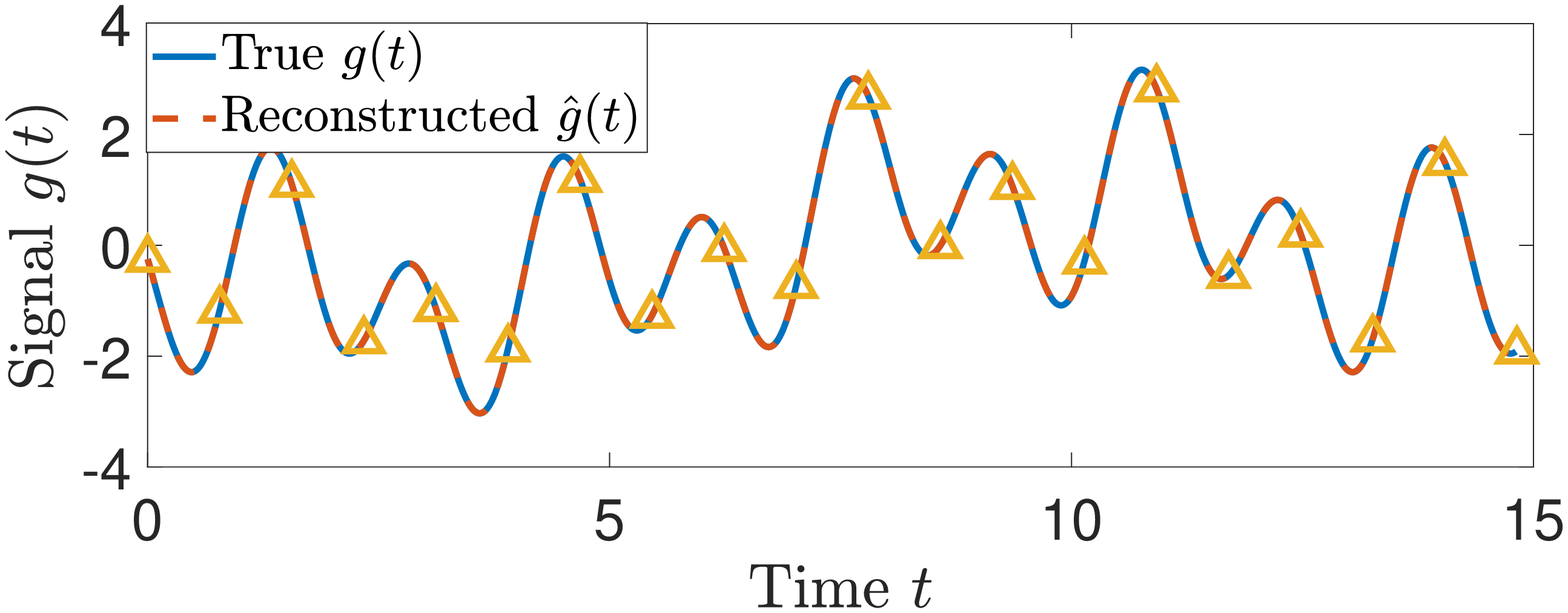}}\\
	\subfloat[]{
		\label{Fig6.sub2}
		\includegraphics[width=.44\textwidth]{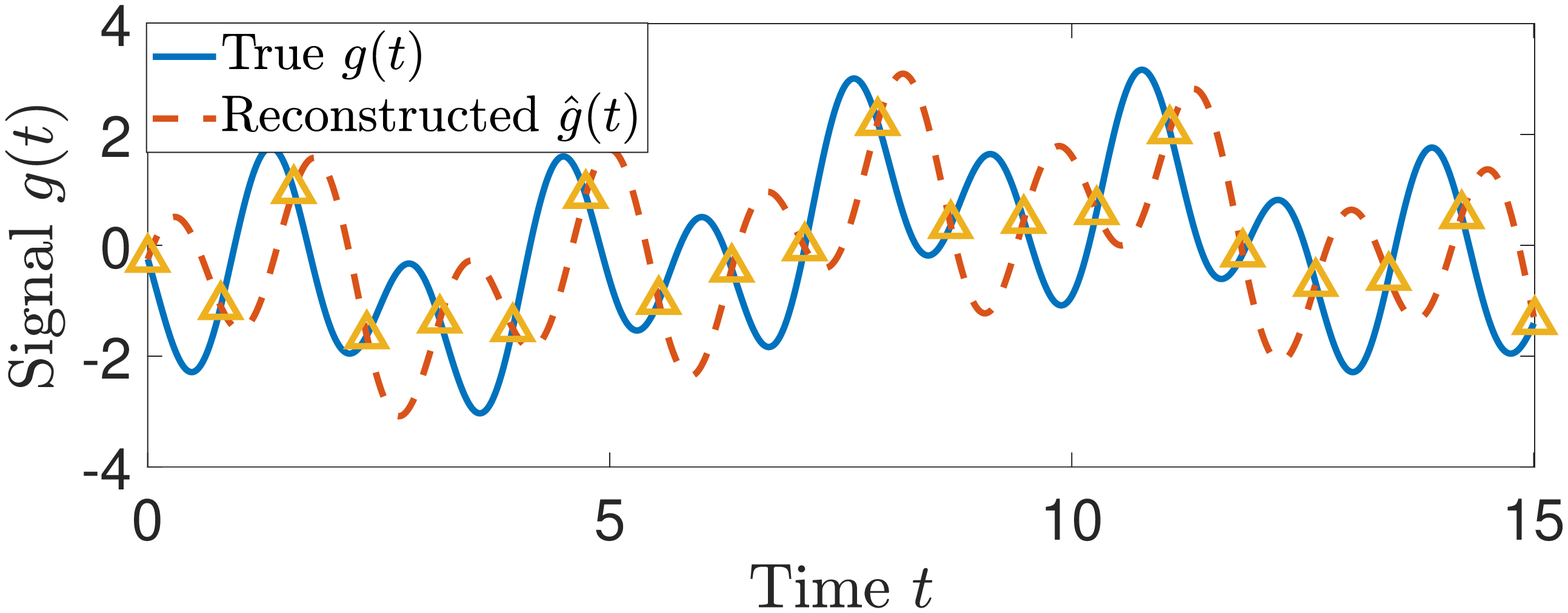}}\\
	%\subfloat[]{	\label{Fig2.sub3}\includegraphics[width=.42\textwidth]{err_fourier.eps}}
	\caption{Reconstruction of the band-limited signal from the samples of sampling period $T_s = 0.78$s (a), and $T_s=0.79$s (b)%, and their reconstruction error (c)
		.}
	\label{fig6}
\end{figure}
\begin{figure}[thpb]
	\centering
	\subfloat[]{
		\label{Fig7.sub1}
		\includegraphics[width=.44\textwidth]{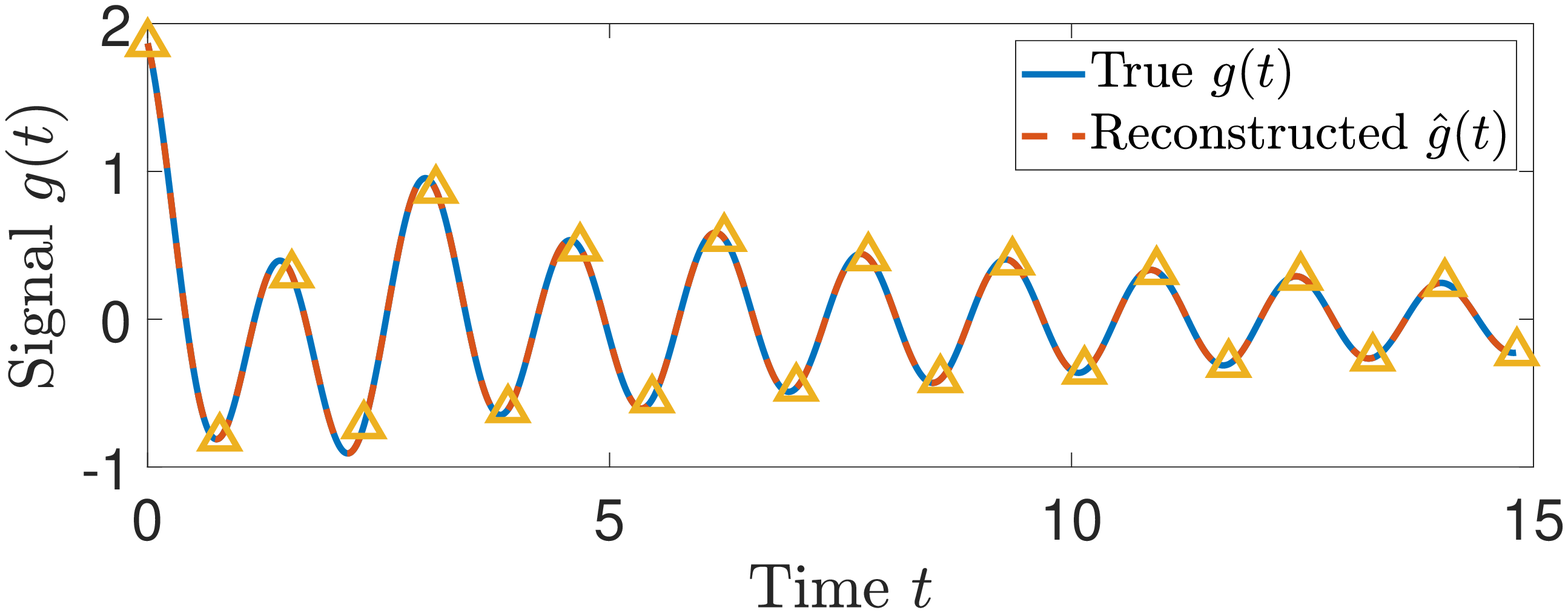}}\\
	\subfloat[]{
		\label{Fig7.sub2}
		\includegraphics[width=.44\textwidth]{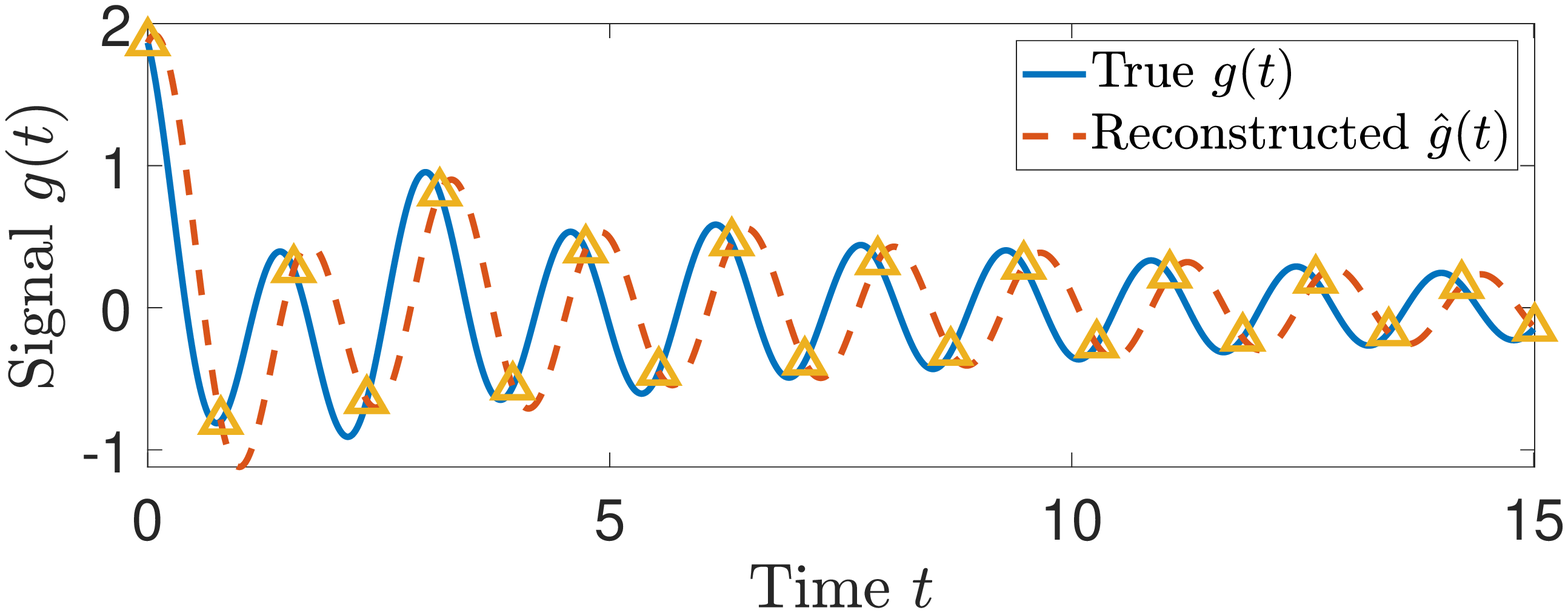}}\\
	%\subfloat[]{	\label{Fig2.sub3}\includegraphics[width=.42\textwidth]{err_fourier.eps}}
	\caption{Reconstruction of the signal with exponential growth from the samples of sampling period $T_s = 0.78$s (a), $T_s=0.79$s (b). }
	\label{fig7}
\end{figure}

\begin{figure}[thpb]
	\centering
	\subfloat[]{
		\label{Fig8.sub1}
		\includegraphics[width=.44\textwidth]{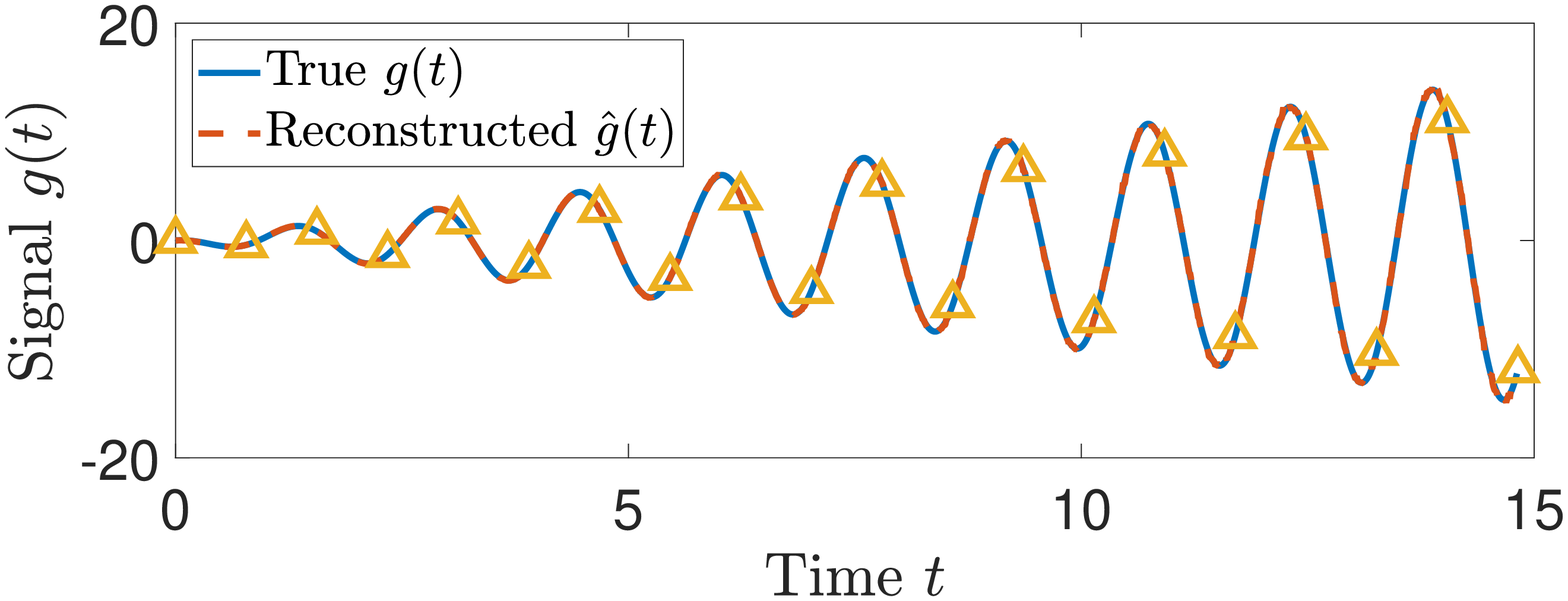}}\\
	\subfloat[]{
		\label{Fig8.sub2}
		\includegraphics[width=.44\textwidth]{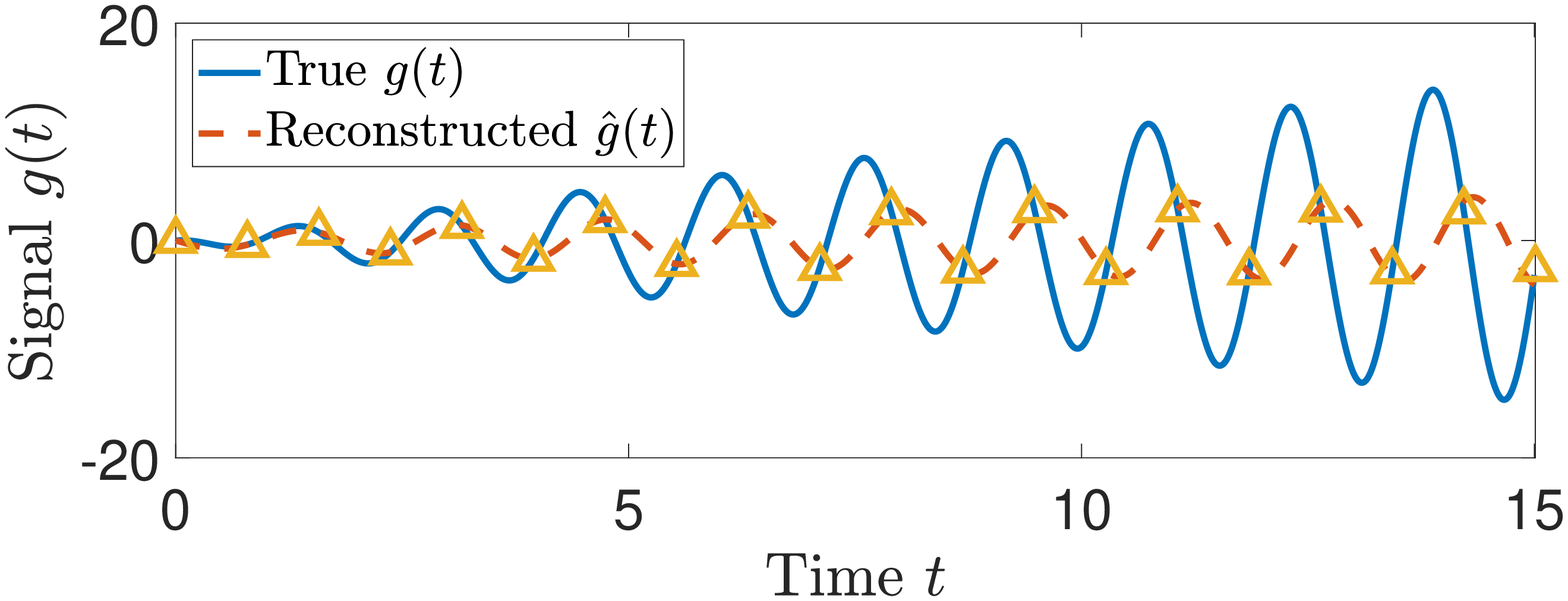}}\\
	%\subfloat[]{\label{Fig4.sub3}\includegraphics[width=.41\textwidth]{err_poly_fourier.eps}}
	\caption{Reconstruction of the signal with polynomial growth from the samples of sampling period $T_s = 0.78$s (a), $T_s=0.79$s (b). }
	\label{fig8}
\end{figure}

\begin{figure}[thpb]
	\centering
	\subfloat[]{
		\label{Fig9.sub1}
		\includegraphics[width=.44\textwidth]{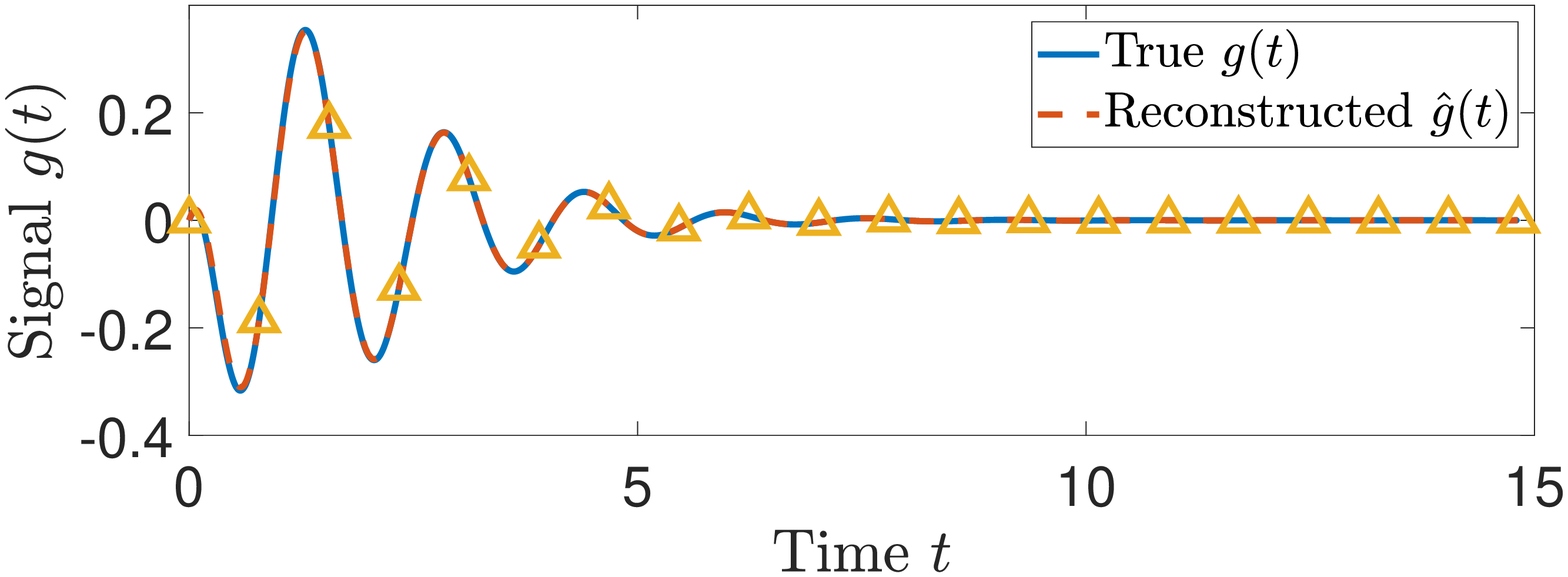}}\\
	\subfloat[]{
		\label{Fig9.sub2}
		\includegraphics[width=.44\textwidth]{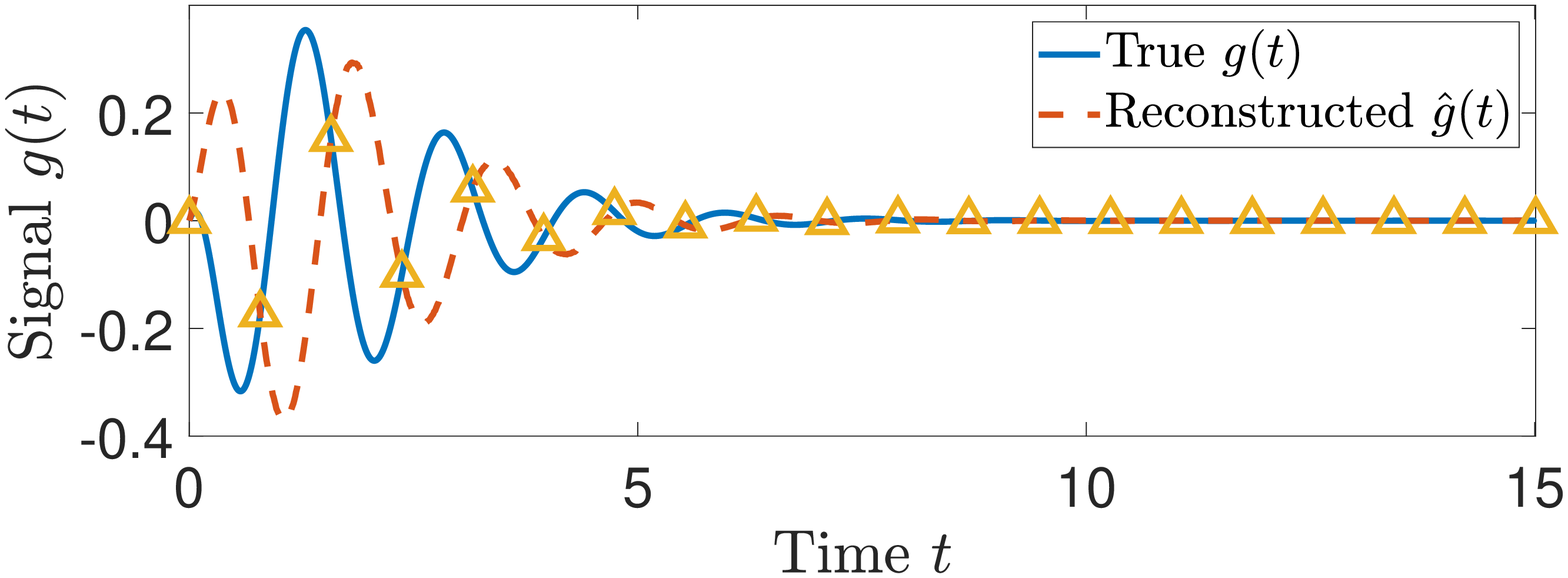}}\\
	%\subfloat[]{\label{Fig5.sub3}\includegraphics[width=.41\textwidth]{err_poly_exp_fourier.eps}}
	\caption{Reconstruction of the signal with polynomial and exponential growth from the samples of sampling period $T_s = 0.78$s (a), $T_s=0.79$s (b). }
	\label{fig9}
\end{figure}

However, samples are inevitably corrupted with noise in practice. To numerically analyze the behavior of KR method in this case, we reconstruct these four signals a)--d) from samples with white noise. To tolerate the perturbations in the KR method, we sample signals with a sampling period $T_s = 0.3$s. Moreover, we use $\tau\in[0,T_s)$ and all the available samples for the signal reconstruction, i.e, $[g(kT_s+\tau),\ldots,g((k+M-1)T_s+\tau)] = [g(kT_s),\ldots,g((k+M-1)T_s)]\overline{U}^\tau.$ The reconstruction results are shown in Fig. \ref{fig10}-- Fig.\ref{fig13} for these four signals, where the top, middle, and bottom subgraphs represent the cases when the signal-to-noise ratio ($\text{SNR}$) is $30,20$, and $10$, respectively. The grey, blue, and dashed orange line\bl{s} denote the noisy $g_n(t)$, clean $g(t)$, and reconstructed signal $\hat{g}(t)$, respectively. The triangles denote the noisy samples. These results show that \bl{the} KR method can reconstruct signals when $\text{SNR}=30$. Even when the noise is strong ($\text{SNR}=10, 20$), the reconstructed signal can still maintain the original signal\bl{'s} evolution trend. Hence, it suggests that this method is able to tolerate a certain level of white noise. \bl{Intuitively, this robustness can be attributed to the KR method's emphasis on capturing the overall characteristics of the signal, which helps mitigate the adverse effects of sample noise. Specifically, when a signal is recovered using this method, its amplitude growth and oscillation frequency are determined by the spectral properties of the Koopman generator, which are identified from the samples. As long as the noisy samples do not obscure the overall trend of the original signal, the identified generator's spectrum remains close to the truth. Consequently, the reconstructed signal retains characteristics similar to those of the original.}

\begin{figure}%[thpb]
	\centering
	%	\DIFdelbeginFL \DIFdelendFL \DIFaddbeginFL
	\subfloat[]{
		\label{Fig10.sub1}
		\includegraphics[width=.42\textwidth]{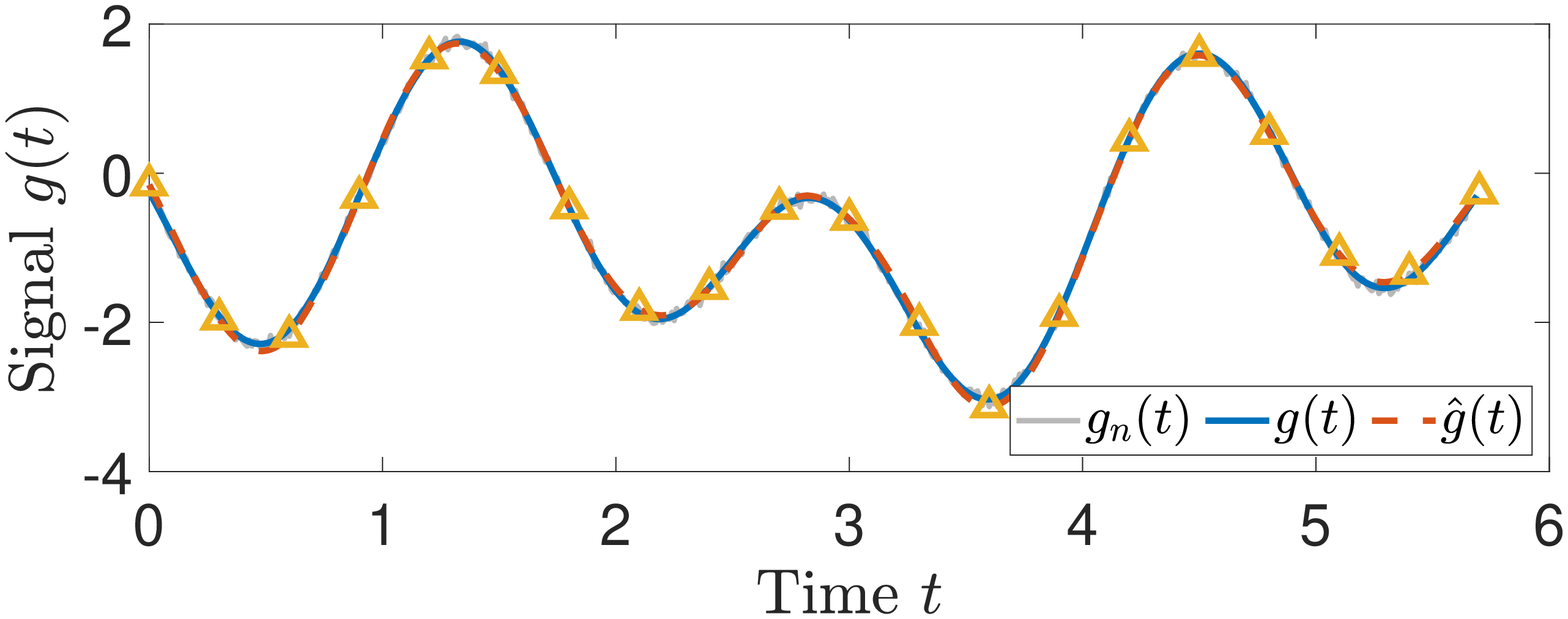}}\\
	\subfloat[]{
		\label{Fig10.sub2}
		\includegraphics[width=.42\textwidth]{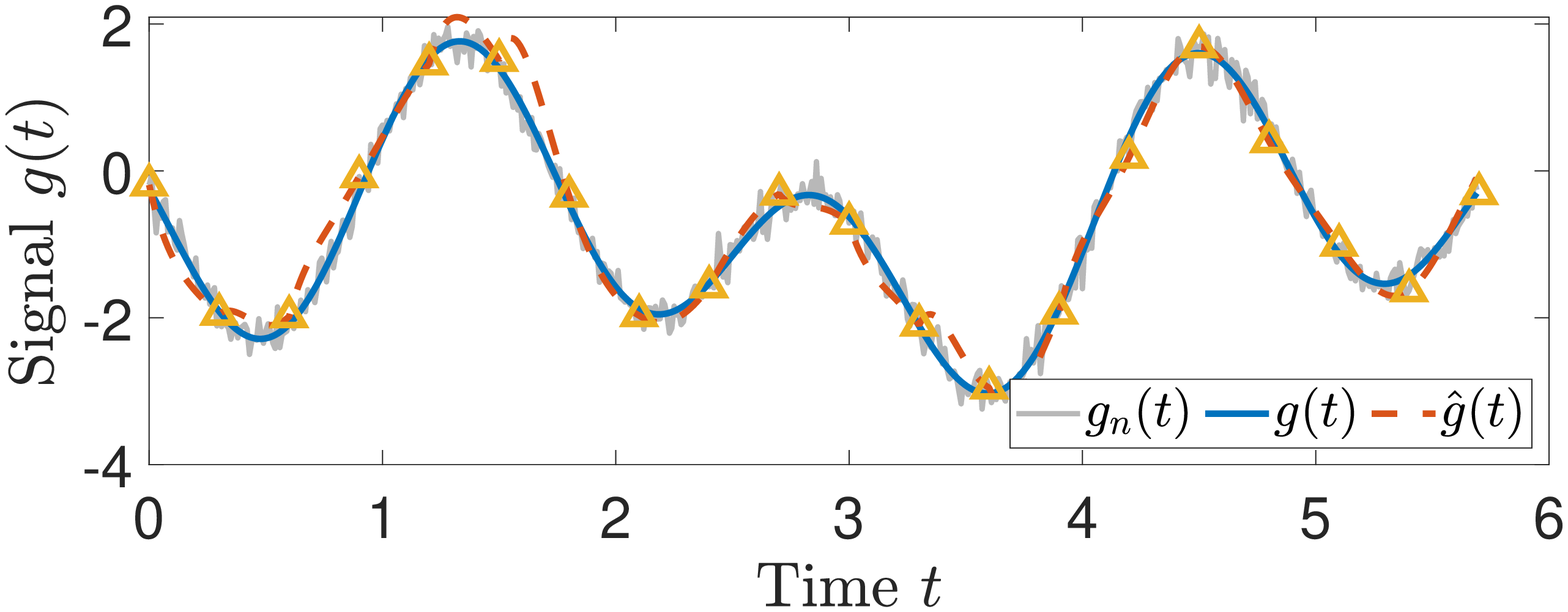}}\\
	\subfloat[]{\label{Fig10.sub3}
		\includegraphics[width=.42\textwidth]{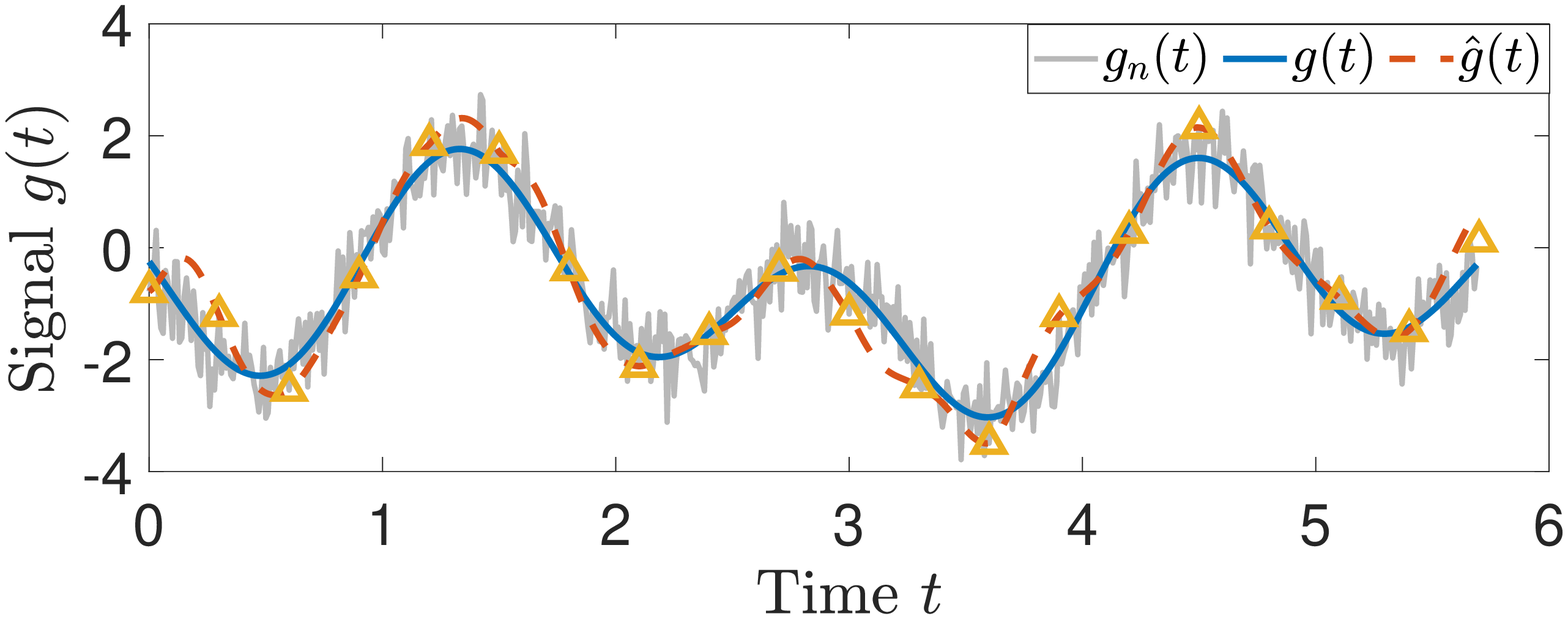}}
	\caption{The reconstruction of the band-limited signal by KR method in the presence of noise with $\text{SNR} = 30$ (a), $\text{SNR}=20$ (b), $\text{SNR}=10$ (c).}
	\label{fig10}
\end{figure}

\begin{figure}[thpb]
	\centering
	%		\DIFdelbeginFL\DIFdelendFL \DIFaddbeginFL
	\subfloat[]{
		\label{Fig11.sub1}
		\includegraphics[width=.42\textwidth]{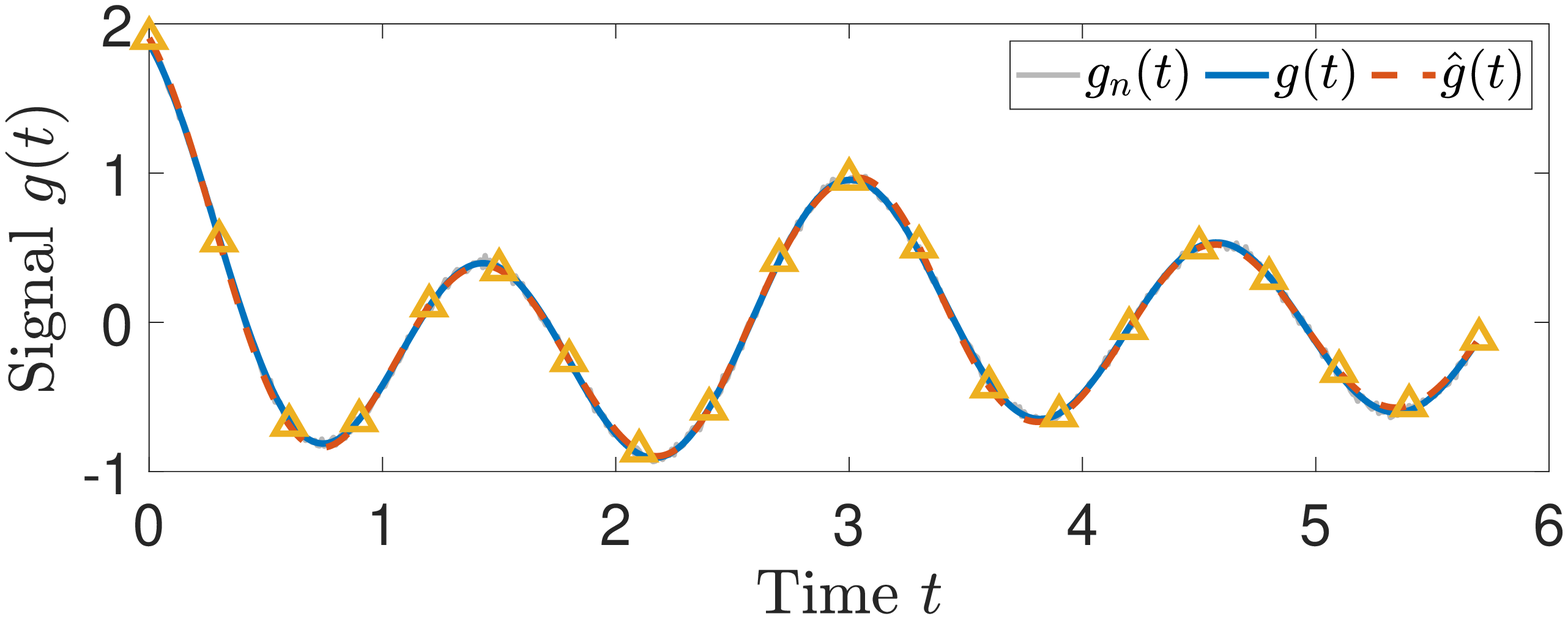}}\\
	\subfloat[]{
		\label{Fig11.sub2}
		\includegraphics[width=.42\textwidth]{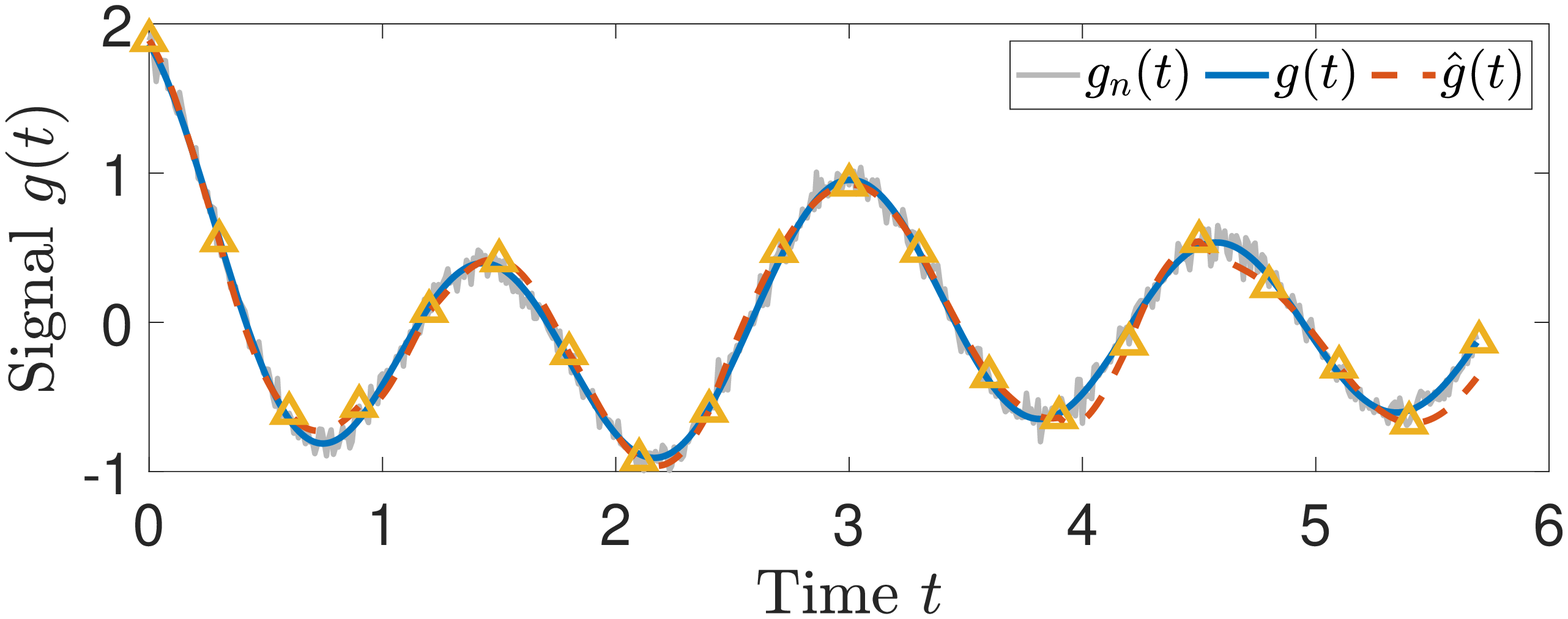}}\\
	\subfloat[]{\label{Fig11.sub3}
		\includegraphics[width=.42\textwidth]{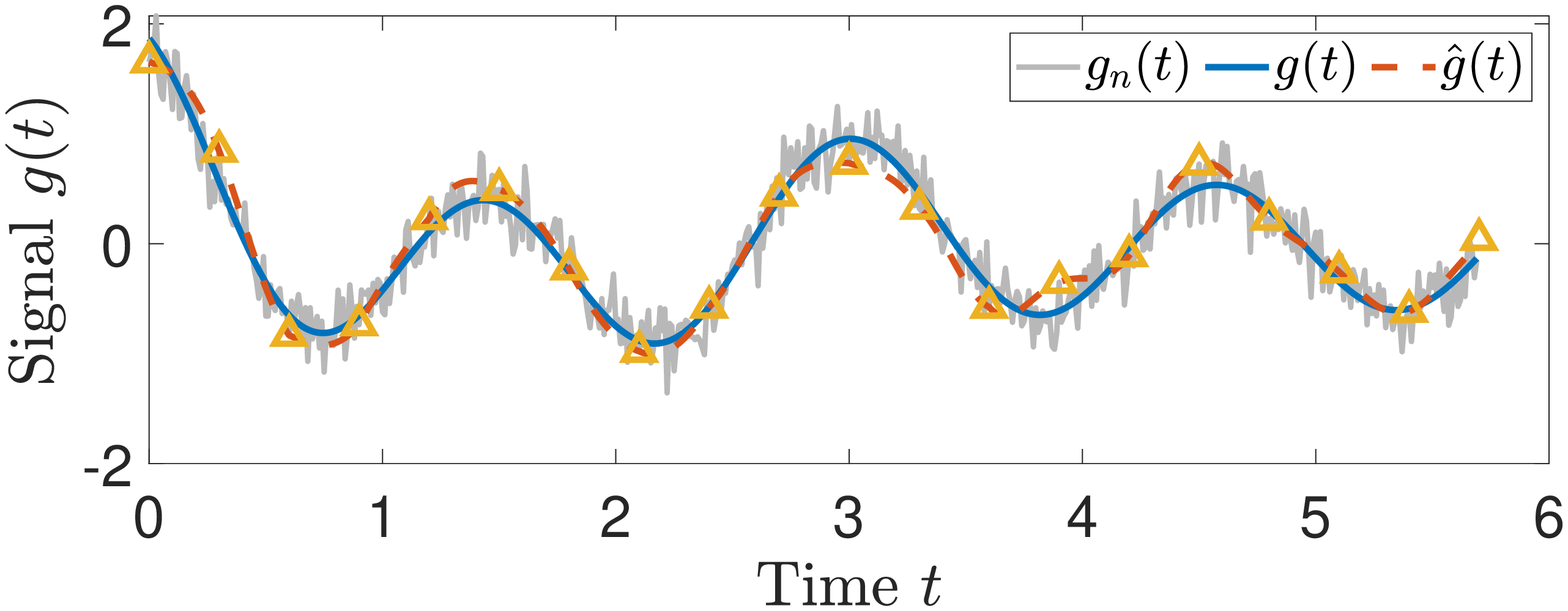}}
	\caption{The reconstruction of the signal with exponential growth by KR method when $ \text{SNR}= 30$ (a), $\text{SNR}=20$ (b),  $\text{SNR}=10$ (c).}
	\label{fig11}
\end{figure}

\begin{figure}[thpb]
	\centering
	%		\DIFdelbeginFL \DIFdelendFL \DIFaddbeginFL
	\subfloat[]{
		\label{Fig12.sub1}
		\includegraphics[width=.42\textwidth]{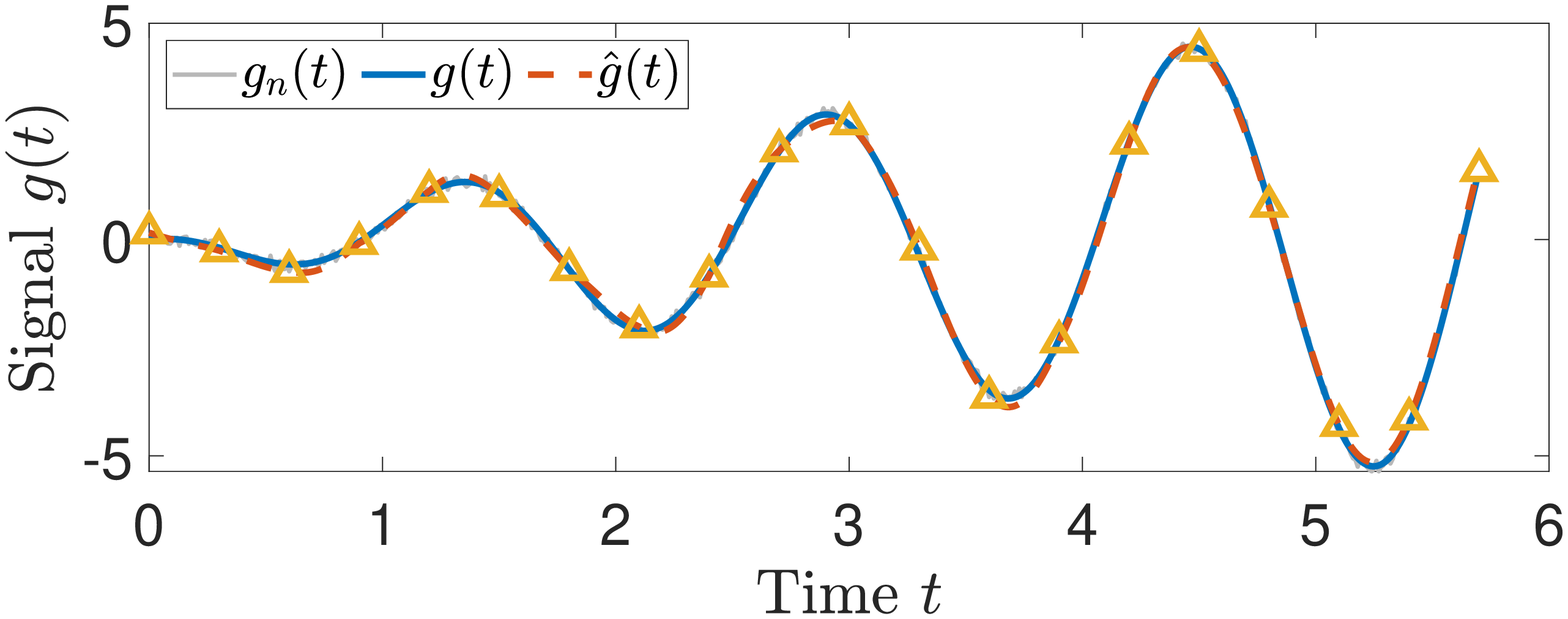}}\\
	\subfloat[]{
		\label{Fig12.sub2}
		\includegraphics[width=.42\textwidth]{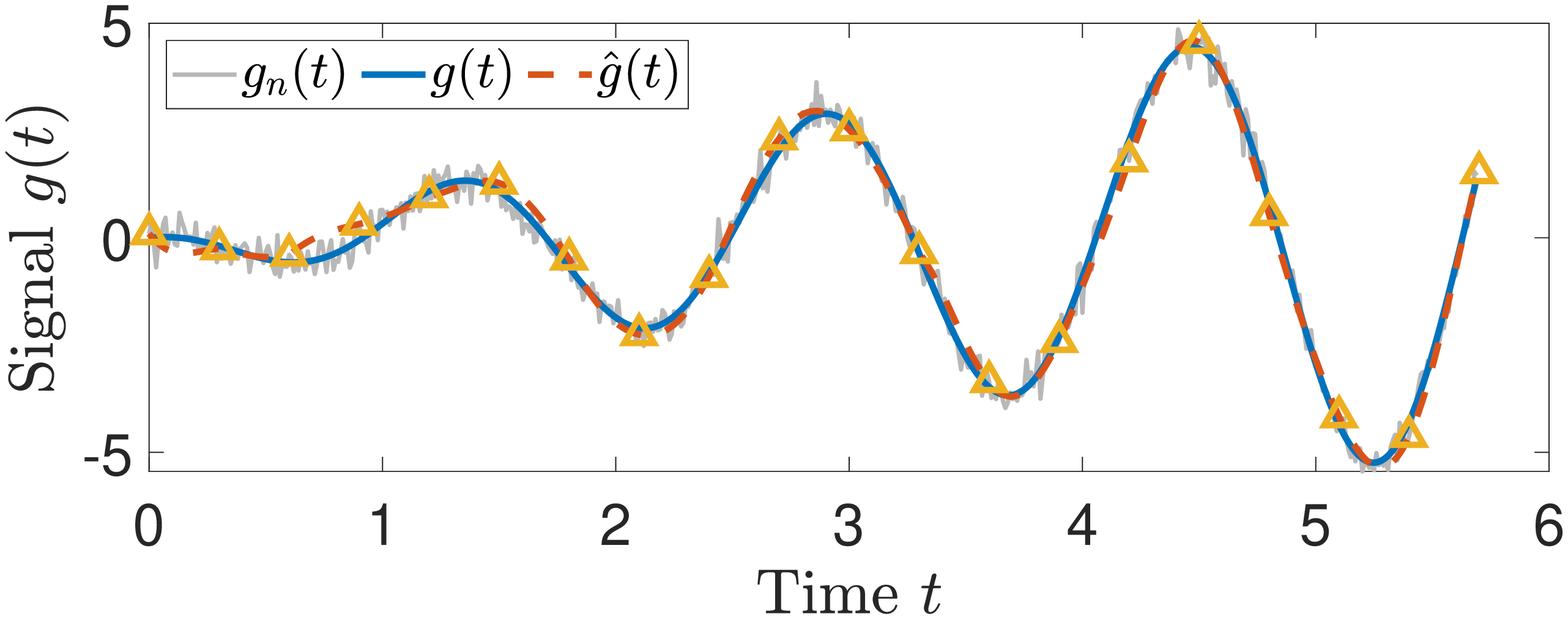}}\\
	\subfloat[]{\label{Fig12.sub3}
		\includegraphics[width=.42\textwidth]{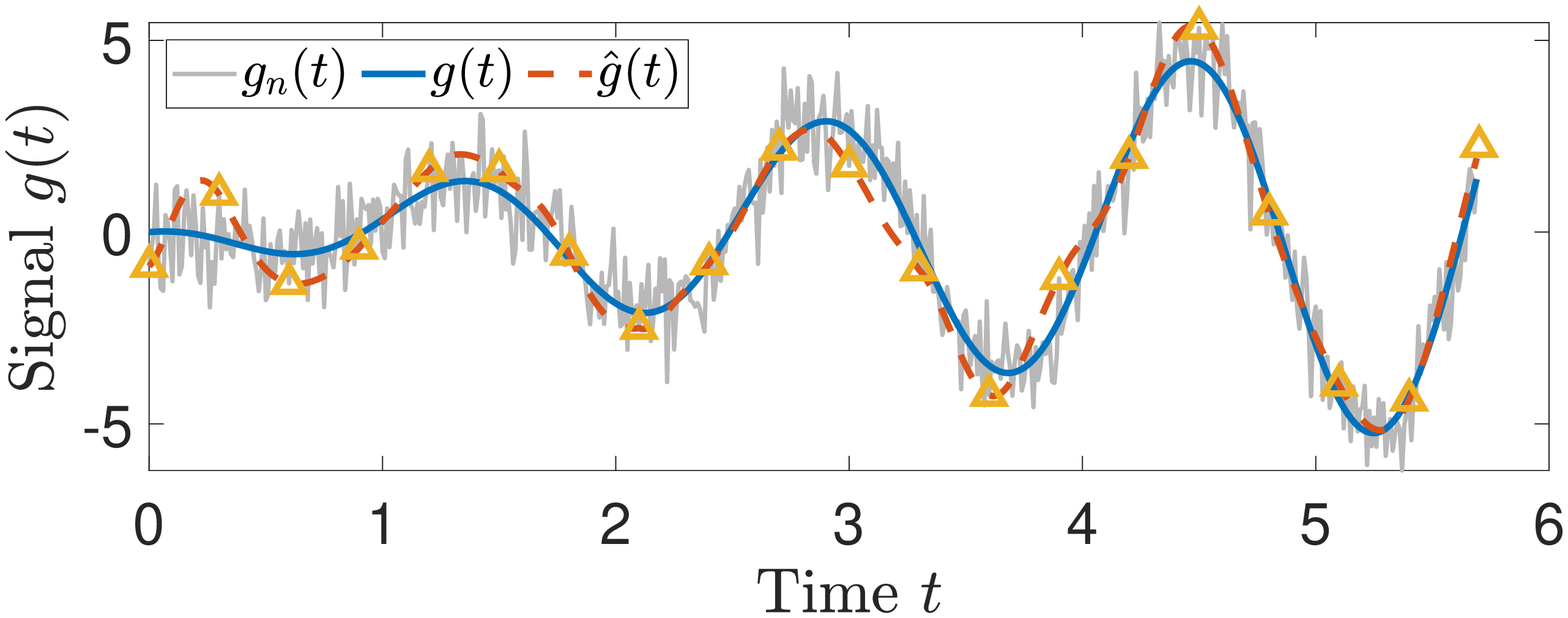}}
	\caption{The reconstruction of the signal with polynomial growth by KR method when $ \text{SNR}= 30$ (a), $\text{SNR}=20$ (b), $\text{SNR}=10$ (c).}
	\label{fig12}
\end{figure}

\begin{figure}[thpb]
	\centering
	%		\DIFdelbeginFL\DIFdelendFL \DIFaddbeginFL
	\subfloat[]{
		\label{Fig13.sub1}
		\includegraphics[width=.42\textwidth]{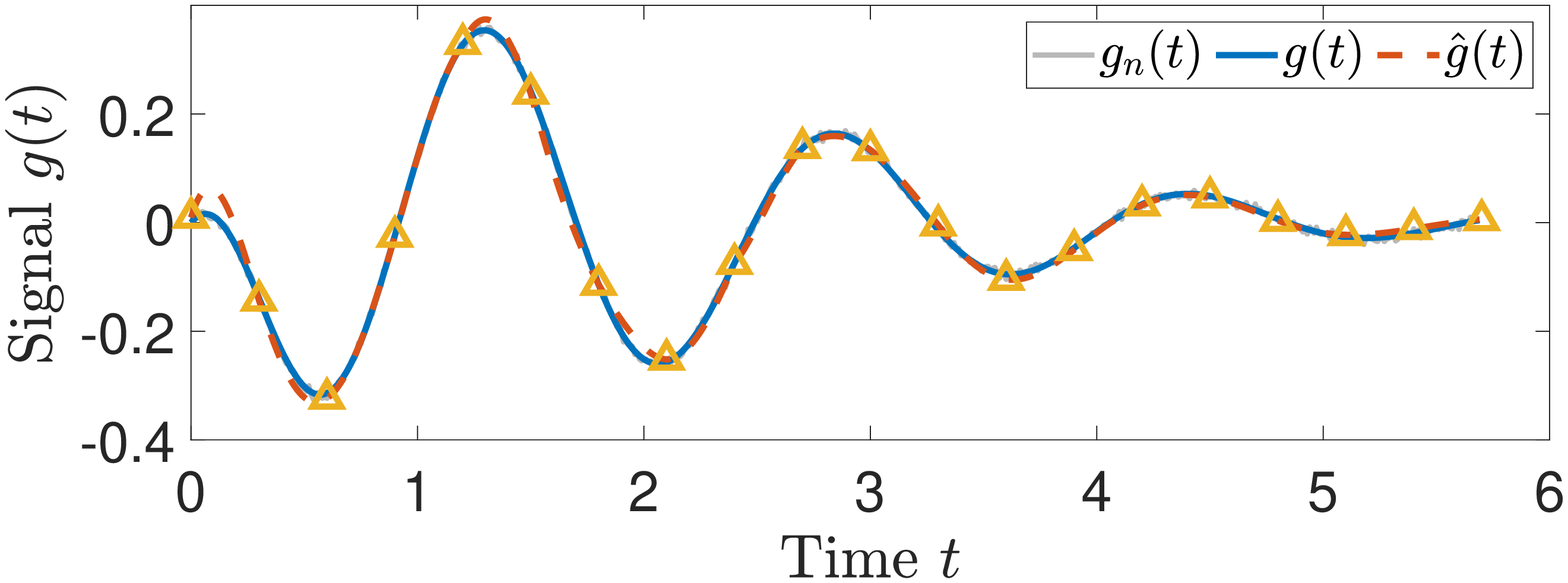}}\\
	\subfloat[]{
		\label{Fig13.sub2}
		\includegraphics[width=.42\textwidth]{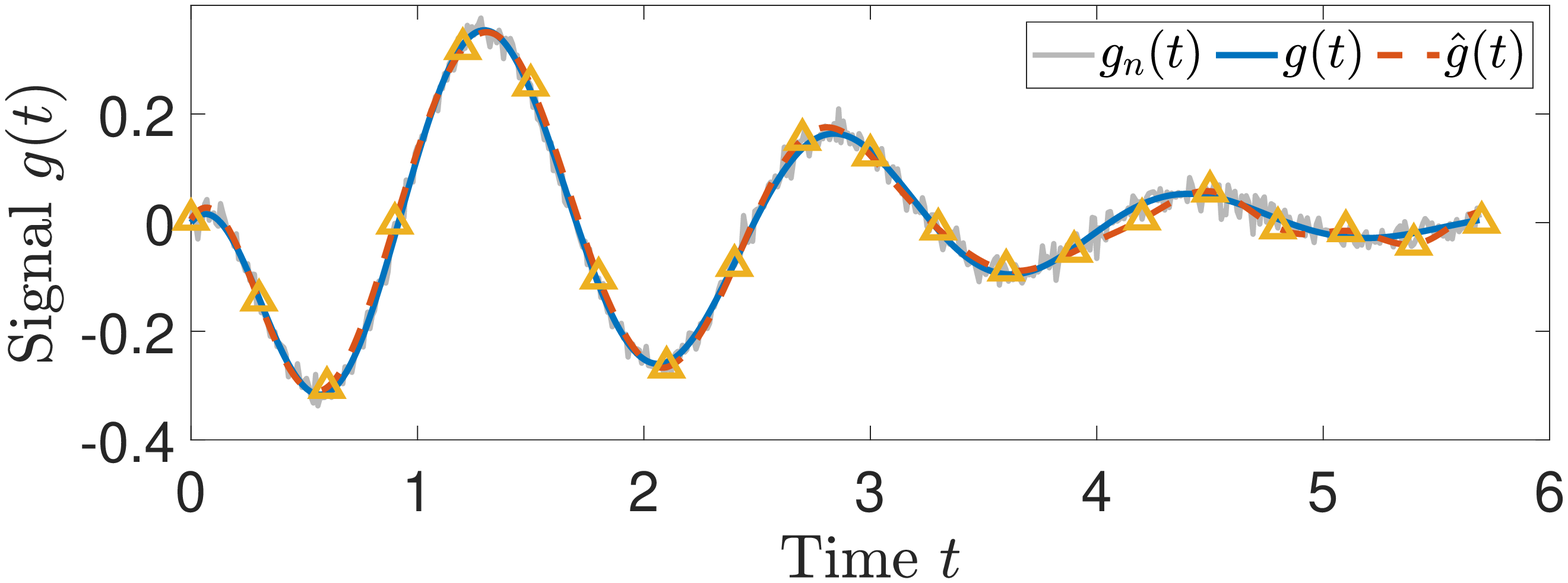}}\\
	\subfloat[]{\label{Fig13.sub3}
		\includegraphics[width=.42\textwidth]{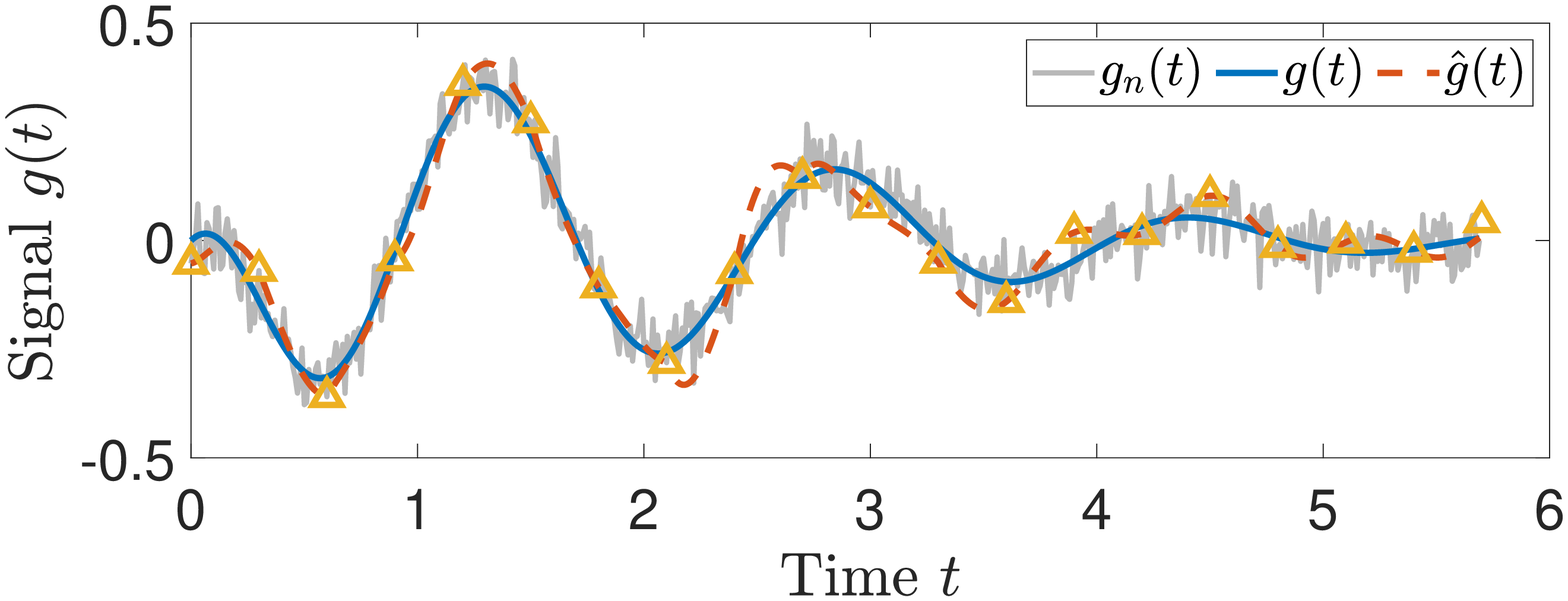}}
	\caption{The reconstruction of the signal with exponential and polynomial growth by KR method %in the presence of noise with
		when  $\text{SNR}= 30$ (a), $\text{SNR}=20$ (b), $\text{SNR}=10$ (c).}
	\label{fig13}
\end{figure}

%\bl{\begin{remark}[The possible improvement of the KR method.]The performance of the KR method depends on Koopman invariance of the functional space $\hF_M$. Therefore, choosing appropriate basis functions is crucial. In this work, we select time-delay functions as basis functions to illustrate the sampling theorem. Actually, there has been a lot of data-driven work that learns the invariant subspace and the representation of the Koopman operator from data \cite{9516947,yeung2019learning,NIPS2017_3a835d32}, which will also be helpful to improve this reconstruction algorithm . \end{remark}}

\section{CONCLUSIONS}\label{sec:con}
Firstly, a generalized sampling theorem is proposed for signals in a generator-bounded space by the Koopman operator theory. This result shows that the sampling bound is determined by the imaginary part of the Koopman spectrum. Through the generalization from the one-dimensional Fourier spectrum to two-dimensional Koopman spectrum, it gives a finite bound of sampling rate for many non-band-limited signals, making the Nyquist rate a special case. Secondly, the reconstruction formula is theoretically investigated, revealing that it can reduce to classical forms for certain signals, such as band-limited signals and Zakai's class of signals. For numerical illustration, the KR method is provided with theoretical convergence, which illustrates the generalized sampling theorem on several signals related to band-limited, exponential, and polynomial functions. Moreover, this method exhibits robustness against low sampling frequency and noise. 

\bl{Several extensions could further advance the application of operator theory in sampling problems, which we consider a promising direction for future work. Firstly, addressing data noise within the Koopman operator framework remains a critical issue. In general, smoothing and filtering are required because interpolating noisy data is dangerous \cite{pawlak2003postfiltering,pawlak2007signal}. There are numerous studies focusing on characterizing and correcting the effects of data noise in Koopman spectral analysis (e.g., \cite{dawson2016characterizing, hemati2017biasing}). Building on these studies, integrating techniques such as Kalman filtering and smoothing into the Koopman operator framework would be beneficial for identifying the Koopman operator from noisy data, thereby enhancing the robustness of signal reconstruction approach. Secondly, the sampling of stochastic signals is of great importance and can be explored through the lens of the stochastic Koopman operator \cite{vcrnjaric2020koopman} or the Frobenius-Perron operator \cite{lasota2013chaos}, which is adjoint to the Koopman operator. Thirdly, irregular sampling and missing values present another problem for research. From the Koopman operator perspective, this could be studied by approximating the infinitesimal generator using the Hille-Yosida theorem \cite[Page 8-13]{pazy2012semigroups}, rather than through the operator logarithm, as numerically implemented in \cite{meng2024koopman}. Considering these aspects, Koopman operator theory has great potential for studying more comprehensive sampling methodology and theory.} 

%, which leverages the expected values of observable functions for random systems. Additionally, the Frobenius-Perron operator, as an adjoint to the Koopman operator, offers an alternative approach by focusing on systems with density functions, thereby providing new insights into stochastic signal sampling. 

%This work focus primarily on the sampling theorem for uniform sampling in the absence of noise, where infinite number of samples are available. 

%\balance
\appendices 
	\bl{\section{Other examples of generator-bounded space}\label{otherexamples}
	
	\begin{lemma}\label{lem:f0}
		Let $f\in L^2([0,b],\Real)$ and its derivative $f^{(1)}\in L^2([0,b],\Real)$. We have 
		$$
		\abs{f(0)}^2\le \frac2b \int_0^{b}\abs{f(t)}^2{\rm d} t + 2b\int_0^{b} \abs{f^{(1)}(t)}^2{\rm d} t.
		$$ 
	\end{lemma}
	\begin{proof}
		By the fundamental theorem of calculus, for $\forall x\in [0,b]$, we can write 
		$$
		f(0) = f(x) - \int_0^x f^{(1)}(t){\rm d} t.
		$$
		It follows that
		$$
		\abs{f(0)} \le \abs{f(x)} + \int_0^{b} \abs{f^{(1)}(t)}{\rm d} t.
		$$
		Integrating both sides over $[0,b]$ gives
		$$
		b\abs{f(0)} \le \int_0^b\abs{f(t)}{\rm d} t + b\int_0^{b} \abs{f^{(1)}(t)}{\rm d} t.
		$$
		Using the Cauchy-Schwarz inequality to the integrals on the right-hand side gives
		\begin{align*}
			b\abs{f(0)} & \le \sqrt{\int_0^{b}\abs{f(t)}^2{\rm d} t}\sqrt{\int_0^{b}1{\rm d} t} + b \sqrt{\int_0^{b} \abs{f^{(1)}(t)}^2{\rm d} t}\sqrt{\int_0^{b}1{\rm d} t}\\
			& = \sqrt{b}\sqrt{\int_0^{b}\abs{f(t)}^2{\rm d} t} + b \sqrt{b}  \sqrt{\int_0^{b} \abs{f^{(1)}(t)}^2{\rm d} t}.
		\end{align*}
		Squaring both sides and using the elementary inequality $(x+y)^2\le 2x^2+2y^2$ conclude the proof.
	\end{proof}
	
	\begin{lemma}\label{lem:g}
		Let $g$ be a (classical) band-limited signal with bandwidth $W$. Then 
		$$
		\abs{g(t)}\le 2W\norm{g}_2^2,\quad \forall t\in\Real,
		$$
		where $\norm{g}_2=\sqrt{\int_{-\infty}^\infty\abs{g(t)}^2{\rm d} t}$.
	\end{lemma}
	
	\begin{proof}
		Write 
		$
		g(t) = \int_{-W}^{W} G(s)e^{2\pi i t s}{\rm d} s, 
		$
		where $G\in L^2[-W,W]$ is the Fourier transform of $g$.  By the Cauchy-Schwarz inequality and Parseval's identity (Lemma \ref{lemma7}), we have
		\begin{align*}
			\abs{g(t)}^2 &\le  \left(\int_{-W}^{W} \abs{G(s)} {\rm d} s\right)^2  \le \int_{-W}^W 1^2 {\rm d} t \cdot \int_{-W}^{W} \abs{G(s)}^2 {\rm d} s \\&= 2W \int_{-W}^{W} \abs{G(s)}^2 {\rm d} s = 2W  \int_{-\infty}^\infty \abs{g(t)}^2 {\rm d} t. 
		\end{align*}
	\end{proof}
	
%	\begin{lemma}\label{lem:b}
%		Let $g$ be a (classical) bandlimited signal with bandwidth $W$. We have
%		$$
%		\int_{-\infty}^\infty \abs{g'(t)}^2 dt  \le 4\pi^2 W^2 \int_{-\infty}^\infty \abs{g(t)}^2 dt.
%		$$
%	\end{lemma}
%	
%	\begin{proof}
%		Write 
%		$
%		g(t) = \int_{-W}^{W} G(s)e^{2\pi i t s}ds, 
%		$
%		where $G$ is the Fourier transform of $g$. We have
%		$$
%		g'(t) = \int_{-W}^{W} 2\pi i s G(s) e^{2\pi i t s}ds.
%		$$
%		It follows that $g'$ is also bandlimited and the Fourier transform of $g'$ is $2\pi i s G(s)$. By Parseval's identity (twice), we have
%		\begin{align*}
%			\int_{-\infty}^\infty \abs{g'(t)}^2 dt &=  \int_{-W}^{W} \abs{2\pi i s G(s) }^2 ds \le 4\pi^2 W^2 \int_{-W}^{W} \abs{G(s)}^2 ds \\ &= 4\pi^2 W^2  \int_{-\infty}^\infty \abs{g(t)}^2 dt. 
%		\end{align*}
%	\end{proof}
	
	\begin{lemma}\label{derivative_bound}
		Let $f(t)=f(0)+tg(t)$, where $g$ is a classical band-limited signal of bandwidth $W$. Then there exists a constant $C_W$, depending only on $W$, such that
		\begin{equation}\label{eq:bound}
			\int_{-\infty}^\infty \frac{\abs{f^{(1)}(t)}^2}{1+t^2} {\rm d} t \le C_W \int_{-\infty}^\infty \frac{\abs{f(t)}^2}{1+t^2}{\rm d} t.         
		\end{equation}
	\end{lemma}
	
	\begin{proof}
		We first show that, without loss of generality, we can assume $f(0)=0$. If not, we can let $\hat f = f - f(0)$. Then $\hat{f}(0)=0$ and $f^{(1)}=\hat f^{(1)}$. Assume that there exists a constant $C_W>1$, which only depends on $W$, such that 
		$$
		\int_{-\infty}^\infty \frac{\abs{\hat f^{(1)}(t)}^2}{1+t^2} {\rm d} t \le C_W \int_{-\infty}^\infty \frac{\abs{\hat f(t)}^2}{1+t^2}{\rm d} t. 
		$$
		Then we have 
		\begin{align*}
			\int_{-\infty}^\infty \frac{\abs{\hat f^{(1)}(t)}^2}{1+t^2} {\rm d} t &\le C_W \int_{-\infty}^\infty \frac{\abs{f(t) -f(0)}^2}{1+t^2}{\rm d} t \\    
			& \le 2 C_W \int_{-\infty}^\infty \frac{\abs{f(t)}^2} {1+t^2}{\rm d} t + 2 C_W\int_{-\infty}^\infty\frac{\abs{f(0)}^2}{1+t^2}{\rm d} t \\
			& = 2 C_W \int_{-\infty}^\infty \frac{\abs{f(t)}^2} {1+t^2}{\rm d} t + 2 C_W \pi \abs{f(0)}^2. 
		\end{align*}
		By Lemma \ref{lem:f0}, we have
		\begin{align*}
			\abs{f(0)}^2 \le \frac2b \int_0^{b}\abs{f(t)}^2{\rm d} t + 2b\int_0^{b} \abs{f^{(1)}(t)}^2{\rm d} t.
		\end{align*}    
		We can choose $b$ such that $4 C_W \pi b = \frac12$, i.e., $b= \frac{1}{8 C_W \pi}$. Recal that $f^{(1)}=\hat f^{(1)}$. It follows that 
		\begin{equation}
		\begin{aligned}
			\int_{-\infty}^\infty \frac{\abs{\hat f^{(1)}(t)}^2}{1+t^2} {\rm d} t \le 2 C_W \int_{-\infty}^\infty \frac{\abs{f(t)}^2} {1+t^2}{\rm d} t \\+ 32 C_W^2 \pi^2  \int_0^{b}\abs{f(t)}^2{\rm d} t + \frac12  \int_0^{b} \abs{\hat f^{(1)}(t)}^2{\rm d} t.  \label{eq:bound2}
		\end{aligned}
		\end{equation}
		Note that, by direct comparison, 
		\begin{align*}
			\int_0^{b}\abs{f(t)}^2{\rm d} t &\le  \left(1+\frac{1}{(8C_W\pi)^2}\right)\int_0^{b}\frac{\abs{f(t)}^2}{1+t^2}{\rm d} t \\&\le \left(1+\frac{1}{(8C_W\pi)^2}\right)\int_{-\infty}^{\infty}\frac{\abs{f(t)}^2}{1+t^2}{\rm d} t,
		\end{align*}
		and
		\begin{align*}
			\int_0^{b}\abs{\hat f^{(1)}(t)}^2{\rm d} t&\le \left(1+\frac{1}{(8C_W\pi)^2}\right)\int_0^{b}\frac{\abs{\hat f^{(1)}(t)}^2}{1+t^2}{\rm d} t \\&\le \left(1+\frac{1}{(8C_W\pi)^2}\right)\int_{-\infty}^{\infty}\frac{\abs{\hat f^{(1)}(t)}^2}{1+t^2}{\rm d} t.
		\end{align*}
		Since $C_W>1$, we have $\frac12 \left(1+\frac{1}{(8C_W\pi)^2}\right)<1$. Substituting these two inequalities into (\ref{eq:bound2}) gives 
		\begin{equation}\label{eq:newbound}
			\int_{-\infty}^\infty \frac{\abs{\hat f^{(1)}(t)}^2}{1+t^2} {\rm d} t \le C_W' \int_{-\infty}^\infty \frac{\abs{f(t)}^2} {1+t^2}{\rm d} t        
		\end{equation}
		for some new constant $C_W'$, which also only depends on $W$. Since $f^{(1)}=\hat f^{(1)}$, this proves (\ref{eq:bound}) in the general case. 
		
		In the following, we assume, without loss of generality, that $f(0)=0$ and derive a constant $C_W$ such that (\ref{eq:bound}) holds. Note that, in this case, we have $f(t)=tg(t)$ for a classical band-limited signal $g$ with bandwidth $W$. 
		
		We have 
		\begin{align}
			\int_{-\infty}^\infty \frac{\abs{f^{(1)}(t)}^2}{1+t^2} {\rm d} t &\le \int_{-\infty}^\infty \frac{\abs{g(t) + tg^{(1)}(t)}^2}{1+t^2} {\rm d} t \notag\\
			& \le 2\int_{-\infty}^\infty \frac{\abs{g(t)}^2}{1+t^2} {\rm d} t  + 2 \int_{-\infty}^\infty \frac{\abs{tg^{(1)}(t)}^2}{1+t^2} {\rm d} t \notag\\
			& \le 2\int_{-\infty}^\infty \abs{g(t)}^2 {\rm d} t + 2\int_{-\infty}^\infty \abs{g^{(1)}(t)}^2 {\rm d} t \notag\\
			& \le (2 + 8\pi^2W^2)  \int_{-\infty}^\infty \abs{g(t)}^2 {\rm d} t, 
			\label{eq:bound0}
		\end{align}
		where in the last inequality is derived by the Parseval's identity (twice), i.e.,
		\begin{align*}
			\int_{-\infty}^\infty \abs{g^{(1)}(t)}^2 {\rm d} t &=  \int_{-W}^{W} \abs{2\pi i s G(s) }^2 {\rm d} s  \\&\le 4\pi^2 W^2 \int_{-W}^{W} \abs{G(s)}^2 {\rm d} s \\&= 4\pi^2 W^2  \int_{-\infty}^\infty \abs{g(t)}^2 {\rm d} t. 
		\end{align*}
		By Lemma \ref{lem:g}, we have 
		\begin{align}
			\int_{-\infty}^\infty\abs{g(t)}^2{\rm d} t  =  \int_{\abs{t}\le\Delta} \abs{g(t)}^2{\rm d} t + \int_{\abs{t}\ge \Delta} \abs{g(t)}^2 {\rm d} t \notag\\
			 \le 4\Delta W\int_{-\infty}^\infty\abs{g(t)}^2 {\rm d} t + \frac{1+\Delta^2}{\Delta^2} \int_{\abs{t}\ge \Delta} \frac{t^2g^2(t)}{1+t^2}{\rm d} t, \notag
		\end{align}   
		which implies 
		\begin{align}
			\int_{-\infty}^{\infty} \abs{g(t)}^2 {\rm d} t  & \le \frac{1 +\Delta^2}{\Delta^2(1-4\Delta W)}\int_{-\infty}^\infty \frac{t^2g^2(t)}{1+t^2}{\rm d} t, \label{eq:bound1}
		\end{align}
		provided that $4\Delta W<1$. Recall that $f(t)=tg(t)$. By (\ref{eq:bound0}) and (\ref{eq:bound1}), we have 
		$$
		\int_{-\infty}^\infty \frac{\abs{f^{(1)}(t)}^2}{1+t^2} {\rm d} t \le C_W \int_{-\infty}^\infty \frac{\abs{f(t)}^2}{1+t^2}{\rm d} t 
		$$
		with 
		$$
		C_W=  \frac{(2 + 8\pi^2W^2)(1 +\Delta^2)}{\Delta^2(1-4\Delta W)}. 
		$$
		Note that $C_W>1$. The proof is complete.  
	\end{proof}
	
	\begin{proposition}\label{generalized_zakai}
		The space of Zakai's signal class \begin{equation}
			\begin{aligned}
				\hF(c,\delta) = \left\{g = g * h: g\in L^2\left((1+t^2)^{-1}\right), \right. \\ \left.\ h(t) = \frac{1}{2\pi} \int_{-\infty}^{\infty}H(\omega)e^{i\omega t}{\rm d}\omega\right\},\label{zk}
			\end{aligned}
		\end{equation} where 
		\begin{equation*}
			H(\omega) = \left\{\begin{matrix}
				&1~ &|\omega|\le c,\\
				&1-\frac{|\omega|-c}{\delta}~& c< |\omega|\le c+\delta,\\
				&0~&|\omega|>c+\delta,
			\end{matrix}
			\right.
		\end{equation*} is a generator-bounded space.
	\end{proposition}
	\begin{proof}
		We first show that, the space given by \eqref{zk} is invariant under the action of the Koopman operator. Specifically, \begin{equation}
			U^\tau g(t) = U^\tau (g*h)(t)  = \int_{-\infty}^{\infty} g(\theta) h(t+\tau-\theta){\rm d}\theta.
		\end{equation}
		Let $v = \theta-\tau$, we have\begin{equation}\label{in1}
			U^\tau g = g(t+\tau) = \int_{-\infty}^{\infty}g(v+\tau)h(t-v) {\rm d}v = (U^\tau g)*h.
		\end{equation} 
		Moreover, we compute \begin{align}
			\|U^\tau g\| &= \int_{-\infty}^{\infty}\frac{|g(t+\tau)|^2}{1+t^2}{\rm d}t \\&=\int_{-\infty}^{\infty}\frac{|g(t+\tau)|^2}{1+(t+\tau)^2}\frac{1+(t+\tau)^2}{1+t^2}{\rm d}t.
		\end{align}
		It can be proved that, for $\forall \tau>0$, there exists $f(\tau)$, where $\frac{\tau^2+\sqrt{\tau^4+4\tau^2}}{2}\le f(\tau)<\infty$, such that \begin{equation}
			\frac{1+(t+\tau)^2}{1+t^2}\le 1+f(\tau)<\infty, \forall t\in\bR.
		\end{equation} Then we have \begin{equation}\label{in2}
			\|U^\tau g\|\le (1+f(\tau))\|g\|<\infty.
		\end{equation} Hence, it follows from \eqref{in1} and \eqref{in2} that $U^\tau g\in\hF(c,\delta)$ for $\forall g\in \hF(c,\delta).$
		
		Then we show that, the generator is bounded on this space, i.e., $\|L|_{\hF_e}\|<\infty$. The signals in \eqref{zk} can be written as $g(t)=g(0)+tg_0(t)$, where $g_0$ is a classical band-limited signal of bandwidth $W=c/2\pi$ \cite{cambanis1976zakai}. Consider the weighted $L^2$ norm $\|g\| = \sqrt{\int_{-\infty}^{\infty}\frac{|g(t)|^2}{1+t^2}{\rm d}t}$.  It follows from Lemma \ref{derivative_bound} that \begin{align*}
			\|L|_{\hF_e}\| &= \sup_{g\in\hF_e}\frac{\|g^{(1)}\|}{\|g\|} = \sup_{g\in\hF_e}\sqrt{\frac{\int_{-\infty}^\infty \frac{\abs{g^{(1)}(t)}^2}{1+t^2} {\rm d} t }{\int_{-\infty}^\infty \frac{\abs{g(t)}^2}{1+t^2}{\rm d} t}}\\ &\le \sqrt{C_W}<\infty,
		\end{align*} where $C_W$ is a constant that only depends on the bandwidth of $g_0(t)$. Hence, the space \eqref{zk} is a generator-bounded space.
	\end{proof}}
	
\hspace{1em}
\section{Proof of Proposition \ref{invers_spectrum}}\label{invers_spectrum_proof}
\begin{proof}
	%Finally we prove that $\sigma(L|_{\hF_e})=\alpha+i[-c,c]$ by demonstrating that the resolvent set $\rho(L|_{\hF_e})$ is $\bC/(\alpha+i[-c,c])$. 
	If $\lambda I - L|_{\hF_e}$ is surjective, there exists $\hat{g}(t) =\frac{1}{2\pi} \int_{-c}^{c}\hat{G}(\omega)e^{(\alpha+i\omega) t}{\rm d}\omega\in\hF_e$ for $$\forall g(t) = \frac{1}{2\pi}\int_{-c}^{c}G(\omega)e^{(\alpha+i\omega) t}{\rm d}\omega\in \hF_e$$ such that $(\lambda I - L|_{\hF_e})\hat{g} = g.$ It follows that $$\hat{G}(\omega) = G(\omega)/(\lambda-\alpha-i\omega)$$ for $\forall G\in L^2[-c,c]$. Then we have $\lambda \in \bC/(\alpha+i[-c,c])$. %Otherwise the function $\hat{h}$ is meaningless when $i\omega=lambda-\alpha$.
	Similarly, the sufficiency can be proved, i.e., $(\lambda I - L|_{\hF_e})$ is surjective if $\lambda\in  \bC/(\alpha+i[-c,c])$. Thus, the resolvent set $\rho(L|_{\hF_e})\subseteq\bC/(\alpha+i[-c,c])$. Then we show that $(\lambda I - L|_{\hF_e})$ is also injective for $\forall \lambda\in  \bC/(\alpha+i[-c,c])$, i.e., \begin{align}
		(\lambda I - L|_{\hF_e})g  %=\bl{\frac{1}{2\pi}}\int_{-c}^c(\lambda-\alpha-i\omega)\bl{G(\omega)}e^{(\alpha+i\omega) t}{\rm d}\omega 
		= 0 %\label{eq3}\\ 
		\Rightarrow g(t) = %\bl{\frac{1}{2\pi}\int_{-c}^c G(\omega)}e^{(\alpha+i\omega) t}{\rm d}\omega= 
		0,~ \forall t>0.\label{eqg}
	\end{align} It follows from the left part of \eqref{eqg} that $\dot{g}(t) = \lambda g(t)$ and $g(t) = e^{\lambda t}g(0)$, i.e., \begin{equation*}
		\frac{1}{2\pi}\int_{-c}^{c}e^{(\alpha+i\omega) t} G(\omega){\rm d}\omega = e^{\lambda t}g(0), \forall t>0.
	\end{equation*} 
	For $\forall \lambda\in \bC/(\alpha+i[-c,c])$, this equation holds only when $g(t) = 0$, which is consistent with \eqref{eqg}. %Because jifenzhongzhidingli Mean value theorem for integrals
	Then $(\lambda I - L|_{\hF_e})$ is bijective for $\forall \lambda\in \bC/(\alpha+i[-c,c])$. By the definition \cite[Def 1.1]{engel2000one} of resolvent set $\rho(L|_{\hF_e})$, we have $\rho(\cdot) = \bC/(\alpha+i[-c,c])$ and the spectrum is $\sigma(L|_{\hF_e})=\alpha+i[-c,c]$.
\end{proof}
\bl{\section{Reconstruction formula of Zakai's class}\label{zakai}
%	Consider the functional space
%	\begin{equation}\label{zk}
%		\hF = \left\{g\in L^2((1+t^2)^{-1}): g = g * h\right\},
%	\end{equation} with the weighted $L^2$ norm $\|g\| = \left(\int_{-\infty}^{\infty}\frac{|g(t)|^2}{1+t^2}{\rm d}t\right)^{1/2}$, where $h(t) = \frac{1}{2\pi} \int_{-\infty}^{\infty}H(\omega)e^{i\omega t}{\rm d}\omega$, and 
%	\begin{equation*}
%		H(\omega) = \left\{\begin{matrix}
%			&1~ &|\omega|\le c,\\
%			&1-\frac{|\omega|-c}{\delta}~& c< |\omega|\le c+\delta,\\
%			&0~&|\omega|>c+\delta.
%		\end{matrix}
%		\right.
%	\end{equation*}
	Consider the signal $g(t) =(g(0)+ t g_0(t))e^{\alpha t},$ where $g_0(t) = \frac{1}{2\pi} \int_{-c}^{c} G(\omega) e^{i\omega t} {\rm d}\omega, G\in L^2[-c,c], g(0)\in\bR,\alpha\in\bR.$ The reconstruction formula for this signal class can be derived in a manner similar to the proof of Proposition \ref{invers_recons}. In particular, we have \begin{equation*}
		U^\tau g (t) = \left(g(0)+tg_0(t+\tau)+ \tau g_0(\tau)\right)e^{\alpha(t+\tau)}.
	\end{equation*}
	For the third term, it follows from \eqref{eitau} that \begin{align*}
		&\tau e^{\alpha(t+\tau)}g_0(t+\tau) = \\
		&\tau e^{\alpha(t+\tau)}\sum_{k=-\infty}^{\infty}\frac{\sin(\pi/T_s)(\tau-kT_s)}{\pi/T_s(\tau-kT_s)}\int_{-c}^{c} \frac{G(\omega)}{2\pi}e^{i\omega (t+kT_s)}{\rm d}\omega
	\end{align*} By letting $t = 0$, we have \begin{align*}
		g(\tau) =& e^{\alpha \tau}\left(g(0)+\tau g(\tau)\right) \\
		=& \tau e^{\alpha \tau}\sum_{\substack{k=-\infty \\ k \neq 0}}^{\infty}\frac{\sin(\pi/T_s(\tau-kT_s))}{k\pi(\tau-kT_s)}\left(g(kT_s)e^{-\alpha kT_s} - g(0)\right)\\+&e^{\alpha \tau}g(0)+e^{\alpha \tau}\frac{\sin(\pi\tau/T_s)}{\pi/T_s}g_0(0).
	\end{align*}
	Since $g_0(0) = g^{(1)}(0)-\alpha g(0)$, and \begin{equation*}
		\frac{\tau}{k\pi(\tau-kT_s)} = \frac{1}{\pi/T_s(\tau-kT_s)}+\frac{1}{k\pi},
	\end{equation*} we have \begin{align*}
	g(\tau) &= \sum_{\substack{k=-\infty \\ k \neq 0}}^{\infty}e^{\alpha \tau}\left(\frac{\sin(\pi/T_s(\tau-kT_s))}{\pi/T_s(\tau-kT_s)}+\frac{(-1)^k\sin\pi\tau/T_s}{k\pi}\right)\\ &\times \left(g(kT_s)e^{-\alpha kT_s} - g(0)\right) + g(0)e^{\alpha \tau}\\ &+\frac{\sin (\pi\tau/T_s)}{\pi/T_s}(g^{(1)}(0)-\alpha g(0))e^{\alpha \tau}.
	\end{align*}
	Based on \cite[Lemma 3]{zakai1965band}, we obtain that \begin{equation*}
		g^{(1)}(0)-\alpha g(0)+\sum_{\substack{k=-\infty \\ k \neq 0}}^{\infty}\frac{g(kT_s)e^{-\alpha kT_s}-g(0)}{kT_s}(-1)^k =0.
	\end{equation*}
	It follows that \begin{align*}
		g(\tau) &= \sum_{\substack{k=-\infty \\ k \neq 0}}^{\infty}\frac{\sin(\pi/T_s(\tau-kT_s))}{\pi/T_s(\tau-kT_s)}\left(g(kT_s)e^{-\alpha kT_s}-g(0)\right)e^{\alpha\tau}\\&+g(0)e^{\alpha \tau}.
	\end{align*} 
	Based on the reconstruction formula for $f(t) = g(0)$, we have $g(0) = \sum_{k=-\infty}^{\infty}g(0)\frac{\sin(\pi/T_s(\tau-kT_s))}{\pi/T_s(\tau-kT_s)}$. Hence, we obtain the reconstruction formula of $g(\tau)$, i.e., 
	%Then the reconstruction formula can be obtained using the idea for proving \cite[Theorem 3]{zakai1965band}, i.e.,
	 \begin{equation}
		g(\tau) = \sum_{k=-\infty}^{\infty}g(kT_s)\frac{\sin(\pi/T_s(\tau-kT_s))}{\pi/T_s(\tau-kT_s)}e^{\alpha(\tau-kT_s)}.
	\end{equation}
}

\section{Proof of Proposition \ref{poly_exp_spectrum}}
\begin{proof}\label{poly_exp_spectrum_proof}
	Consider the basis function of \eqref{poly_exp_signal} $$\Phi = [e^{\lambda_1 t},\ldots, t^{b_1}e^{\lambda_1 t},e^{\lambda_2 t},\ldots,t^{b_m}e^{\lambda_m t}].$$ The generator $L|_{\hF_e}$ can be represented by $\overline{L}|_{\hF_e}$, i.e., $L|_{\hF_e}\Phi= \Phi$, where $\overline{L}|_{\hF_e} = \mathrm{diag}\{L_1,\ldots,L_m\}$, %\begin{equation*}\overline{L}|_{\hF_e} =  \left(\begin{matrix}L_1 &&\\&\ddots&\\&&L_m\end{matrix}\right),\end{equation*} 
	and for $\forall k=1,\ldots,m,$%\ye{double check the form of $L_k$, does not seem to be correct to me}
	\begin{equation*}
		L_k = \left(\begin{matrix}
			\lambda_k &1&&\\
			&\ddots&\ddots&\\
			&&\lambda_k&b_k\\
			&&&\lambda_k
		\end{matrix}\right)_{(b_k+1)\times (b_k+1)}.
	\end{equation*} Hence the spectrum of the generator $L|_{\hF_e}$ can be obtained by the eigenvalues of $\overline{L}|_{\hF_e}$, i.e., $\sigma(L|_{\hF_e}) = \{\lambda_k\}_{k=1}^m$.
\end{proof}
\section{Proof of Lemma \ref{exponential}}\label{sec:exp}
\begin{proof}
	Since the operators $A, A_n$ are bounded, the exponential can be expanded by the Taylor series, i.e.,
	\begin{equation*}
		\begin{aligned}
			\|(\exp A-\exp A_n )g\|
			=\left\|\sum_{k=0}^{\infty}\frac{A^k-A_n^k}{k!}g\right\|, ~\forall g\in X.
		\end{aligned}
	\end{equation*}
	By the convergence of $\exp(A)$ and $\exp(A_n)$, for any $\epsilon>0$, there exists $q>0,$ such that
	\begin{equation*}
		\begin{aligned}
			\left\|\sum_{k=q}^{\infty}\frac{A^k}{k!}g \right\|<\frac{\epsilon}{3}, ~
			\left\|\sum_{k=q}^{\infty}\frac{A_n^k}{k!}g \right\|<\frac{\epsilon}{3}.
		\end{aligned}
	\end{equation*}
	Here we show that $\lim_{n\to\infty}\|(A^k-A_n^k)g\|=0$ for $k<q$ by mathematical induction. Consider $q=1$, we have $\lim_{n\to\infty}\|(A-A_n)g\|=0$. Assume that $\lim_{n\to\infty}\|(A^k-A_n^k)g\|=0$ for $k=p$. Then for $k=p+1$, we have
	\begin{equation*}
		\begin{aligned}
			&\lim_{n\to\infty}\|(A^{p+1}-A_n^{p+1})g\|\\
			\le&\lim_{n\to\infty}\|(AA^{p}-AA_n^p)g\|+\lim_{n\to\infty}\|(AA_n^p-A_nA_n^p)g\|\\
			\le&\lim_{n\to\infty}\|A\|\|(A^{p}-A_n^p)g\|+\lim_{n\to\infty}\|(A-A_n)A_n^pg\|=0.
		\end{aligned}
	\end{equation*}
	It follows that,
	$%\begin{equation*}
	\lim_{n\to\infty}\sum_{k=0}^{q-1}\|(A^{k}-A_n^{k})g\|=0,\forall g\in X.
	$%\end{equation*}
	Then for any $\epsilon>0,$ there exists $n>0$ such that \begin{equation*}
		\lim_{n\to\infty}\sum_{k=0}^{q-1}\|(A^{k}-A_n^{k})g\|<\epsilon/3,
	\end{equation*}
	Therefore, for any $\epsilon>0,g\in X$, we have 
	\begin{equation*}
		\begin{aligned}
			&\|(\exp A-\exp A_n )g\|\\
			\le&\sum_{k=0}^{q-1} \left\|\frac{A^k-A_n^k}{k!}g\right\|+ \left\|\sum_{k=q}^{\infty}\frac{A^k}{k!}g\right\|+\left\|\sum_{k=q}^{\infty}\frac{A_n^k}{k!}g\right\|<\epsilon.
		\end{aligned}
	\end{equation*}
\end{proof}

\section{Proof of Lemma \ref{logarithm}}\label{sec:log}
\begin{proof}
	Let $\|A\|<K$, $\|A_n\|<K$ with $K>0$ and $\Omega = \{\lambda\in\bC:|\lambda|< K \} $. Consider the resolvent operator $R(\lambda,A) = (\lambda I-A)^{-1}$. Since $|\lambda|>\|A\|$ and $|\lambda|> \|A_n\|$ for $\lambda\in+\partial\Omega$, the resolvent operator is expanded by \cite[P. 584]{dunford1988linear}
	\begin{equation*}
		R(\lambda,A) = \sum_{k=0}^{\infty}\frac{A^k}{\lambda^{k+1}},~ R(\lambda,A_n) = \sum_{k=0}^{\infty}\frac{A_n^k}{\lambda^{k+1}}.
	\end{equation*}
	By the convergence of these series and $\lim_{n\to\infty}\|(A-A_n)g\| = 0$, we can prove that \begin{equation*}
		\lim_{n\to\infty}\|(R(\lambda,A)-R(\lambda,A_n))g\| = \lim_{n\to\infty}\|\sum_{k=0}^{\infty}\frac{A^k-A_n^k}{\lambda^{k+1}}g\| = 0,
	\end{equation*} which is similar to the proof of Lemma \ref{exponential}. So, for $\forall \epsilon>0,$ there exists $n>0$ such that $%\begin{equation*}
	\|(R(\lambda,A)-R(\lambda,A_n))g\|\le \frac{\epsilon}{Kr},
	$%\end{equation*}
	where $r = \max_{\lambda\in +\partial \Omega} |\log\lambda|.$
	It follows that
	\begin{equation*}
		\begin{aligned}
			&\|\log A-\log A_n )g\|\\=&\frac{1}{2\pi}\left\|\int_{+\partial \Omega}\log\lambda (R(\lambda,A) - R(\lambda,A_n))g {\rm d}\lambda\right\|\\
			\le &\frac{1}{2\pi}\int_{+\partial \Omega}|\log\lambda| \left\|(R(\lambda,A) - R(\lambda,A_n))g\right\| {\rm d}\lambda\\
			< & \frac{r}{2\pi} \int_{+\partial \Omega} \frac{\epsilon}{Kr} {\rm d}\lambda = \epsilon.
		\end{aligned}
	\end{equation*}
	Therefore, we have $\lim_{n\to\infty}\|(\log A-\log A_n )g\|=0.$
\end{proof}

\bibliographystyle{IEEEtran}
%\bibliography{IEEEabrv,cite}

%\balance

%\newpage
%\balance

%\vspace{-10pt}

\begin{IEEEbiography}[{\includegraphics[width=1in,height=1.25in,clip,keepaspectratio]{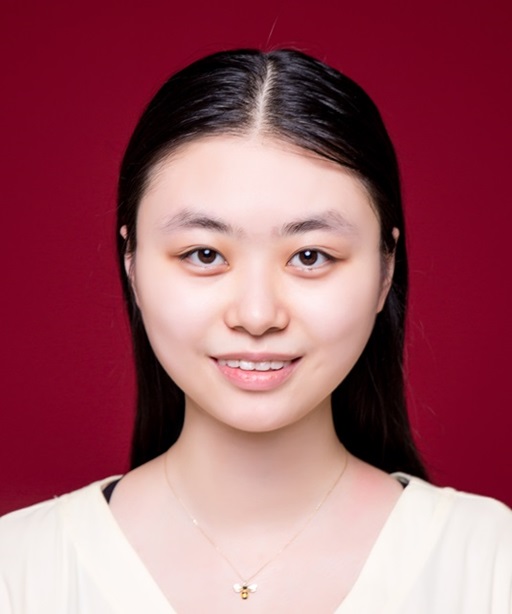}}]
	{Zhexuan Zeng} is currently working toward the Ph.D. degree in control science and engineering with the Department of Automatic Control, Huazhong University of Science and Technology, Wuhan, China, under the supervision of Prof. Ye Yuan. 
	
	Her research interests include system identification, sampling theorem, and applications of operator theoretic methods.
\end{IEEEbiography}

\begin{IEEEbiography}[{\includegraphics[width=1in,height=1.25in,clip,keepaspectratio]{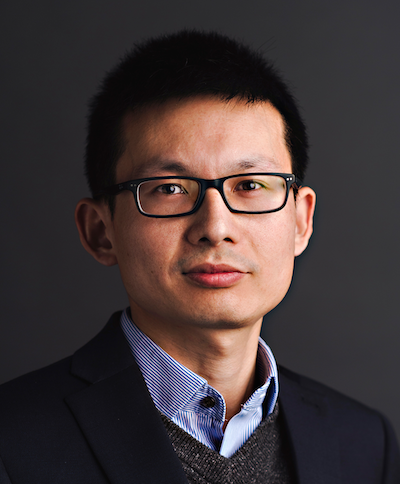}}]
	{Jun Liu} (Senior Member, IEEE) received the B.S. degree in Applied Mathematics from Shanghai Jiao-Tong
	University in 2002, the M.S. degree in Mathematics
	from Peking University in 2005, and the Ph.D. degree in Applied Mathematics from the University of
	Waterloo in 2011. 
	
	He held an NSERC Postdoctoral Fellowship in Control and Dynamical Systems at Caltech between 2011 and 2012. He was a Lecturer in Control and Systems Engineering at the University of Sheffield between 2012 and 2015. In 2015, he joined the Faculty of Mathematics at the University of Waterloo, where he currently is a Professor of
	Applied Mathematics and directs the Hybrid Systems Laboratory.  His main research interests are in the theory and applications of hybrid systems and control, including rigorous computational methods for control design with applications in cyber-physical systems and robotics. He was awarded a Marie-Curie Career Integration Grant from the European Commission in 2013, a Canada Research Chair from the Government of Canada in 2017, an Early Researcher Award from the Ontario Ministry of Research, Innovation and Science in 2018, and an Early Career Award from the Canadian Applied and Industrial Mathematics Society and Pacific Institute for the Mathematical Sciences (CAIMS/PIMS) in 2020. His best paper awards include the Zhang Si-Ying Outstanding Youth Paper Award (2010, 2015) and the Nonlinear Analysis: Hybrid Systems Paper Prize (2017). 
\end{IEEEbiography}

%\vspace{-5pt}

\begin{IEEEbiography}[{\includegraphics[width=1in,height=1.25in,clip,keepaspectratio]{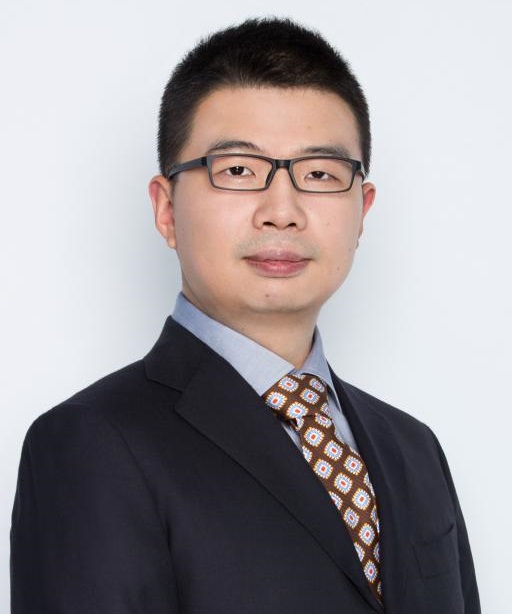}}]
	{Ye Yuan} (Senior Member, IEEE) received the B.Eng. degree (Valedictorian) in Automation from Shanghai Jiao Tong University, Shanghai, China, in 2008, and the 
	M.Phil. and Ph.D. degrees in control theory from 
	the Department of Engineering, University of 
	Cambridge, Cambridge, U.K., in 2009 and 2012,
	respectively.  
	
	He was a Postdoctoral Researcher with UC Berkeley, and a Junior Research Fellow with	Darwin College, University of Cambridge. He is currently a Full Professor with the Huazhong University of Science and Technology, Wuhan, China. His research interests include system identification and control with applications to cyber-physical systems.
\end{IEEEbiography}
\balance

%\vfill

\end{document}